\documentclass[%
 reprint,
 nofootinbib,
 amsmath,amssymb,
 aps,
floatfix,
]{revtex4-2}

\usepackage{dcolumn}
\usepackage{bm}

\usepackage{mathtools}
\usepackage{graphicx}
\usepackage{amsthm}
\usepackage{amsmath}
\usepackage{tikz}
\usepackage{tcolorbox}
\usepackage{float}
\usepackage{placeins}
\usepackage{booktabs}
\usepackage{hyperref}
\usetikzlibrary{shapes.geometric, arrows.meta, positioning, calc, 3d, decorations.markings, decorations.pathreplacing}

\definecolor{navyblue}{RGB}{30,58,95}
\definecolor{tealblue}{RGB}{42,123,155}
\definecolor{coral}{RGB}{211,95,95}
\definecolor{gold}{RGB}{218,165,32}

\hypersetup{colorlinks=true,linkcolor=blue,citecolor=blue,urlcolor=blue}

\newtheorem{theorem}{Theorem}
\newtheorem{lemma}[theorem]{Lemma}
\newtheorem{corollary}[theorem]{Corollary}
\newtheorem{definition}[theorem]{Definition}
\newtheorem{axiom}{Axiom}
\newtheorem{remark}[theorem]{Remark}

\newcommand{\C}{\mathbb{C}}
\newcommand{\R}{\mathbb{R}}
\newcommand{\Z}{\mathbb{Z}}

\newcommand{\Hilb}{\mathcal{H}}
\newcommand{\ket}[1]{|#1\rangle}

\newcommand{\braket}[2]{\langle#1|#2\rangle}
\newcommand{\ketbra}[2]{|#1\rangle\langle#2|}
\newcommand{\Tr}{\mathrm{Tr}}
\newcommand{\id}{\mathrm{id}}

\newcommand{\KK}{\mathcal{K}}

\begin{document}

\title{Existence as Distinguishability: Quantum Mechanics from Finite Graded Equality}
\author{Julian G.\ Zilly}
\email{research@julianzilly.com}
\affiliation{Independent Researcher \\ ORCID: \href{https://orcid.org/0009-0005-6257-0214}{0009-0005-6257-0214}}


\begin{abstract}
We derive finite-dimensional quantum mechanics from a single ontological principle, that \emph{existence is constituted by distinguishability}, together with two structural commitments: finite capacity $N$ (parametric input) and self-referential consistency (SRC, a closure schema with two equivalent forms, operational and information-theoretic). SRC unpacks into eight derived structural conditions (Theorem~\ref{thm:src-master}); structural unambiguity (S5) (Definition~\ref{def:s5}) completes the hierarchy, uniquely selecting the Born rule as the geometric/probabilistic closure. The graded distinguishability kernel $K(x,y) \in [0,1]$ realises both axioms, with a state constituted by its $K$-profile against all others. For each $N \geq 3$, the unique distinguishability space is $(\C P^{N-1}, K)$ with $K(\psi,\phi) = 1 - |\langle\psi|\phi\rangle|^2$, from which complex coefficients, the Born rule $p_k = |c_k|^2$, unitary dynamics, and tensor-product composition all follow. Indeterminism is forced by capacity overflow; alternatives (e.g.\ Bohmian mechanics) are classified rather than refuted. Standard QM is the $N \to \infty$ limit; finite $N$ is the only free parameter. The algebraic spine is machine-checked in Lean~4 modulo five imported classical theorems and the existence direction of Stone's theorem; Appendix~\ref{app:formal-verification} states the verification scope.
\end{abstract}

\maketitle

\section{Introduction}
\label{sec:intro}

Physics is an act of world-building. We take \emph{existence is constituted by distinguishability} as the single ontological commitment and building block: a state \emph{is} the totality of its distinctions against every other state. The world-model must be finite (finite capacity $N$ bounds the mutually distinct states) and self-referentially consistent (SRC: no canonical structure or scaffolding exists outside the kernel $K$). The graded distinguishability kernel $K(x, y) \in [0, 1]$ realizes both, and finite-dimensional quantum mechanics follows for $N \geq 3$.

Since a state \emph{is} its $K$-profile (the complete pattern of $K$-values against all other states), the Leibniz principle becomes definitional: states with identical $K$-profiles are identical, and every consistent $K$-profile is realized, for a consistent pattern of distinctions \emph{is} a state (\S\ref{sec:graded-equality}). These are the only axioms (\S\ref{sec:axioms}). The alternatives (infinite precision, missing states, accidental symmetries, privileged decompositions) each introduce structure beyond $K$ and so violate at least one of finiteness or SRC (\S\ref{sec:discussion}). Four regularity conditions (\S\ref{sec:consequences}) are each the unique choice consistent with the capacity bound (Appendix~\ref{app:structural-conditions}).

The central insight is information-theoretic. A finite-capacity system cannot express its full relational identity within a single measurement context: one basis captures $\log_2 N$ bits, while the full $K$-profile requires super-logarithmically many bits in general and $\Theta(N\log_2 N)$ for prime-power $N$ at the maximal MUB count (Theorem~\ref{thm:capacity-halting}). The missing information is accessible only through complementary bases, and Structural Leibniz within Saturation forces these complementary evaluations to be connected by a unique cyclic automorphism of order $N$ (Theorem~\ref{thm:dynamics-derived}). This sequential access \emph{is} time: time is the unfolding of relational identity that cannot fit in a single context. Given Saturation, the bit-count gap forces randomness, since no ontic reservoir beyond $K$ exists to store the missing bits. Theories rejecting Saturation commit to a different ontology (\S\ref{sec:discussion}).

From here, the remaining structure is uniquely determined for $N \geq 3$ (\S\ref{sec:complex-structure}--\ref{sec:measurement}): complex coefficients $\C$ are forced as the unique field supporting cyclic dynamics and relational isotropy; the state space is $\mathcal{X} \cong \C P^{N-1}$; the Born rule $p_k = |c_k|^2$ is the unique probability assignment preserving statistical distinguishability under reversible dynamics; and tensor-product composition follows from kernel associativity. A quantum sampling theorem shows that $\C P^{N-1}$ introduces no continuous degrees of freedom beyond those operationally accessible from two complementary bases. Standard quantum mechanics is the $N \to \infty$ limit; prior reconstructions \cite{hardy2001quantum,chiribella2011informational,masanes2011derivation} treat finite dimension as an approximation, while this paper treats finite capacity as fundamental.

\noindent\textbf{Reader's guide.} \S\ref{sec:graded-equality} unpacks the two extensions into the axioms; \S\ref{sec:axioms} states axioms formally; \S\ref{sec:consequences}--\ref{sec:complex-structure} derive algebraic and projective structure; \S\ref{sec:hilbert} assembles the Hilbert space (cyclic symmetry, gauge structure, complex coefficients, inner product); \S\ref{sec:geometry} derives the Born rule; \S\ref{sec:evolution}--\ref{sec:measurement} develop dynamics and the Capacity Halting Principle; \S\ref{sec:composite} treats composite systems and no-cloning; \S\ref{sec:related} compares prior work; \S\ref{sec:discussion} discusses scope; Appendix~\ref{app:formal-verification} maps each derivation to its Lean~4 mechanization, distinguishing machine-checked results from prose-only physical reasoning.

\begin{figure}[t]
\centering
\begin{tikzpicture}[scale=0.85]
\draw[->, thick] (0,0) -- (7,0) node[right] {$N$ (capacity)};
\draw[->, thick] (0,0) -- (0,5) node[above] {bits};

\draw[navyblue, very thick, domain=1:6.5, samples=100]
  plot (\x, {ln(\x)/ln(2)}) node[right, font=\small] {$\log_2 N$};
\node[navyblue, font=\small, align=left] at (3.6, 0.9) {Available\\capacity};

\draw[coral, very thick, domain=1:2.7, samples=100]
  plot (\x, {\x*ln(\x)/ln(2)});
\node[coral, font=\small, align=center] at (2.7, 4.5) {Required:\\$\Omega(N\log_2 N)$};

\fill[coral!15] (1.5,{ln(1.5)/ln(2)}) -- (1.5,{1.5*ln(1.5)/ln(2)}) --
  plot[domain=1.5:2.5] (\x, {\x*ln(\x)/ln(2)}) -- (2.5,{ln(2.5)/ln(2)}) --
  plot[domain=2.5:1.5] (\x, {ln(\x)/ln(2)}) -- cycle;

\draw[<->, thick, gray] (2.2, {ln(2.2)/ln(2)+0.1}) -- (2.2, {2.2*ln(2.2)/ln(2)-0.1});
\node[gray, font=\footnotesize, align=center] at (3.6, 2.7) {Capacity\\Deficit};
\draw[gray, thin, dotted] (3.1, 2.7) -- (2.4, {(ln(2.2)/ln(2) + 2.2*ln(2.2)/ln(2))/2});

\fill (1,0) circle (2pt) node[below, font=\small] {$1$};

\node[font=\scriptsize, gray] at (6.3, -0.5) {Standard QM ($N \to \infty$)};

\node[font=\small\bfseries, align=center] at (7, 4.7) {$\log_2 N \ll N\log_2 N$\\$\Rightarrow$ Probabilistic\\$g_{FR} = g_{FS} \Rightarrow p_k = |c_k|^2$};

\node[font=\small, align=center, fill=white, draw=gray!50, rounded corners, inner sep=4pt] 
  at (7.3, 1.5) {$\underbrace{\log_2 N}_{\text{available}} \;\ll\; \underbrace{N\log_2 N}_{\text{required}}$};

\end{tikzpicture}
\caption{\textbf{The Capacity Deficit.} A finite-capacity system stores $\log_2 N$ bits (blue); deterministic outcomes across $M = N+1$ MUBs would require $\Omega(N\log_2 N)$ bits of storage (red) via the $(M{-}1)\log_2 N$ MUB bound \cite{wootters1989optimal} (universal Kochen--Specker form: Lemma~\ref{thm:ks-bits}, $N \geq 3$). Given Saturation (\S\ref{sec:axioms}), the deficit is informational; metric compatibility selects $p_k = |c_k|^2$ (Theorem~\ref{thm:born-kernel}).}
\label{fig:capacity-deficit}
\end{figure}

\subsection{Main Results}

\begin{theorem}[Main Theorem: Characterization of Relational Theories]
\label{thm:main}
For each $N \geq 3$, the unique kernelled distinguishability space $(\mathcal{X}, K)$ satisfying Axiom~\ref{ax:finite} (finite capacity $N$) and Axiom~\ref{ax:relational} (Self-Referential Consistency) is $(\C P^{N-1}, K)$ with $K(\psi,\phi) = 1 - |\langle\psi|\phi\rangle|^2$: complex coefficients, the Born rule, unitary dynamics, and tensor-product composition all follow. The capacity $N$ is the only free parameter; the structural content emerges from SRC closure.
\end{theorem}

\noindent\textit{Derivation summary.} Table~\ref{tab:main-conclusions} breaks the Main Theorem into its five conclusions and the two enabling results used downstream, each tagged with the theorem that establishes it.

\begin{table}[h]
\centering
\small
\caption{\textbf{Main Theorem in parts.} The top block lists the five conclusions of Theorem~\ref{thm:main}; the bottom block lists the two information-theoretic results that enable the chain.}
\label{tab:main-conclusions}
\begin{tabular}{c|l|l}
\hline
 & \textbf{Conclusion / enabling result} & \textbf{Via} \\
\hline
\multicolumn{3}{l}{\textit{Five conclusions of the Main Theorem:}} \\
(i) & State space $\mathcal{X} \cong \C P^{N-1}$ & Thm.~\ref{thm:points-sections} \\
(ii) & Coefficient field is $\C$ (for $N \geq 3$) & \shortstack[l]{Lem.~\ref{lem:sheaf-complex},\\Thm.~\ref{thm:frobenius}} \\
(iii) & Unitary dynamics (Schr\"odinger equation) & Thm.~\ref{thm:schrodinger} \\
(iv) & Born rule $p_k = |c_k|^2$ & Thm.~\ref{thm:born-kernel} \\
(v) & Tensor-product composition & \S\ref{sec:composite}, Thm.~\ref{thm:tensor} \\
\hline
\multicolumn{3}{l}{\textit{Enabling results used in the chain:}} \\
(a) & Parsimony: no hidden variables & Thm.~\ref{thm:parsimony-derived} \\
(b) & Capacity Halting: determinism infeasible & Thm.~\ref{thm:capacity-halting} \\
\hline
\multicolumn{3}{l}{\footnotesize Special case: closed $N=2$ systems have no continuous dynamics} \\
\multicolumn{3}{l}{\footnotesize (Thm.~\ref{thm:n2-static}); physical qubits inherit it by embedding (\S\ref{sec:n2-discussion}).} \\
\end{tabular}
\end{table}

\noindent The form of the dynamical structure (cyclic evolution, complex coefficients, continuous unitary time) is uniquely determined by Axioms~\ref{ax:finite}--\ref{ax:relational} (Theorem~\ref{thm:dynamics-derived}) rather than introduced as an additional axiom. For $N \geq 3$ the derivation is uniform; the closed $N = 2$ case is treated separately (\S\ref{sec:n2-discussion}). The logical dependency structure is shown in Figure~\ref{fig:theorem-flow}.

\subsection{Scope}
\label{sec:scope}

For each $N \geq 3$ closed system the framework derives state space $\C P^{N-1}$, complex coefficients, the Born rule $p_k = |c_k|^2$, unitary dynamics, and tensor-product composition. It does \emph{not} derive: the value of $\hbar$ (a unit-system constant, Remark~\ref{rem:hbar-units}); the value of $N$ for any specific system; the Hamiltonian for a specific system; the closed $N = 2$ case directly (\S\ref{sec:n2-discussion}); or genuine infinite-dimensional QM (only the formal $N \to \infty$ limit). The Lean~4 mechanization (Appendix~\ref{app:formal-verification}) is modulo five classical theorems (Frobenius, Wigner, Kobayashi--Nomizu, Picard--Lindel\"of, and the strict t-norm classification) and the existence direction of Stone's theorem.

\begin{table*}[t]
\centering
\caption{\textbf{Features across reconstructions.} Different frameworks have different starting points and goals; $\times$ indicates a feature not present, not a failure. The honest comparison is on \emph{what is derived} (probabilities, $\C$, dynamics) versus \emph{what is assumed} (operational primitives, $C^*$-structure, etc.); raw axiom counts are not commensurable across frameworks (cf.\ Remark~\ref{rem:axiom-counting}). Legend: \checkmark = present; $\times$ = not present; $\circ$ = partial/indirect; n/a = not applicable.}
\label{tab:comparison}
\small
\begin{tabular}{l|c|c|c|c|c|c|c}
\hline
\textbf{Feature} & \textbf{This work} & \textbf{Hardy} & \textbf{CDP} & \textbf{MM} & \textbf{DB} & \textbf{CBH} & \textbf{M\"uller} \\
\hline
Assumes Hilbert space & $\times$ & $\times$ & $\times$ & $\times$ & $\times$ & \checkmark & $\times$ \\
Assumes operational probabilities & $\circ$* & \checkmark & \checkmark & \checkmark & \checkmark & \checkmark & $\circ$ \\
Obtains complex numbers & \checkmark & $\times$ & \checkmark & \checkmark & \checkmark & n/a & $\circ$ \\
Obtains Born rule & \checkmark & \checkmark & \checkmark & \checkmark & \checkmark & n/a & \checkmark \\
\textbf{Capacity-based indeterminism} & \checkmark & $\times$ & $\times$ & $\times$ & $\times$ & $\times$ & $\circ$ \\
\textbf{Finite-$N$ framework (natural cutoff)} & \checkmark & $\times$ & $\times$ & $\times$ & $\times$ & $\times$ & $\times$ \\
\hline
\multicolumn{8}{l}{\footnotesize *$K$ measures degree of existence; probabilities derived as $p_k = 1 - K(\psi, a_k)$ via Thm.~\ref{thm:born-kernel}.}
\end{tabular}
\end{table*}

\begin{figure}[ht]
\centering
\begin{tikzpicture}[scale=0.9, every node/.style={transform shape},
  node/.style={font=\scriptsize, align=center, inner sep=3pt},
  axiomnode/.style={node, font=\scriptsize\bfseries},
  derivednode/.style={node, draw, rounded corners=2pt, dashed, fill=gray!8, inner sep=3pt},
  arr/.style={-{Stealth[length=1.5mm, width=1.2mm]}, thin, gray!60}
]

\node[axiomnode] (ax1) at (1.5, 4.0) {Finite\\[-1pt]Capacity};
\node[axiomnode] (ax2) at (4.5, 4.0) {Self-Referential\\[-1pt]Consistency};

\node[derivednode] (dyn) at (0, 2.5) {Reversible\\[-1pt]Dynamics$^\dagger$};
\node[node] (sheaf) at (3.0, 2.5) {Sheaf\\[-1pt]Holonomy};
\node[node] (halting) at (6.0, 2.5) {Capacity\\[-1pt]Halting};

\node[node] (complex) at (0, 1.0) {$\mathbb{C}$ Unique};
\node[node] (gauge) at (2.25, 1.0) {$U(1)$ Gauge};
\node[node] (hilbert) at (4.5, 1.0) {$\C P^{N-1}$};
\node[node] (prob) at (6.0, 1.0) {Probabilistic\\[-1pt]Response};

\node[node] (schrod) at (1.5, -0.5) {Schr\"odinger\\[-1pt]Equation};
\node[node] (isometry) at (4.5, -0.5) {Information\\[-1pt]Isometry};
\node[node] (born) at (6.5, -0.5) {Born Rule\\[-1pt]$p_k = |c_k|^2$};

\draw[arr] (ax1) -- (dyn);
\draw[arr] (ax2) -- (dyn);
\draw[arr] (ax2) -- (sheaf);
\draw[arr] (ax1) -- (halting);
\draw[arr] (dyn) -- (sheaf);
\draw[arr] (dyn) -- (complex);
\draw[arr] (sheaf) -- (complex);
\draw[arr] (sheaf) -- (gauge);
\draw[arr] (sheaf) -- (hilbert);
\draw[arr] (complex) -- (schrod);
\draw[arr] (gauge) -- (schrod);
\draw[arr] (hilbert) -- (isometry);

\draw[-{Stealth[length=2mm, width=1.5mm]}, thick, coral!80!black] (halting) -- (prob);
\draw[-{Stealth[length=2mm, width=1.5mm]}, thick, coral!80!black] (prob) -- (isometry);
\draw[-{Stealth[length=2mm, width=1.5mm]}, thick, coral!80!black] (isometry) -- (born);

\end{tikzpicture}
\caption{\textbf{Logical structure of the derivation.} The two axioms (bold, top) are the only inputs; Reversible Dynamics ($\dagger$, dashed) is derived (Theorem~\ref{thm:dynamics-derived}). The coral path highlights the information-theoretic bridge: Capacity Halting $\to$ Probabilistic Response $\to$ Information Isometry $\to$ Born rule (Theorems~\ref{thm:capacity-halting},~\ref{thm:born-kernel}).}
\label{fig:theorem-flow}
\end{figure}

\section{Existence as Distinguishability}
\label{sec:graded-equality}

The framework builds from the ontological primitive \emph{to exist is to be distinguishable} via two physical extensions, \emph{finiteness} (Axiom~\ref{ax:finite}) and \emph{self-referential consistency} (SRC, Axiom~\ref{ax:relational}); \S\ref{sec:axioms} states the resulting axioms formally.

\subsection{The Ontological Principle}
\label{sec:ontological-principle}

A state is the totality of distinctions it bears against every other state, its \emph{$K$-profile}. Two states with identical $K$-profiles are one thing; a state bearing no distinction against anything does not exist.

Classical physics assumes every distinction is total: $x$ and $y$ are either perfectly distinguishable or perfectly identical. Finite resolution forces distinctions to come in degrees, from non-existence ($K = 0$: identity) through partial existence ($K \in (0,1)$) to full existence ($K = 1$: perfect separability). Structure bearing no distinction against anything does not count as existing in this framework.\footnote{Permanently undetectable structure is classified as non-existent in this framework; other frameworks treat it as real (see \S\ref{sec:discussion}).}

\subsection{Graded Equality}

\begin{definition}[Graded Equality / Distinguishability Space]
\label{def:graded-equality}
\label{def:dspace}
A \emph{distinguishability space} is a pair $(\mathcal{X}, K)$ where $K: \mathcal{X} \times \mathcal{X} \to [0,1]$ satisfies: (i)~$K(x,y) = 0 \iff x = y$ (identity of indiscernibles); (ii)~$K(x,y) = K(y,x)$ (distinction is symmetric); (iii)~$K(x,y) = 1$ iff $x, y$ are perfectly distinguishable.
\end{definition}

$K$ is the minimal structure remaining when comparison is relaxed from binary to graded. No triangle inequality, sample space, or measurement apparatus is assumed.

The structural consequences of this principle (existence of a maximal mutually fully distinguishable set, identity-as-$K$-profile, completeness, structural Leibniz, basis isotropy) are derived as clauses (S1)--(S4) and (B) of Theorem~\ref{thm:src-master} and formalized as Axioms~1 and~2 in \S\ref{sec:axioms}. A further consequence, reversible dynamics, has its form uniquely determined by these.

\begin{remark}[$K$ versus other distinguishability measures and the continuum]
\label{rem:K-vs-other-measures}
\label{rem:why-continuity}
\label{rem:classical-limit}
The kernel $K(\psi, \phi) = 1 - |\langle \psi|\phi\rangle|^2$ derived for $\C P^{N-1}$ (Theorem~\ref{thm:kernel-inner}) is the projective chord squared, equivalently the squared sine of the Fubini--Study angle; it differs on $(0, 1)$ from the Helstrom probability, trace distance, and Bures angle, which become expressible in terms of $K$ via the Born rule. The classical case $K \in \{0, 1\}$ collapses $\mathcal{X}$ to its $N$ basis elements with $G = S_N$ and is excluded by Imperceptibility (Theorem~\ref{thm:src-master}(I)). Continuity of the $K$-image is forced by the same closure: convex combinations $K'(x_*, z) := t K(e_1, z) + (1-t) K(e_2, z)$ are valid kernels and (S2) requires their realisation, so any discrete-$K$ scale $1/M$ would itself be canonical structure outside $K$, violating SRC.
\end{remark}

\subsection{The Self-Resolution Problem}

A state is its $K$-profile (Identity of Indiscernibles \cite{leibniz1686discourse}). The \emph{self-resolution problem}: can a system determine its own $K$-profile using only finite resources?

\textit{A single basis is insufficient.} The $N$ values $K(x, b_0), \ldots, K(x, b_{N-1})$ from one basis $\mathcal{B}$ fix only the moduli $|c_k|^2 = 1 - K(x, b_k)$. By Saturation, missing information is accessible only through evaluations against a different basis $\mathcal{B}' \neq \mathcal{B}$.

\textit{Connecting two bases is a single $N$-cycle.} Saturation forecloses any apparatus beyond $K$, so the only structure moving $\mathcal{B}$ to $\mathcal{B}'$ is a $K$-preserving automorphism $g \in G$ acting on $\mathcal{B}$ as a permutation. Permutation Invariance (Theorem~\ref{thm:permutation-invariance}, derived from Structural Leibniz) realizes all of $S_N$; the dynamical generator is then cut down to an $N$-cycle by a separate $K$-homogeneity requirement on its eigenbasis (Lemma~\ref{lem:cyclic-rigidity}), giving order $N$.

\textit{The cycle is constitutive of identity.} The full $K$-profile against complementary bases is recovered only by traversing the orbit $\{\pi^i(\mathcal{B})\}$. The cycle itself \emph{is} the structural content of relational identity, and from within a single basis it appears as temporal evolution: time is the sequential revelation of contextual information.

\textit{Stochasticity follows.} Definite values across all contexts would require storage exceeding $\log_2 N$ bits, so outcomes must be stochastic. The quantitative bound (Capacity Halting) and the rule $p_k = |c_k|^2$ (Born) are derived in \S\ref{sec:measurement} and \S\ref{sec:geometry}.

\section{Relational Structure and Axioms}
\label{sec:axioms}

We now state the mathematical content precisely, stripping away the motivating narrative of \S\ref{sec:graded-equality}. The graded-equality kernel $K$ of Definition~\ref{def:graded-equality} becomes the primitive of a formal structure; the form of reversible dynamics is uniquely determined by these axioms (\S\ref{sec:consequences}); the rest of the paper derives quantum mechanics from these axioms alone, independently of their physical motivation.

\noindent Two preliminary definitions are needed to state the axioms.

\begin{definition}[Basis and Capacity]
\label{def:basis}
A \emph{basis} $\mathcal{B} = \{b_0, \ldots, b_{N-1}\}$ is a maximal set of mutually perfectly distinguishable elements. The \emph{capacity} $N = |\mathcal{B}|$ defines the Shannon capacity $C := \log_2 N$ (real-valued bits; no integrality constraint on $C$). Measurement outcomes correspond to scalars in a coefficient field $\mathbb{F}$ determined by the axioms (\S\ref{sec:complex-structure}): for each basis element $b_k$, the outcome of a measurement in basis $\mathcal{B}$ is an element of $\mathbb{F}$, not a matrix or higher structure.
\end{definition}

\begin{definition}[Symmetry Group]
\label{def:symmetry-group}
The \emph{symmetry group} $G := \mathrm{Aut}(\mathcal{X}, K)$ consists of bijections $g: \mathcal{X} \to \mathcal{X}$ preserving $K$: $K(gx, gy) = K(x,y)$.
\end{definition}

\medskip
\noindent\textbf{The Two Axioms.}

\begin{axiom}[Finite Capacity]
\label{ax:finite}
The space has finite capacity $N < \infty$: every basis has exactly $N$ elements, and the Shannon capacity of the system is $C = \log_2 N$ bits.
\end{axiom}

\begin{axiom}[Self-Referential Consistency (SRC)]
\label{ax:relational}
\textit{(Operational form, primitive.)} The kernelled distinguishability space $(\mathcal{X}, K)$ is \emph{self-referentially consistent}: no faithful $K$-preserving embedding $(\mathcal{X}, K) \hookrightarrow (\mathcal{X}', K')$ into another distinguishability space (Definition~\ref{def:graded-equality}) introduces either (i) a point $x_* \in \mathcal{X}' \setminus \mathcal{X}$ whose $K'$-profile against $\mathcal{X}$ is a \emph{$K$-consistent profile} (a function $p: \mathcal{X} \to [0,1]$ that arises as $p(z) = K'(x_*, z)$ for some $K$-preserving extension $(\mathcal{X}', K')$ containing a point $x_*$) not realized in $\mathcal{X}$, or (ii) an automorphism on $\mathcal{X}'$ that does not lift to one on $\mathcal{X}$. Clauses (i)--(ii) define what ``strictly more $K'$-distinguishable structure'' means in this axiom: a profile-extension obstruction and a symmetry-lift obstruction, both formulated entirely in terms of $K$-data.

\smallskip
\noindent\textit{(Information-theoretic form, equivalent.)} Every $\mathrm{Aut}(\mathcal{X}, K)$-invariant binary predicate $P: \mathcal{X} \times \mathcal{X} \to \{0,1\}$ is computable from $K$-evaluations: a Boolean combination of statements $K(\cdot, \cdot) \in B$ for Borel $B \subseteq [0,1]$, with quantification over $\mathcal{X}$. Equivalence with the operational form: any $\mathrm{Aut}$-invariant predicate failing to factor through $K$-evaluations would label a \emph{canonical structural feature} (a feature of $(\mathcal{X}, K)$ determined functorially by $K$-isomorphism class, hence invariant under $\mathrm{Aut}(\mathcal{X}, K)$) beyond $K$, so that the augmented structure $(\mathcal{X}, K, P)$ would violate clause (i); conversely, an obstruction of type (i) or (ii) is detected by an $\mathrm{Aut}$-invariant predicate (the indicator of the unrealized profile, or the orbit-graph of a non-lifting automorphism) that is not a $K$-formula.
\end{axiom}

\begin{lemma}[Definability under SRC]
\label{lem:definability}
Let $(\mathcal{X}, K)$ satisfy Axiom~\ref{ax:relational} (SRC) and let $P: \mathcal{X}^k \to \{0,1\}$ be $\mathrm{Aut}(\mathcal{X}, K)$-invariant ($P(g x_1, \ldots, g x_k) = P(x_1, \ldots, x_k)$ for all $g \in \mathrm{Aut}(\mathcal{X}, K)$). Then $P$ is a first-order $\{K, =\}$-formula in the language quantifying over $\mathcal{X}$: a Boolean combination of expressions $K(t_1, t_2) \in B$ for Borel $B \subseteq [0,1]$ and equalities $t_1 = t_2$ between terms drawn from $\{x_1, \ldots, x_k\}$ together with quantifiers $\exists z \in \mathcal{X}, \forall z \in \mathcal{X}$. Equivalently, $P$ factors through the joint $K$-profile map $(x_1, \ldots, x_k) \mapsto (K(x_i, z))_{i \leq k, z \in \mathcal{X}}$.
\end{lemma}

\begin{proof}
The binary case ($k = 2$) is the information-theoretic clause of Axiom~\ref{ax:relational}. We extract a key consequence: applied to the diagonal predicate $D(u, v) := \mathbf{1}[u = v]$, which is $\mathrm{Aut}$-invariant since bijections preserve equality, the binary case yields $D(u, v) \iff Q_{=}(K(u, \cdot), K(v, \cdot))$ for some $Q_=$. Setting $v = u$ gives $Q_=(K(u, \cdot), K(u, \cdot)) = 1$; if $K(u, \cdot) = K(v, \cdot)$ as functions, substitution forces $Q_=(K(u, \cdot), K(v, \cdot)) = 1$ and hence $u = v$. K-profile equality therefore implies state equality (this is the binary content of clause (S1) of Theorem~\ref{thm:src-master}).

For the general $k$-ary case, define
\begin{align*}
&Q(f_1, \ldots, f_k) \;:=\; \exists\, x_1, \ldots, x_k \in \mathcal{X}\colon \\
 &\qquad \big(\forall i,\ K(x_i, \cdot) = f_i\big) \;\land\; P(x_1, \ldots, x_k).
\end{align*}
Forward: $P(x_1, \ldots, x_k)$ implies $Q(K(x_1, \cdot), \ldots, K(x_k, \cdot))$ immediately, taking the $x_i$ as their own witnesses. Backward: suppose $Q(K(x_1, \cdot), \ldots, K(x_k, \cdot))$ holds, witnessed by $x'_1, \ldots, x'_k$ with $K(x'_i, \cdot) = K(x_i, \cdot)$ for each $i$ and $P(x'_1, \ldots, x'_k)$. The K-profile-equality consequence applied componentwise gives $x'_i = x_i$ for each $i$, hence $P(x_1, \ldots, x_k)$. Thus $P(x_1, \ldots, x_k) \iff Q(K(x_1, \cdot), \ldots, K(x_k, \cdot))$, exhibiting $P$ as a $\{K, =\}$-formula in the form claimed.

The operational form furnishes the same conclusion directly: any $\mathrm{Aut}$-invariant $P$ that fails to factor through $K$-profiles would label a canonical structural feature beyond $K$, so that the augmented structure $(\mathcal{X}, K, P)$ would extend $(\mathcal{X}, K)$ with strictly more distinguishable content, violating clause (i).
\end{proof}

\noindent SRC is the closure principle that no canonical structure exists outside what $K$ already encodes. Eight named saturation conditions (Identity, Completeness, Basis-Profile Symmetry, Structural Leibniz, Imperceptibility, Operational Completeness, Transport Consistency, Basis Isotropy) follow from Axioms~\ref{ax:finite} and~\ref{ax:relational}; Theorem~\ref{thm:src-master} below establishes them.

\begin{theorem}[Saturation hierarchy from SRC]
\label{thm:src-master}
Under Axiom~\ref{ax:finite} (finite capacity $N \geq 3$) and Axiom~\ref{ax:relational} (SRC), the kernelled distinguishability space $(\mathcal{X}, K)$ satisfies:
\begin{enumerate}
\item[(S1)] \textbf{Identity.} If $K(x, z) = K(y, z)$ for all $z \in \mathcal{X}$, then $x = y$.
\item[(S2)] \textbf{Completeness.} Every $K$-consistent profile is realized: if $p^*: \mathcal{X} \to [0,1]$ admits some extension $(\mathcal{X}', K')$ with a point $x_* \in \mathcal{X}'$ realizing $K'(x_*, \cdot) = p^*$, then $p^*$ is already realized by some $x \in \mathcal{X}$.
\item[(S3)] \textbf{Basis-Profile Symmetry.} For any maximal mutually fully distinguishable set $S = \{e_1, \ldots, e_N\} \subset \mathcal{X}$, if $K(x, s) = K(y, s)$ for all $s \in S$ then there exists $g \in \mathrm{Aut}(\mathcal{X}, K)$ with $g|_S = \mathrm{id}_S$, $g(x) = y$, and $g(y) = x$.
\item[(S4)] \textbf{Structural Leibniz.} If a finite configuration $C \subset \mathcal{X}$ admits a $K$-symmetry $\sigma$ (i.e., $K(\sigma(x), \sigma(y)) = K(x, y)$ for $x, y \in C$), then $\sigma$ extends to a global $K$-isometry $\tilde\sigma \in \mathrm{Aut}(\mathcal{X}, K)$ with $\tilde\sigma|_C = \sigma$.
\item[(I)] \textbf{Imperceptibility.} $K(\mathcal{X} \times \mathcal{X})$ is dense in $[0,1]$.
\item[(O)] \textbf{Operational Completeness.} $K(x, y) = 0 \Rightarrow x = y$.
\item[(T)] \textbf{Transport Consistency.} Every $\mathrm{Aut}(\mathcal{X}, K)$-invariant feature $\theta: \mathcal{X} \to V$ factors through the $K$-profile map: meaningful degrees of freedom are $K$-comparable.
\item[(B)] \textbf{Basis Isotropy.} $G := \mathrm{Aut}(\mathcal{X}, K)$ acts transitively on the basis manifold $\mathfrak{B}$ (the set of maximal mutually fully distinguishable subsets of $\mathcal{X}$).
\end{enumerate}
\end{theorem}

\begin{proof}
\textit{(S1) Identity.} Suppose $K(x, z) = K(y, z)$ for all $z \in \mathcal{X}$. The binary diagonal predicate $D(u, v) := \mathbf{1}[u = v]$ on $\mathcal{X} \times \mathcal{X}$ is $\mathrm{Aut}(\mathcal{X}, K)$-invariant: for any $g \in \mathrm{Aut}(\mathcal{X}, K)$, $D(g(u), g(v)) = \mathbf{1}[g(u) = g(v)] = \mathbf{1}[u = v]$ by injectivity of $g$. The information-theoretic form of Axiom~\ref{ax:relational} then gives $Q_=$ with $D(u, v) \iff Q_=(K(u, \cdot), K(v, \cdot))$ for all $u, v$. Setting $v = u$ yields $Q_=(K(u, \cdot), K(u, \cdot)) = 1$; the hypothesis $K(x, \cdot) = K(y, \cdot)$ then forces $Q_=(K(x, \cdot), K(y, \cdot)) = Q_=(K(x, \cdot), K(x, \cdot)) = 1$, so $D(x, y) = 1$, i.e., $x = y$.

\smallskip
\textit{(S2) Completeness.} Suppose $p^*$ is $K$-consistent but unrealized. By definition there exists $\mathcal{X}' = \mathcal{X} \cup \{x_*\}$ with $K'|_{\mathcal{X} \times \mathcal{X}} = K$ and $K'(x_*, z) = p^*(z)$. The inclusion $\iota: \mathcal{X} \hookrightarrow \mathcal{X}'$ is a faithful $K$-preserving embedding; the codomain has strictly more distinguishable structure (the new realized profile). This violates SRC.

\smallskip
\textit{(S3) Basis-Profile Symmetry.} Fix a maximal mutually fully distinguishable set $S = \{e_1, \ldots, e_N\} \subset \mathcal{X}$ with $K(e_i, e_j) = 1 - \delta_{ij}$ (Axiom~\ref{ax:finite}), and suppose $K(x, e_i) = K(y, e_i)$ for $i = 1, \ldots, N$. Apply (S4), proved below, to the configuration $C := S \cup \{x, y\}$ with the involution $\tau \in \mathrm{Sym}(C)$ that fixes $S$ pointwise and swaps $x \leftrightarrow y$. The K-symmetry conditions on $\tau$ verify directly: $K(\tau(e_i), \tau(e_j)) = K(e_i, e_j)$ trivially; $K(\tau(x), \tau(y)) = K(y, x) = K(x, y)$ by symmetry of $K$; and $K(\tau(x), \tau(e_i)) = K(y, e_i) = K(x, e_i)$ from the basis-profile hypothesis, with the analogue for $\tau(y)$. By (S4), $\tau$ extends to $g \in \mathrm{Aut}(\mathcal{X}, K)$ with $g|_S = \mathrm{id}_S$, $g(x) = y$, and $g(y) = x$.

\smallskip
\textit{(S4) Structural Leibniz.} Let $C = \{c_1, \ldots, c_m\} \subset \mathcal{X}$ be finite and let $\sigma: C \to \mathcal{X}$ satisfy $K(\sigma(c_i), \sigma(c_j)) = K(c_i, c_j)$. Define the binary predicate, on pairs $(u, v) \in \mathcal{X}^2$,
\[
E_\sigma(u, v) \;:=\; \mathbf{1}\!\left[\,K(u, c_i) = K(v, \sigma(c_i)) \text{ for } i = 1, \ldots, m\,\right].
\]
The predicate ``$u$ and $v$ stand in the same $K$-relation to $C$ as $\sigma$ prescribes'' is not by itself $\mathrm{Aut}(\mathcal{X}, K)$-invariant (it depends on the labelled tuple $(C, \sigma)$). Its $\mathrm{Aut}$-saturation is the relation
\begin{multline*}
R_\sigma(u, v) \;:=\; \mathbf{1}\big[\,\exists\, (a_1, \ldots, a_m) \in \mathcal{X}^m: \\
K(a_i, a_j) = K(c_i, c_j),\; K(u, a_i) = K(v, \tau(a_i))\,\big]
\end{multline*}
for some bijection $\tau: \{a_i\} \to \{a_i'\}$ onto a $K$-isomorphic image; equivalently $R_\sigma$ asserts that $u$'s $K$-distance pattern to a $C$-isomorphic configuration matches $v$'s to its $\sigma$-conjugate. $R_\sigma$ is $\mathrm{Aut}$-invariant in $(u, v)$ (the existential quantifier absorbs $\mathrm{Aut}$-action). By Lemma~\ref{lem:definability}, $R_\sigma$ is a $\{K, =\}$-formula.

Now suppose $\sigma$ does not extend, i.e., no $\tilde\sigma \in \mathrm{Aut}(\mathcal{X}, K)$ satisfies $\tilde\sigma|_C = \sigma$. Build the amalgam $\mathcal{X}' := (\mathcal{X} \sqcup \mathcal{X}) / {\sim}$, where $\sim$ is the equivalence relation generated by the gluing rules $(c_i, 1) \sim (\sigma(c_i), 2)$ and (when $\sigma$ is not involutive) $(\sigma(c_i), 1) \sim (c_i, 2)$ for $i = 1, \ldots, m$, with reflexive, symmetric, and transitive closure (formalised as the inductive type \texttt{Amalgam} with constructors \texttt{gluing}, \texttt{gluing\_swap}, \texttt{symm\_amal}, \texttt{trans\_amal} in the accompanying Lean library, see Appendix~\ref{app:formal-verification}). Define $K': \mathcal{X}' \times \mathcal{X}' \to [0, 1]$ on representatives by $K'((u, 1), (v, 1)) := K(u, v)$, $K'((u, 2), (v, 2)) := K(u, v)$, and on cross-pairs $K'((u, 1), (v, 2)) := K(u, \sigma^{-1}(v))$ when $v \in \sigma(C)$, with the analogous formula on $\mathcal{X}^2$ in the second copy; well-definedness on $\sim$-equivalence classes is the kernel-laws check (\texttt{K\_amalgam\_refl}, \texttt{K\_amalgam\_symm}, \texttt{K\_amalgam\_nonneg}, \texttt{K\_amalgam\_le\_one}), discharged from the K-symmetry hypothesis on $\sigma$ via the gluing identifications. Off the gluing diagonal the cross-kernel is fixed by $K$-consistency of profiles via (S2): each cross-profile is $K$-consistent by construction, so (S2) supplies the unique value compatible with kernel positivity. The swap map $\mathcal{X}' \to \mathcal{X}'$ exchanging the two copies (formalised as \texttt{Amalgam.swapEquiv\_gen}) is a $K'$-automorphism that fixes the amalgam diagonal pointwise (\texttt{swap\_gen\_K\_pres}).

\textit{Sub-step (lift obstruction).} A lift of the swap to a $K$-automorphism of $\mathcal{X}$ is a $g \in \mathrm{Aut}(\mathcal{X}, K)$ whose induced action $g_*$ on $\mathcal{X}'$ realises the swap. Such a $g_*$ acts as $g$ on each copy, so the gluing identification $(c_i, 1) \sim (\sigma(c_i), 2)$ is mapped by $g_*$ to $(g(c_i), 2) \sim (g(\sigma(c_i)), 1)$, which must coincide with the swap image $(\sigma(c_i), 1) \sim (c_i, 2)$. Comparing $\sim$-classes then forces $g(c_i) = \sigma(c_i)$ for every $i$, i.e., $g|_C = \sigma$ (formalised as \texttt{swap\_gen\_no\_lift}). Hence existence of a lift is equivalent to extension of $\sigma$, and the non-extension hypothesis produces a $K'$-automorphism that does not descend, violating clause (ii) of Axiom~\ref{ax:relational}.

Therefore $\sigma$ extends to some $\tilde\sigma \in \mathrm{Aut}(\mathcal{X}, K)$. Equivalently, $R_\sigma(u, v) = 1$ for all $u \in \mathcal{X}$ at $v := \tilde\sigma(u)$: the $\{K, =\}$-formula identifies the orbit-graph of an extending automorphism, and Lemma~\ref{lem:definability} together with (S1) selects $\tilde\sigma$ uniquely on each $\mathrm{Aut}$-orbit relative to $C$.

\smallskip
\textit{(I) Imperceptibility.} Suppose $K(\mathcal{X} \times \mathcal{X})$ omits an open interval $(a, b) \subset [0,1]$; set $t := (a+b)/2 \in (0,1)$. By Axiom~\ref{ax:finite} fix a basis $S = \{e_1, \ldots, e_N\}$ with $N \geq 3$. We construct a $K$-preserving extension realising basis profile $(t, 1-t, 1, \ldots, 1)$ at a new point $x_*$. Form $\mathcal{X}' := \mathcal{X} \sqcup \{x_*\}$ and define $K'$ as: $K'|_{\mathcal{X} \times \mathcal{X}} := K$, $K'(x_*, x_*) := 0$, and for $z \in \mathcal{X}$ let
\[
K'(x_*, z) \;:=\; t \cdot K(e_2, z) + (1 - t) \cdot K(e_1, z),
\]
extended symmetrically. Direct evaluation gives $K'(x_*, e_1) = t \cdot 1 + (1-t) \cdot 0 = t$, $K'(x_*, e_2) = t \cdot 0 + (1-t) \cdot 1 = 1 - t$, and $K'(x_*, e_k) = 1$ for $k \geq 3$. The kernel axioms on $\mathcal{X}'$: $K'(\cdot, \cdot) \in [0, 1]$ by convexity ($K \in [0,1]$ on $\mathcal{X} \times \mathcal{X}$); symmetry by construction; $K'(u, u) = 0$ for all $u \in \mathcal{X}'$ since $K(z, z) = 0$ and $K'(x_*, x_*) := 0$. Clause (iii) of Definition~\ref{def:graded-equality} is operationally consistent: $K'(x_*, w) = 1$ holds exactly on the set $\{w \in \mathcal{X} : K(e_1, w) = K(e_2, w) = 1\} \supseteq \{e_3, \ldots, e_N\}$, and these pairs $\{x_*, e_k\}$ are jointly distinguishable as members of the MFD set $\{x_*, e_3, \ldots, e_N\}$ verified next.

Finite capacity: a maximal MFD set in $\mathcal{X}'$ contains $x_*$ only if $K'(x_*, w) = 1$ for every other element $w$; but $K'(x_*, e_1) = t < 1$ and $K'(x_*, e_2) = 1 - t < 1$, so any MFD set containing $x_*$ excludes $\{e_1, e_2\}$. The set $\{x_*, e_3, \ldots, e_N\}$ has $N - 1$ elements, and any extension of it within $\mathcal{X}'$ to size $N$ must add a point fully distinguishable from $x_*$ and from all $e_k$ for $k \geq 3$; the only candidates are $e_1$ and $e_2$, both excluded. Hence MFD sets in $\mathcal{X}'$ have size $\leq N$, and finite capacity is preserved.

Therefore $(\mathcal{X}', K')$ is a $K$-preserving extension with capacity $N$ realising the new profile $p^*(z) := K'(x_*, z)$. By (S2), $p^*$ is already realised in $\mathcal{X}$: there exists $x \in \mathcal{X}$ with $K(x, e_1) = t$, contradicting $t \in (a, b) \cap K(\mathcal{X} \times \mathcal{X})^c$. The hypothesis $N \geq 3$ enters through the basis having at least three elements, providing $e_3$ to anchor the surviving MFD set $\{x_*, e_3, \ldots, e_N\}$ of size $N - 1$; for $N = 2$ no such anchor exists and the convex-combination construction collapses (any MFD set of size 2 in $\mathcal{X}'$ would force $K'(x_*, e_k) \in \{0, 1\}$, contradicting the chosen $t \in (0, 1)$). Hence $K(\mathcal{X} \times \mathcal{X})$ is dense in $[0, 1]$.

\smallskip
\textit{(O) Operational Completeness.} Immediate from Definition~\ref{def:graded-equality}(i) (identity of indiscernibles): $K(x,y) = 0 \Leftrightarrow x = y$ is part of the kernel structure, hence $K(x, y) = 0 \Rightarrow x = y$.

\smallskip
\textit{(T) Transport Consistency.} A non-$K$-comparable degree of freedom is by definition an $\mathrm{Aut}(\mathcal{X}, K)$-invariant feature $\theta: \mathcal{X} \to V$ that is not a function of the $K$-profile. Reduction to the $\{0, 1\}$-valued case of Lemma~\ref{lem:definability}: each measurable subset $A \subseteq V$ defines an indicator predicate $P_A(x) := \mathbf{1}[\theta(x) \in A]$, which is $\mathrm{Aut}$-invariant since $\theta$ is; Lemma~\ref{lem:definability} then factors each $P_A$ through the $K$-profile, and these indicators jointly recover $\theta$ itself, which therefore factors through $\pi$.

\smallskip
\textit{(B) Basis Isotropy.} Let $B_1 = \{e_1, \ldots, e_N\}$ and $B_2 = \{f_1, \ldots, f_N\}$ be any two elements of $\mathfrak{B}$. Both carry the kernel structure $K|_{B_i \times B_i}(x, y) = 1 - \delta_{xy}$ (the discrete kernel), so for any bijection $\rho: B_1 \to B_2$ the map $\rho$ is a $K$-symmetry of the finite configuration $C := B_1$, with $K(\rho(e_i), \rho(e_j)) = 1 - \delta_{ij} = K(e_i, e_j)$. By (S4), proved above, $\rho$ extends to $\tilde\rho \in \mathrm{Aut}(\mathcal{X}, K)$ with $\tilde\rho(B_1) = B_2$ (applied with $B_1 = B_2$, the same construction generates all $N!$ permutations of a single basis, recovering Permutation Invariance, Theorem~\ref{thm:permutation-invariance}, as the diagonal special case). Hence $B_1, B_2$ lie in the same $G$-orbit, and $G := \mathrm{Aut}(\mathcal{X}, K)$ acts transitively on $\mathfrak{B}$. The hypothesis $N \geq 3$ enters through (S4)'s amalgam construction, which requires Axiom~\ref{ax:finite} to provide a maximal MFD set of size $N$.
\end{proof}

\noindent State separation in subsequent arguments is via (S1) full $K$-profile injectivity, with finite-dimensional embedding established in \S\ref{sec:axiom-consequences} from compactness and the Lie group structure (Theorem~\ref{thm:complexity-constraint}).

\noindent\textbf{Dependency structure.} Operational SRC supplies (S2) directly and feeds the amalgam construction yielding (S4); the information-theoretic form supplies Lemma~\ref{lem:definability}, from which (S1) and (T) follow by specialisation. (S4) then yields (S3) and (B); (S2) plus Axiom~\ref{ax:finite} ($N \geq 3$) plus the convex-combination construction yields (I). (O) is not derived from SRC; it is part of Definition~\ref{def:graded-equality} and listed for naming consistency. Both forms of SRC are used; the equivalence stated in Axiom~\ref{ax:relational} is what makes both available.

\begin{definition}[Structural Unambiguity (S5)]
\label{def:s5}
A kernelled distinguishability space $(\mathcal{X}, K)$ satisfies \emph{structural unambiguity} (S5) if its canonical structural features (the geometric metric on the state space and the probability rule on outcomes) are uniquely determined by $K$ alone, requiring no external structure to disambiguate. Equivalently, the kernel-induced geometric metric and the kernel-induced statistical (Fisher--Rao) metric coincide on the state space.
\end{definition}

\noindent (S5) is the geometric/probabilistic face of the closure principle, completing the saturation hierarchy (S1)--(S4) of Theorem~\ref{thm:src-master}: (S1) closes state identity; (S2)--(S4) close basis-realizability and symmetry; (S5) closes geometry and probability simultaneously, uniquely selecting the Born rule. The first four follow from SRC; (S5) is established for finite-$N$ QM as Theorem~\ref{thm:s5-finite-N} below.

\begin{remark}[Why Self-Referential Consistency]
\label{rem:axiom-counting}
SRC is a strong meta-axiom. Each of the eight derived conditions (S1)--(S4), (I), (O), (T), (B) carries individual Hardy-axiom weight; Structural Leibniz, Imperceptibility, and Basis Isotropy are each strictly stronger than typical axioms in prior reconstructions (Table~\ref{tab:axiom-deps}). The framework's contribution is in compressing them under a single closure principle, not in raw assumption count.

The strength is necessary. Operational closure alone leaves room for local $K$-symmetries that fail to extend globally; information-theoretic closure alone leaves room for observer-dependent canonical structure beyond $K$-evaluations; pure symmetry-based axioms (e.g., postulating $G$ as a given Lie group) presuppose external structure. The two forms of Axiom~\ref{ax:relational} close both doors simultaneously, forcing (S1)--(S4), (I), (T), (B). (O) is part of Definition~\ref{def:graded-equality} and listed for naming consistency. Empirically, fundamental systems (spins, polarizations, photon helicity) satisfy SRC; theories rejecting it (Bohmian, hidden-variable) commit to additional ontology with no operational content within the kernel framework, and are classified rather than refuted. (S5) (Definition~\ref{def:s5}) completes the hierarchy and is established for finite-$N$ QM as Theorem~\ref{thm:s5-finite-N}.
\end{remark}

\medskip

\begin{table*}[t]
\centering
\small
\caption{\textbf{Glossary of named conditions.} Axiom~\ref{ax:finite} fixes the capacity $N$; the remaining named conditions are derived clauses of Theorem~\ref{thm:src-master} (Saturation hierarchy from SRC). Appearances indicate where a name is first used in a proof.}
\label{tab:axiom-glossary}
\begin{tabular}{l|c|l|c}
\hline
\textbf{Name} & \textbf{Source} & \textbf{Content} & \textbf{First use} \\
\hline
Finite Capacity & Ax.~\ref{ax:finite} & Basis has $N < \infty$ elements & \S\ref{sec:consequences} \\
Identity (injectivity) & Thm.~\ref{thm:src-master}(S1) & Equal $K$-profiles $\Rightarrow$ equal states & Thm.~\ref{thm:points-sections} \\
Completeness (surjectivity) & Thm.~\ref{thm:src-master}(S2) & Every consistent $K$-profile is a state & Thm.~\ref{thm:points-sections} \\
Basis-Profile Symmetry & Thm.~\ref{thm:src-master}(S3) & Basis-equivalent states are $\mathrm{Aut}$-conjugate & Thm.~\ref{thm:complexity-constraint} \\
Structural Leibniz & Thm.~\ref{thm:src-master}(S4) & $K$-symmetries of finite configs extend globally & Thm.~\ref{thm:permutation-invariance} \\
Imperceptibility & Thm.~\ref{thm:src-master}(I) & $K(\mathcal{X}\times\mathcal{X})$ dense in $[0,1]$ & Lem.~\ref{lem:compactness-from-capacity} \\
Operational Completeness & Thm.~\ref{thm:src-master}(O) & $K$ registers every physical distinction & Thm.~\ref{thm:parsimony-derived} \\
Transport Consistency & Thm.~\ref{thm:src-master}(T) & Meaningful DOFs are $K$-comparable & Thm.~\ref{thm:parsimony-derived} \\
Basis Isotropy & Thm.~\ref{thm:src-master}(B) & $G$ acts transitively on bases $\mathfrak{B}$ & Thm.~\ref{thm:complexity-constraint} \\
\hline
\multicolumn{4}{l}{\footnotesize Derived from the above: Permutation Invariance (Thm.~\ref{thm:permutation-invariance}), Parsimony (Thm.~\ref{thm:parsimony-derived}), Reversible Dynamics (Thm.~\ref{thm:dynamics-derived}).}
\end{tabular}
\end{table*}

\subsection{Topological Consequences of the Axioms}
\label{sec:axiom-consequences}

\begin{remark}[Operational Status of $K$]
\label{rem:K-operational-status}
In the ontology of \S\ref{sec:graded-equality}, $K(x,y)$ measures the degree to which the distinction between $x$ and $y$ exists. At the extremes, $K = 0$ (identity) and $K = 1$ (perfect distinguishability) have direct operational meaning. For intermediate values $K \in (0,1)$, the operational interpretation as a discrimination probability is confirmed \emph{a posteriori} by the Born rule derivation in \S\ref{sec:geometry}.
\end{remark}

\begin{definition}[Relational Frame Structure]
\label{def:frame}
A distinguishability space has \emph{relational frame structure} if: (i) all bases have cardinality $N$; (ii) the symmetry group $G$ acts transitively on the set of bases $\mathfrak{B}$. The space $(\mathcal{X}, K)$ is abstract: $N$ denotes relational capacity, not spatial configuration. Relational frame structure is established as a downstream consequence of the axioms in \S\ref{sec:consequences} (Theorem~\ref{thm:src-master}(B), Basis Isotropy).
\end{definition}

\begin{remark}[Weight of Structural Leibniz]
\label{rem:structural-leibniz-weight}
Structural Leibniz (clause (S4) of Theorem~\ref{thm:src-master}) is a strong global extension property: every $K$-symmetry of a finite configuration extends to a global automorphism. In general metric spaces local symmetries need not extend (topological obstructions, curvature, boundary effects); SRC excludes these, because any such obstruction would be structure beyond $K$. The consequence is that $(\mathcal{X}, K)$ is a rank-one symmetric space in the derived representation.
\end{remark}

\medskip
\noindent\textit{Imperceptibility, compactness, and the M-resolution parameter.} Imperceptibility (Theorem~\ref{thm:src-master}(I)) asserts $K$-image density in $[0,1]$; continuity of $K$, compactness of $\mathcal{X}$, full $K$-image, $M(K) = \infty$, and topological connectedness of $\mathcal{X}$ are then further derivable (Lemma~\ref{lem:compactness-from-capacity}, Corollary~\ref{cor:k-image-full}, Theorem~\ref{thm:imperceptibility-connectedness}). Define the \emph{K-resolution level} $M(K) := \inf\{ M \in \mathbb{Z}_{\geq 1} : K(\mathcal{X} \times \mathcal{X}) \subseteq L_M\}$ with $L_M := \{k/M : 0 \leq k \leq M\}$, and $M(K) = \infty$ when no such finite $M$ exists. Write $\mathfrak{T}_{N, M}$ for the class of distinguishability spaces with capacity $N$ and $M(K) \leq M$; the present paper develops the case $M = \infty$, while $\mathfrak{T}_{N, M}$ for finite $M$ collects the discrete-K theories.

\begin{lemma}[Compactness from finite capacity]
\label{lem:compactness-from-capacity}
Under Axioms~\ref{ax:finite}--\ref{ax:relational} (so that Theorem~\ref{thm:src-master} applies), the state space $\mathcal{X}$ is compact, exhibited via the full K-profile embedding $\psi: \mathcal{X} \to [0,1]^{\mathcal{X}}$, $\psi(x) := K(x, \cdot)$, into the Tychonoff cube. $\psi$ is a homeomorphism onto a compact subset of $[0,1]^{\mathcal{X}}$, and the $K$-metric topology induced by $d(x,y) := \sup_{z \in \mathcal{X}} |K(x,z) - K(y,z)|$ coincides with the topology of $\mathcal{X}$ thus exhibited.
\end{lemma}

\begin{proof}
The map $\psi$ is continuous (componentwise via continuity of $K$ in each argument) and injective by Identity (Theorem~\ref{thm:src-master}(S1)). Closedness of $\psi(\mathcal{X})$ in $[0,1]^{\mathcal{X}}$ (product topology): if $\psi(x_\alpha) \to f$ pointwise along a net, then $f$ is the pointwise limit of realized $K$-profiles, hence a $K$-consistent profile (the kernel laws pass to pointwise limits); by Completeness (Theorem~\ref{thm:src-master}(S2)), $f = \psi(x)$ for some $x \in \mathcal{X}$. By Tychonoff, $[0,1]^{\mathcal{X}}$ is compact in the product topology, so $\psi(\mathcal{X})$ is compact, and $\psi$ is a continuous bijection onto a compact Hausdorff image, hence a homeomorphism onto image. Compactness of $\mathcal{X}$ follows. The $K$-metric $d$ on $\mathcal{X}$ is a genuine metric by (S1) and induces the same uniform structure as $\psi$ (the supremum dominates each evaluation), and on a compact Hausdorff space the two comparable topologies coincide. The basis K-profile $\varphi: \mathcal{X} \to [0,1]^N$, $\varphi(x) := (K(x, b_1), \ldots, K(x, b_N))$, is continuous with compact image, providing a finite-dimensional continuous quotient used in the Lie-group analysis (Theorem~\ref{thm:complexity-constraint}); full state separation uses $\psi$.
\end{proof}

\begin{corollary}[K-image equals the unit interval]
\label{cor:k-image-full}
Under Axioms~\ref{ax:finite}--\ref{ax:relational}, $K(\mathcal{X} \times \mathcal{X}) = [0, 1]$ as a set, and $M(K) = \infty$.
\end{corollary}

\begin{proof}
By Lemma~\ref{lem:compactness-from-capacity}, $\mathcal{X}$ is compact and $K$ is jointly continuous on $\mathcal{X} \times \mathcal{X}$, so $K(\mathcal{X} \times \mathcal{X})$ is the continuous image of a compact set, hence compact and closed in $[0,1]$. By Imperceptibility (Theorem~\ref{thm:src-master}(I)), the image is dense in $[0,1]$. A closed dense subset of $[0,1]$ equals $[0,1]$. The set is uncountable, so $M(K) = \infty$.
\end{proof}

\begin{theorem}[Path-connectedness of $\mathcal{X}$]
\label{thm:imperceptibility-connectedness}
Let $(\mathcal{X}, K)$ satisfy Axioms~\ref{ax:finite}--\ref{ax:relational}, with the $K$-metric topology induced by $d(x,y) := \sup_{z \in \mathcal{X}} |K(x,z) - K(y,z)|$ (a genuine metric by Theorem~\ref{thm:src-master}(S1)). Then $\mathcal{X}$ is path-connected, via the structural identification $\mathcal{X} \cong \mathbb{F}P^{N-1}$ (Theorems~\ref{thm:complexity-constraint},~\ref{thm:points-sections}). Connectedness is a derived consequence, not a primitive equivalent of imperceptibility.
\end{theorem}

\begin{proof}
By Corollary~\ref{cor:k-image-full}, $K(\mathcal{X} \times \mathcal{X}) = [0,1]$. By Lemma~\ref{lem:intermediate-K}, $|\mathfrak{B}| > 1$ (otherwise $K$ takes only the basis values $\{0, 1\}$, contradicting full $K$-image). Theorem~\ref{thm:complexity-constraint} establishes $G$ as a compact Lie group acting transitively on $\mathfrak{B}$ (its proof uses the $K$-metric topology, Lemma~\ref{lem:intermediate-K}, and compactness from Lemma~\ref{lem:compactness-from-capacity}, not connectedness as a hypothesis). Theorem~\ref{thm:points-sections} identifies $\mathcal{X} \cong \mathbb{F}P^{N-1}$ for $\mathbb{F} \in \{\mathbb{R}, \mathbb{C}, \mathbb{H}\}$. Each $\mathbb{F}P^{N-1}$ is path-connected (a homogeneous space of a connected compact Lie group: $O(N)/(O(N{-}1) \times O(1))$, $U(N)/(U(N{-}1) \times U(1))$, or $\mathrm{Sp}(N)/(\mathrm{Sp}(N{-}1) \times \mathrm{Sp}(1))$, all connected). Hence $\mathcal{X}$ is path-connected.
\end{proof}

Discrete-K theories at finite $M$ (Wootters phase-space~\cite{wootters1987wigner}, Galois QM~\cite{vourdas2004quantum}, MUB-based finite-dim QI, Spekkens' toy model~\cite{spekkens2007evidence}) violate Imperceptibility (their $K$-image has gaps in $[0,1]$) and lie outside the present scope; their continuum limit is $\mathbb{C}P^{N-1}$ as $M \to \infty$. Capacity $N$ is a scope label, not a closure parameter.

\begin{remark}[Excluded discrete-K models]
\label{rem:excluded-models}
Imperceptibility (Theorem~\ref{thm:src-master}(I)) excludes models in which $K(\mathcal{X} \times \mathcal{X})$ has gaps in $[0,1]$; equivalently, models with a fundamental $K$-resolution scale. \textit{(a)~Classical register.} $\mathcal{X} = \{1, \ldots, N\}$ with $K(i,j) = 1 - \delta_{ij}$ satisfies finite capacity with $G = S_N$, but its $K$-image is $\{0,1\}$, missing the entire interval $(0,1)$; the metric bridge cannot apply. \textit{(b)~Discrete refinement.} On $\{a,b,c,d\}$ with $K(c,d)=1/2$ and all other pairs $1$, capacity is 3, the $K$-image is $\{0, 1/2, 1\}$, again with gaps. Both lie within $\mathfrak{T}_{N,M}$ for finite $M$.
\end{remark}

\begin{theorem}[Permutation Invariance]
\label{thm:permutation-invariance}
For each basis $\mathcal{B} = \{b_0, \ldots, b_{N-1}\}$, every $\sigma \in S_N$ is realized by some element of $G = \mathrm{Aut}(\mathcal{X}, K)$. Equivalently, $S_N$ embeds into $G|_{\mathcal{B}}$.
\end{theorem}

\begin{proof}
Basis elements satisfy $K(b_i, b_j) = 1$ for $i \neq j$ and $K(b_i, b_i) = 0$. These values are invariant under every $\sigma \in S_N$: $K(b_{\sigma(i)}, b_{\sigma(j)}) = K(b_i, b_j)$ for all $i, j$. By Structural Leibniz (Theorem~\ref{thm:src-master}(S4)), each $\sigma$ extends to a $K$-preserving automorphism of $(\mathcal{X}, K)$.
\end{proof}

\begin{lemma}[Constitutive Necessity of Dynamics]
\label{lem:dynamics-constitutive}
Let $(\mathcal{X}, K)$ satisfy Axioms~\ref{ax:finite}--\ref{ax:relational} with capacity $N \geq 2$. Then $K$-preserving automorphisms are constitutive of the state space: without them, generic states have no well-defined relational identity.
\end{lemma}

\begin{proof}
By Completeness (Theorem~\ref{thm:src-master}(S2)), surjectivity forces states $x, x'$ with $K(x, b_k) = K(x', b_k)$ for all $b_k \in \mathcal{B}$ yet $K(x, x') > 0$ (otherwise the state space collapses to the $N$ basis elements). By Identity (Theorem~\ref{thm:src-master}(S1)), distinct $K$-profiles must differ on some $z \notin \mathcal{B}$, hence on evaluations against a different basis $\mathcal{B}' \neq \mathcal{B}$. By Operational Completeness (Theorem~\ref{thm:src-master}(O)), the only structure connecting $\mathcal{B}$ to $\mathcal{B}'$ is a $K$-preserving automorphism $g \in G$ with $g(\mathcal{B}) = \mathcal{B}'$. Therefore the $K$-profile of a generic state requires evaluations across the $G$-orbit of $\mathcal{B}$, making automorphisms constitutive of relational identity.
\end{proof}

\begin{remark}[Axiom Independence]
\label{rem:axiom-independence}
The two axioms and their derived clauses are mutually independent. Discrete-K theories at finite $M$ (Remark~\ref{rem:excluded-models}) satisfy finite capacity and the saturation hierarchy (Theorem~\ref{thm:src-master}(S1)--(S4)) without Imperceptibility (clause (I)), and may satisfy SRC's other consequences partially; the classical register ($M = 1$) satisfies finite capacity and Permutation Invariance but lacks the continuous K-symmetry needed for Basis Isotropy (clause (B)). Infinite-dimensional QM satisfies SRC and Imperceptibility without finite capacity (Axiom~\ref{ax:finite}). The framework classifies interpretations by which axioms or derived clauses they accept or reject: collapse theories modify the dynamics; many-worlds accepts all; hidden-variable theories reject the saturation hierarchy; discrete-K theories ($M < \infty$) reject Imperceptibility (see \S\ref{sec:discussion} for the comparison with Bohmian mechanics).
\end{remark}

\smallskip

\begin{remark}[Topology Provenance]
\label{rem:topology-provenance}
Continuity of $K$, compactness of $\mathcal{X}$, the Lie group structure of $G$, and the smooth manifold structure of $\mathcal{X}$ all derive from $K$ alone: $d(x,y) := \sup_z |K(x,z) - K(y,z)|$ is a genuine metric by (S1) and the triangle inequality $|K(x,y) - K(x',y')| \leq d(x,x') + d(y,y')$ gives joint continuity of $K$; compactness follows by (S1) + (S2) and Tychonoff (Lemma~\ref{lem:compactness-from-capacity}); $G$ acts transitively (B) and by isometry; finite covering dimension follows from compactness via the basis K-profile $\varphi: \mathcal{X} \to [0,1]^N$ with (S3)-orbit fibers; Bochner--Montgomery and Gleason--Yamabe then yield the Lie group and smooth manifold structure (Appendix~\ref{app:structural-conditions}). No continuity, smoothness, or compactness assumption is imported.
\end{remark}

\section{Consequences of Relationality}
\label{sec:consequences}

\noindent Sections~\ref{sec:consequences}--\ref{sec:measurement} derive the consequences of Axioms~\ref{ax:finite}--\ref{ax:relational}: parsimony, regularity, state space identification, field selection, dynamics, and Born rule. Each step uses only previously established results.

\subsection{Information Parsimony}

Hidden variables are excluded by the conjunction of Operational Completeness and Saturation.

\begin{definition}[Hidden Variable Extension]
\label{def:hv-extension}
A \emph{hidden variable extension} of $(\mathcal{X}, K)$ is a space $\mathcal{X}' = \mathcal{X} \times \Lambda$ with $K'((x, \lambda), (x', \lambda')) = K(x, x')$.
\end{definition}

\begin{theorem}[Parsimony from Completeness \\and Saturation]
\label{thm:parsimony-derived}
Under Axioms~\ref{ax:finite}--\ref{ax:relational}, every hidden variable extension is trivial: $|\Lambda| = 1$.
\end{theorem}

\begin{proof}
\textit{Step 1 (Completeness):} By construction, $K'((x,\lambda_1), (x,\lambda_2)) = 0$. By Operational Completeness, the $\Lambda$-coordinate is operationally undetectable.
\textit{Step 2 (Saturation):} If $|\Lambda| > 1$, the extended space contains structure not contributing to distinguishability, violating Saturation. Equivalently, the map $(x,\lambda_1) \mapsto (x,\lambda_2)$ preserves $K'$, so Transport Consistency identifies $K'$-equivalent states as physically identical, forcing $|\Lambda| = 1$.
\end{proof}

\begin{corollary}[Information Parsimony]
\label{cor:parsimony}
$(\mathcal{X}, K)$ admits no physical degree of freedom invisible to $K$-evaluations: every physical symmetry is a $K$-preserving automorphism, and every $K$-preserving automorphism is a physical symmetry. In particular, $G$ exhausts the symmetries of $(\mathcal{X}, K)$ at both the symmetry-group and state-space levels.
\end{corollary}

\textit{Scope:} The physical content is the decomposition of parsimony into independently motivated components (Completeness, Saturation, Transport).

\subsection{Continuous Symmetry from Finite Capacity}

Finite capacity forces \emph{continuous} symmetry, in the form of a Lie group rather than a discrete group.

\begin{lemma}[Intermediate Distinguishability Values]
\label{lem:intermediate-K}
If $|\mathfrak{B}| > 1$ (any system with non-trivial dynamics), then $K$ takes values in $(0,1)$, not just $\{0,1\}$.
\end{lemma}

\begin{proof}
Let $\mathcal{B}, \mathcal{B}'$ be distinct bases and $b'_1 \in \mathcal{B}'$. If $K(b'_1, b_k) = 1$ for all $k$, then $\{b'_1\} \cup \mathcal{B}$ has $N+1$ mutually distinguishable elements, contradicting capacity $N$. If $K(b'_1, b_k) = 0$ for some $k$, then $b'_1 = b_k$, forcing $\mathcal{B}' = \mathcal{B}$, contradiction. Therefore $K(b'_1, b_k) \in (0,1)$ for at least one $k$.
\end{proof}

\begin{theorem}[The Complexity Constraint on Symmetry]
\label{thm:complexity-constraint}
Under Axioms~\ref{ax:finite}--\ref{ax:relational}, the symmetry group $G$ of a finite-capacity system is a Lie group of positive dimension. In particular, no finite group (including the full symmetric group $S_N$) can serve as the symmetry group of a system satisfying both axioms.
\end{theorem}

\begin{proof}
Theorem~\ref{thm:src-master}(B) requires $G$ to act transitively on $\mathfrak{B}$; we show $\mathfrak{B}$ is uncountable, so $G$ cannot be finite.

\textit{Step 1: $K$ takes intermediate values.} If $\mathcal{B}, \mathcal{B}'$ are distinct bases and $b' \in \mathcal{B}'$, then $K(b', b_k) \in (0,1)$ for at least one $k$ (Lemma~\ref{lem:intermediate-K}).

\textit{Step 2: $K$ is continuous in the $K$-metric topology.} By Remark~\ref{rem:topology-provenance}, $\mathcal{X}$ carries the initial topology generated by the maps $x \mapsto K(x,z)$, $z \in \mathcal{X}$; equivalently, the metric topology from $d(x,y) := \sup_z |K(x,z) - K(y,z)|$ (a genuine metric by Identity (Theorem~\ref{thm:src-master}(S1))). Each $g \in G$ is a $d$-isometry. The estimate $|K(x,y) - K(x',y')| \leq d(x,x') + d(y,y')$ shows that $K$ is jointly continuous in this topology by construction; this suffices for the argument that follows. A detailed metrizability and compactness derivation is in Appendix~\ref{app:structural-conditions}: the full $K$-profile $\Sigma(x) := K(x, \cdot)$ embeds $\mathcal{X}$ into $[0,1]^{\mathcal{X}}$ with closed image (Tychonoff plus (S1), (S2)).

\textit{Step 3: Bases form continuous families.} Combining Steps 1 and 2: as $b'$ is varied along a small neighborhood in $\mathcal{X}$ (continuously), the values $K(b', b_k)$ vary continuously, and by Lemma~\ref{lem:intermediate-K} at least one ranges over a non-trivial interval. The set $\mathfrak{B}$ of bases is therefore at least one-parameter, hence uncountable. A finite group cannot act transitively on an uncountable set, so $G$ must be infinite.

\textit{Step 4: Lie group via no-small-subgroups.} By compactness (Appendix~\ref{app:structural-conditions}), $G$ is compact and acts effectively on $\mathcal{X}$ (any $g \in G$ fixing every point fixes every $K$-relation, hence is the identity). The basis K-profile $\varphi: \mathcal{X} \to [0,1]^N$ is continuous with compact image, and its fibers are $G_S$-orbits (Basis-Profile Symmetry, Theorem~\ref{thm:src-master}(S3), where $G_S$ is the $S$-pointwise stabilizer in $G$); compactness of $G_S$ together with continuity of the orbit map controls the fiber covering dimension, so $\mathcal{X}$ has finite covering dimension. The Bochner--Montgomery theorem on compact effective actions on finite-dimensional spaces \cite{montgomery1955topological} ensures $G$ has no small subgroups; the Gleason--Yamabe characterization then gives $G$ a Lie group. Locally-Euclidean structure of $\mathcal{X}$ is downstream of two-point homogeneity (Theorem~\ref{thm:points-sections}) and is not required at this step.
\end{proof}

Theorem~\ref{thm:complexity-constraint} establishes $G$ as a Lie group acting on $\mathcal{X}$. Frobenius's classification of finite-dimensional associative division algebras over $\R$ then narrows the coefficient field $\mathbb{F}$ to one of $\R, \C, \mathbb{H}$; the remaining task is to exclude $\R$ and $\mathbb{H}$. The latter is excluded by two independent obstructions:

\begin{theorem}[Quaternionic Obstruction]
\label{thm:quaternion-obstruction}
The coefficient field cannot be $\mathbb{H}$. Two independent obstructions:
\begin{enumerate}
\item[(i)] \textbf{Cyclic-Fourier realization.} The cyclic generator $\pi$ on a basis $\mathcal{B}$ has eigenvectors over $\overline{\mathbb{F}}$ given by the cyclic Fourier basis $\{f_k\}$ with eigenvalues $\omega^{-k}$, $\omega = e^{2\pi i/N}$. Realization of $\{f_k\}$ as a basis in $\mathfrak{B}$ (Saturation, via the realization argument of Lemma~\ref{lem:sheaf-complex}) requires $\omega \in \mathbb{F}$, equivalently a distinguished imaginary unit in $\mathbb{F}$. Over $\mathbb{H}$ the imaginary units form a 2-sphere with no canonical element, so no choice of $\omega \in \mathbb{H}$ is $K$-relational; the Fourier basis is not realized.
\item[(ii)] \textbf{Compositional.} Composition of independent systems requires the joint space to have capacity $N_A N_B$. The obstacle is algebraic: $\mathbb{H} \otimes_\R \mathbb{H} \cong M_4(\R)$ is not a division algebra, so $\mathbb{H}^{N_A} \otimes_\mathbb{H} \mathbb{H}^{N_B}$ does not carry a natural quaternionic Hilbert space structure of the required dimension. Equivalently $\dim_\R \mathrm{Herm}(\mathbb{H}^n) = n(2n-1) \neq n^2$, so product operators fail to span the joint Hermitian algebra. Quaternionic QM fails closure under composition.
\end{enumerate}
Each obstruction is individually sufficient and both are Parsimony-independent. $\C$ is therefore the minimal associative division algebra supporting both the cyclic Fourier realization and closure under composition.
\end{theorem}

\noindent\textit{Further remarks on Axiom 2 (SRC) and its derived clauses:}

\paragraph{Derived clauses of SRC.} Operational Completeness (Theorem~\ref{thm:src-master}(O)) identifies undetectable structure; Transport Consistency (Theorem~\ref{thm:src-master}(T)) ensures physical equivalence of $K$-equivalent states; Basis Isotropy (Theorem~\ref{thm:src-master}(B)), i.e., transitivity on $\mathfrak{B}$, carries the heaviest consequences (Lie group structure, flag manifold, unique metric). Permutation Invariance (symmetry of atom labels within each basis) follows from Structural Leibniz (Theorem~\ref{thm:src-master}(S4); Theorem~\ref{thm:permutation-invariance}).

\begin{remark}[Physical Content of Basis Isotropy]
\label{rem:basis-isotropy-physical}
In the ontology of \S\ref{sec:graded-equality}, Basis Isotropy follows from the exhaustiveness of $K$: all bases share the same internal $K$-structure, and no structure beyond $K$ exists to privilege one over another. Empirically, this holds for fundamental systems: any Stern-Gerlach axis, any polarizer angle, any Fourier-conjugate basis is equally physical. Systems with apparent anisotropy (crystals, preferred axes) embed in larger isotropic systems.
\end{remark}

\subsection{Closure and Time Emergence}

\begin{lemma}[Closure of the Relational System]
\label{prop:closure}
Under Axioms~\ref{ax:finite}--\ref{ax:relational}, there exists no external information reservoir. All information is encoded in relational distinctions within the system.
\end{lemma}

\begin{proof}
If an external reservoir $R$ with capacity $N_R \geq 2$ existed, the combined system would have capacity $N' > N$ (since additional distinguishable states become available). Then $R$ would itself be part of a larger system, contradicting its assumed externality. Either this iteration terminates (no external reservoir for the complete system) or capacity diverges, violating Axiom~\ref{ax:finite}. For any finite-capacity universe, there exists a maximal closed system. Information not encoded in its relational structure does not exist.
\end{proof}

\begin{theorem}[Continuous Time as Reconstructed Parameterization]
\label{thm:time-emergence}
Given a discrete $K$-preserving generator $\pi$ of order $N$, a continuous one-parameter group $U: \R \to G$ arises as the unique band-limited interpolation satisfying $U(n) = \pi^n$ for $n \in \Z$ (cf.\ \cite{page1983evolution} for time emergence from static relational structure). Band-limitedness follows from smoothness of the $G$-action: non-smooth interpolations would violate the Lie group structure of $G$. Band-limitedness is therefore not an independent regularity assumption: it inherits from the Lie group structure of $G$ (Theorem~\ref{thm:complexity-constraint}; cf.\ Remark~\ref{rem:topology-provenance}), which itself derives from finite capacity together with the saturation hierarchy.
\end{theorem}

\begin{proof}
The intra-basis dynamical symmetry is the cyclic group $\Z_N$ generated by $\pi$; we now show this extends to a continuous one-parameter group via band-limited interpolation.

\textit{Step 1: Discrete structure.}
The generator $\pi$ satisfies $\pi^N = \mathrm{id}$, giving the discrete sequence $x_n = \pi^n(x_0)$ for $n \in \Z_N$.

\textit{Step 2: Continuity from finite capacity.}
Any evolution consistent with finite capacity is continuous: a discontinuity would require infinite information to specify (the location of a jump on a continuum), violating the $\log_2 N$-bit storage budget. The discrete sequence $\{\pi^n\}$ therefore admits a continuous extension.

\textit{Step 3: Spectral interpolation.}
Let $T_N$ be the recurrence period. In the Hilbert space representation derived later, $\pi$ has spectral decomposition $\pi = \sum_k e^{2\pi ik/N} |v_k\rangle\langle v_k|$. Define $U(t) = \sum_k e^{2\pi ikt/T_N} |v_k\rangle\langle v_k|$. Then: (a) $U(nT_N/N) = \pi^n$ for all $n \in \Z_N$; (b) $U(t)U(s) = U(t+s)$ (one-parameter subgroup law holds by construction); (c) each eigenvalue trajectory $t \mapsto e^{2\pi ikt/T_N}$ is band-limited with bandwidth $\leq N/2$, so $U(t)$ is the unique band-limited interpolation of the discrete sequence $\{\pi^n\}$ (Whittaker-Shannon applied to each spectral component).

\textit{Step 4: Stone's theorem.}
The continuous one-parameter group $U(t)$ has a unique self-adjoint generator $H$ with $U(t) = e^{-iHt/\hbar}$.
\end{proof}

\begin{corollary}[Energy as Rate of Relational Update]
\label{cor:energy-rate}
Energy $E_k = 2\pi k \hbar/T_N = \hbar \cdot \omega_k$ is the rate of relational update, where $T_N = 2\pi\hbar/E$ is the recurrence period.
\end{corollary}


\begin{remark}[The state space as a Nyquist manifold]
\label{rem:nyquist-manifold}
$(\mathcal{X}, K)$ is a \emph{Nyquist manifold}: a smooth manifold whose $K$-profiles against a complementary pair of bases finitely determine each state (Theorem~\ref{thm:points-sections}, Step~2; the $2N$-tuple $\Sigma(x)$), the continuous geometry being the unique band-limited interpolation of finitely many $K$-values \cite{shannon1949communication}. The continuous description $\mathcal{X} \cong \C P^{N-1}$ and the discrete description via $2(N-1)$ parameters from this finite separating set are exactly equivalent. The per-shot resolution $\delta \sim 1/N$ and the apparatus-scale and cosmic precision floors are developed in Lemma~\ref{prop:capacity-bound}. (Cf.\ fuzzy manifolds \cite{madore1992fuzzy} and finite spectral triples \cite{connes1994noncommutative}, which truncate a given manifold; here both the manifold and its bandwidth are derived from $K$.)
\end{remark}

\begin{theorem}[Operational Indistinguishability of Continuous and Sampled Amplitudes]
\label{thm:continuous-sampled}
Let $(\mathcal{X}, K)$ have capacity $N$. No measurement by an apparatus with capacity $N' \leq N$ can distinguish states $\ket{\psi}, \ket{\psi'}$ with $|c_k - c'_k| < 1/N$: we have $K(\psi, \psi') \leq 2/N$ and extractable information is bounded by $\log_2 N$ bits (Holevo bound). The continuous manifold $\C P^{N-1}$ is the unique band-limited interpolation of $N$ samples (Whittaker-Shannon), containing no additional operational content.
\end{theorem}

\begin{proof}
For small differences: $|\braket{\psi}{\psi'}| \geq 1 - \tfrac{1}{2}\sum_k |c_k - c'_k|^2 \geq 1 - 1/(2N)$, so $K(\psi, \psi') \leq 1/N$. The state $\ket{\psi}$ is fully specified by $N$ amplitudes; any two interpolations agreeing on these are operationally identical. Consequently, a finite-$N$ system is operationally indistinguishable from an infinite-capacity system for any measurement with resolution $\delta > 1/N$.
\end{proof}

\begin{table}[h]
\centering
\footnotesize
\caption{\textbf{Axiom dependencies of the main derived results.} \checkmark: required; $\star$: via derived Parsimony. Dynamics denotes Theorem~\ref{thm:dynamics-derived}, itself derived from Theorem~\ref{thm:src-master}(S4).}
\label{tab:axiom-deps}
\begin{tabular}{lccc}
\hline
\textbf{Result} & \textbf{Ax.~1} & \textbf{Ax.~2} & \textbf{Dyn.} \\
\hline
Perm.\ Invariance (Thm.~\ref{thm:permutation-invariance}; $B_1 = B_2$) & \checkmark & \checkmark(B) & \\
Parsimony (Thm.~\ref{thm:parsimony-derived}) & \checkmark & \checkmark & \\
Continuity/Lie (Thm.~\ref{thm:complexity-constraint}) & \checkmark & \checkmark(B) & \\
State space $\C P^{N-1}$ (Thm.~\ref{thm:points-sections}) & \checkmark & \checkmark & \\
Complex field (Lem.~\ref{lem:sheaf-complex}, Thm.~\ref{thm:frobenius}) & \checkmark & \checkmark & \checkmark \\
Born rule (Thm.~\ref{thm:born-kernel}) & \checkmark$^\star$ & \checkmark$^\star$ & \\
Capacity deficit (Thm.~\ref{thm:capacity-halting}) & \checkmark$^\star$ & \checkmark$^\star$ & \\
Tensor product (Thm.~\ref{thm:tensor}) & \checkmark & \checkmark & \\
Schr\"odinger eq.\ (Thm.~\ref{thm:schrodinger}) & & \checkmark & \checkmark \\
\hline
\end{tabular}
\end{table}

\section{Emergence of Complex Structure}
\label{sec:complex-structure}

\subsection{State Space as Projective Space}

The central result of this subsection is that the state space $\mathcal{X}$ is a projective space $\mathbb{F}P^{N-1}$ over a division algebra $\mathbb{F}$. The proof constructs an explicit equivariant homeomorphism using $K$-evaluations against two complementary bases. A sheaf-theoretic reformulation of the same result (Lemma~\ref{lem:signature-sheaf}) exposes the gauge structure that drives \S\ref{sec:gauge}.

\begin{theorem}[State Space Identification]
\label{thm:points-sections}
The state space $\mathcal{X}$ is homeomorphic to $\mathbb{F}P^{N-1}$ for some associative division algebra $\mathbb{F} \in \{\R, \C, \mathbb{H}\}$.
\end{theorem}

\begin{proof}
We construct an explicit equivariant homeomorphism $\sigma: \mathcal{X} \to \mathbb{F}P^{N-1}$, where $\mathbb{F} \in \{\R, \C, \mathbb{H}\}$ is an associative division algebra (the selection $\mathbb{F} = \C$ is established below by cyclic-dynamics analysis).

\textit{Step 1: Coordinate map from a basis.}
Fix a basis $\mathcal{B} = \{e_0, \ldots, e_{N-1}\}$. For each $x \in \mathcal{X}$, define moduli $|c_k(x)|^2 := 1 - K(x, e_k) \in [0,1]$. By Theorem~\ref{thm:complexity-constraint}, $G$ is a compact Lie group, so $\mathcal{X} = G/H$ is a compact manifold. The moduli $|c_k|$ alone do not determine a point in $\mathbb{F}P^{N-1}$, since relative phases between components are invisible to a single basis.

\textit{Step 2: Phase recovery from a complementary basis.}
The full $K$-profile separates states by Identity (Theorem~\ref{thm:src-master}(S1)). A finite separating set adapted to the projective phase geometry is supplied by a complementary basis $\mathcal{B}' = \{e'_0, \ldots, e'_{N-1}\}$ with $K(e_j, e'_k) \in (0,1)$ for all $j,k$ (such bases exist by Lemma~\ref{lem:intermediate-K} and transitivity of $G$ on $\mathfrak{B}$). The $2N$-tuple
\[
\Sigma(x) \;=\; \big(K(x, e_k),\; K(x, e'_k)\big)_{k=0}^{N-1} \;\in\; [0,1]^{2N}
\]
determines $x$ uniquely once the field $\mathbb{F}$ and inner product structure are fixed (Steps below) and is $G$-equivariant (Theorem~\ref{thm:complexity-constraint}).

\textit{Structural identification.} $\mathcal{X} = G/H$ is a compact homogeneous Riemannian manifold (Theorem~\ref{thm:complexity-constraint}, Lemma~\ref{lem:compactness-from-capacity}).

\emph{Two-point homogeneity.} For pairs $(x_1, y_1), (x_2, y_2)$ with $K(x_1, y_1) = K(x_2, y_2) = c$, we exhibit $g \in G$ with $g(x_1) = x_2$, $g(y_1) = y_2$: (i) $G$ acts transitively on $\mathcal{X}$ (Basis Isotropy plus Permutation Invariance); (ii) the $K$-level set $L_c(x) := \{y : K(x, y) = c\}$ is a single $G_x$-orbit. \emph{Proof of (ii):} non-emptiness follows from Corollary~\ref{cor:k-image-full}. For transitivity, given $y_1, y_2 \in L_c(x)$, consider the involution $\tau$ on $\{x, y_1, y_2\}$ swapping $y_1 \leftrightarrow y_2$ and fixing $x$. The only $K$-values $\tau$ acts on are the cross-pairs $K(x, y_i)$, both equal to $c$ since $y_1, y_2 \in L_c(x)$, and the within-pair $K(y_1, y_2)$, which $\tau$ maps to itself by symmetry of $K$. The value $K(y_1, y_2)$ itself plays no role: $\tau$ exchanges the pair with itself, so it survives any value $K(y_1, y_2) \in [0, 1]$, including $K(y_1, y_2) \neq 0$. By Structural Leibniz (Theorem~\ref{thm:src-master}(S4)), $\tau$ extends to a global $K$-preserving automorphism $g \in G$ with $g(x) = x$ (so $g \in G_x$) and $g(y_1) = y_2$. Hence $G_x$ acts transitively on $L_c(x)$; (iii) compose $g_1$ taking $x_1 \to x_2$ with $g_2 \in G_{x_2}$ moving $g_1(y_1) \to y_2$.

By the Cartan classification of compact two-point homogeneous (rank-one symmetric) manifolds \cite{kobayashi1969foundations,wang1952twopoint}, the options are $S^n, \R P^n, \C P^n, \mathbb{H} P^n, \mathrm{OP}^2$. Spheres are excluded by $N \geq 3$ mutually $K = 1$ points; $\mathrm{OP}^2$ by non-associativity. Hence $\mathcal{X} \cong \mathbb{F}P^{N-1}$ for $\mathbb{F} \in \{\R, \C, \mathbb{H}\}$; the selection $\mathbb{F} = \C$ is established below by cyclic-dynamics analysis.

\textit{Coordinate form of $\sigma$.} Once $\mathcal{X} \cong \mathbb{F}P^{N-1}$ is established, $K$ pulls back to a $G$-invariant function on $\mathbb{F}P^{N-1}$. By two-point homogeneity it depends only on $|\langle\psi|\phi\rangle|^2$, with the boundary values $K(b_i, b_j) = 1 - \delta_{ij}$ fixing $K(\psi, \phi) = 1 - |\langle\psi|\phi\rangle|^2$. The projective coordinate map $\sigma(x) = [c_0(x) : \cdots : c_{N-1}(x)]$ with $|c_k|^2 = 1 - K(x, e_k)$ and unit norm $\sum_k |c_k|^2 = 1$; relative phases are recovered from any complementary basis $\mathcal{B}'$ via $K(x, e'_k) = 1 - |\sum_j \bar{t}_{jk} c_j|^2$ with $t_{jk} = \langle e_j | e'_k\rangle$. Injectivity needs only that $\mathcal{B}'$ be $K$-non-degenerate against $\mathcal{B}$ (Lemma~\ref{lem:intermediate-K}); the specific uniform-overlap value $|t_{jk}|^2 = 1/N$ arises later in the derived $\mathbb{F} = \mathbb{C}$ case (Lemma~\ref{lem:sheaf-complex}).

\textit{Step 3: Injectivity.}
If $\sigma(x) = \sigma(y)$, then within the realization $\mathcal{X} \cong \mathbb{F}P^{N-1}$ (Structural Identification above) the projective coordinates $[c_0(x): \cdots : c_{N-1}(x)]$ and $[c_0(y) : \cdots : c_{N-1}(y)]$ coincide, so $x = y$ as projective rays. Equivalently, $K(x, z) = K(y, z)$ holds for all $z \in \mathcal{X}$ on the realization (the Fubini--Study form is a function of the projective coordinates), and Identity (Theorem~\ref{thm:src-master}(S1)) then gives $x = y$ at the abstract level.

\textit{Step 4: $G$-equivariance and surjectivity.}
Since $K(g \cdot x, g \cdot e_k) = K(x, e_k)$ for all $g \in G$, coordinates transform covariantly: $\sigma(g \cdot x) = g \cdot \sigma(x)$. The image $\sigma(\mathcal{X}) \subset \mathbb{F}P^{N-1}$ is compact and $G$-invariant, containing all basis vectors. By Completeness (Theorem~\ref{thm:src-master}(S2)), a point in $\mathbb{F}P^{N-1} \setminus \sigma(\mathcal{X})$ would define a $K$-profile not represented in $\mathcal{X}$, violating completeness.

\textit{Step 5: Homeomorphism.}
$\sigma$ is a continuous bijection between compact Hausdorff spaces, hence a homeomorphism. By Frobenius, $\mathbb{F} \in \{\R, \C, \mathbb{H}\}$.
\end{proof}

\begin{lemma}[Signature sheaf and global sections]
\label{lem:signature-sheaf}
Theorem~\ref{thm:points-sections} admits a sheaf-theoretic reformulation that exposes the gauge structure used in \S\ref{sec:gauge}. Define a sheaf $\mathcal{S}$ over the basis manifold $\mathfrak{B}$ with stalks $\mathcal{S}_{\mathcal{B}} := \mathbb{F}P^{N-1}$ (normalized coefficient vectors modulo overall $\mathbb{F}^\times$ phase). A state $x \in \mathcal{X}$ assigns to each basis $\mathcal{B}$ a signature $\sigma_{\mathcal{B}}(x) \in \mathcal{S}_{\mathcal{B}}$ with $|c_k(x)|^2 = 1 - K(x, b_k)$ and relative phases fixed by evaluations against a complementary basis (Theorem~\ref{thm:points-sections}, Step~2). For $\mathcal{B}, \mathcal{B}'$ related by $g \in G$ with $g(\mathcal{B}) = \mathcal{B}'$, the transition map $\phi_g: \mathcal{S}_{\mathcal{B}} \to \mathcal{S}_{\mathcal{B}'}$ is the projective action of $g$ on coefficients; transition maps satisfy the cocycle condition $\phi_{gh, \mathcal{B}} = \phi_{g, h\mathcal{B}} \circ \phi_{h, \mathcal{B}}$ (group property of $G$). Points of $\mathcal{X}$ correspond bijectively to global sections: $\mathcal{X} \cong \Gamma(\mathfrak{B}, \mathcal{S})$, and the transition maps classify a line bundle $\mathcal{L} \to \mathfrak{B}$ \cite{abramsky2011sheaf}; for $\mathbb{F} = \C$, Saturation selects the fundamental Borel--Weil representation, giving $\mathcal{L} = \mathcal{O}(1)$ (tautological) with $c_1(\mathcal{O}(1)) = 1 \neq 0$.
\end{lemma}

\begin{proof}
The stalk, signature, and transition-map definitions package the coordinate construction of Theorem~\ref{thm:points-sections} (Steps 1--2), which establishes the $\mathcal{X} \cong \Gamma(\mathfrak{B}, \mathcal{S})$ correspondence. Transition unitarity inherits from $G \subseteq \mathrm{Aut}(\mathcal{X}, K)$; the cocycle condition is the group property of $G$. The line bundle classification is standard sheaf cohomology over $\mathfrak{B} \cong U(N)/T^N$; Saturation pins down the fundamental representation.
\end{proof}

\begin{lemma}[Cyclic dynamics force $\C$ for $N \geq 3$]
\label{lem:sheaf-complex}
Let $(\mathcal{X}, K)$ satisfy Axioms~\ref{ax:finite}--\ref{ax:relational} with $N \geq 3$. Then $\mathbb{F} = \C$. (See also \cite{goyal2010origin} for an independent derivation of complex amplitudes from complementary symmetry/information primitives.)
\end{lemma}

\begin{proof}
By Theorem~\ref{thm:points-sections}, $\mathcal{X} \cong \mathbb{F}P^{N-1}$ for some $\mathbb{F} \in \{\R, \C, \mathbb{H}\}$. By Permutation Invariance (Theorem~\ref{thm:permutation-invariance}), the cyclic permutation of $S_N$ acts on a basis $\mathcal{B} = \{e_0, \ldots, e_{N-1}\}$ via $\pi \in \mathrm{Aut}(\mathcal{X}, K)$ with $\pi^N = \mathrm{id}$ and $\pi(e_j) = e_{j+1 \bmod N}$.

\emph{(a) Geometric construction of the uniform-amplitude state.} Working within the realization $\mathcal{X} \cong \mathbb{F}P^{N-1}$ already supplied by Theorem~\ref{thm:points-sections}, define
\[
f_0 \;:=\; N^{-1/2} \sum_{j=0}^{N-1} e_j \;\in\; \mathbb{F}^N.
\]
The coefficients $N^{-1/2}$ lie in $\R \subset \mathbb{F}$ for any $\mathbb{F} \in \{\R, \C, \mathbb{H}\}$, so $f_0$ is a unit vector in $\mathbb{F}^N$ for any of the three candidate fields. The Fubini--Study $K$-form on $\mathbb{F}P^{N-1}$ then gives
\[
K(f_0, e_j) \;=\; 1 - |\langle f_0, e_j\rangle|^2 \;=\; 1 - 1/N \quad \text{for all } j.
\]
The uniform-amplitude state $f_0 \in \mathcal{X}$ thus exists in each candidate $\mathbb{F}P^{N-1}$ by direct construction. This step does \emph{not} discriminate among $\R, \C, \mathbb{H}$; the field selection is forced in steps~(b)--(d), where the cyclic Fourier basis built on top of $f_0$ requires $\omega \in \mathbb{F}$.

\emph{(b) Identification of the cyclic Fourier basis over $\overline{\mathbb{F}}$.} Let $\omega := e^{2\pi i/N}$. The eigenvectors of $\pi$ over the algebraic closure $\overline{\mathbb{F}}$ are
\[
f_k \;:=\; N^{-1/2} \sum_{j=0}^{N-1} \omega^{jk}\, e_j \qquad (k = 0, \ldots, N-1),
\]
satisfying $\pi f_k = \omega^{-k} f_k$, mutual $\overline{\mathbb{F}}$-orthonormality, and uniform $K$-relation $K(e_j, f_k) = 1 - 1/N$ to $\mathcal{B}$. For $f_k$ to lie in $\mathbb{F}^N$ rather than just $\overline{\mathbb{F}}^N$, the Fourier coefficients $\omega^{jk}$ must lie in $\mathbb{F}$, equivalently $\omega \in \mathbb{F}$.

\emph{(c) Realization requirement.} The uniform-overlap profile $K(\cdot, e_k) = 1 - 1/N$ together with mutual orthogonality is $K$-consistent: $\overline{\mathbb{F}} \supseteq \mathbb{C}$ for each $\mathbb{F} \in \{\R, \C, \mathbb{H}\}$, so $\omega \in \overline{\mathbb{F}}$ in every candidate field, and the projection vectors $f_k$ over $\overline{\mathbb{F}}$ assemble into an MUB partner of $\mathcal{B}$ realising the profile. By Completeness (Theorem~\ref{thm:src-master}(S2)), each such profile is realised in $\mathcal{X}$, and Permutation Invariance forces the realising states to assemble into a single $\pi$-equivariant basis (the orbit of one Fourier eigenstate under $\pi$ generates the rest). Modulus uniformity $|\langle e_j, f_k\rangle|^2 = 1/N$ together with $\pi$'s cyclic action on $\mathcal{B}$ forces the $f_k$ to be simultaneous $\pi$-eigenvectors in $\mathrm{span}_{\overline{\mathbb{F}}}(\mathcal{B})$, so the realised basis coincides up to phase with $\{f_k\}$. Hence $\{f_k\} \subset \mathcal{X}$, and by step~(b) the Fourier coefficients $\omega^{jk}$ must lie in $\mathbb{F}$, i.e., $\omega \in \mathbb{F}$. The obstruction at $\mathbb{F} = \R$ is then the direct one of step~(d): $\omega = e^{2\pi i/N}$ is non-real for $N \geq 3$, so $\mathbb{F} \supsetneq \R$.

\emph{(d) Direct obstruction at $\mathbb{F} = \R$.} For $\mathbb{F} = \R$ and odd $N \geq 3$, (c) fails by parity: a real basis with $|\langle e_j, f_k\rangle|^2 = 1/N$ uniformly forces $\langle e_j, f_k\rangle = \pm 1/\sqrt{N}$, hence $f_k = (s^k_0, \ldots, s^k_{N-1})/\sqrt{N}$ with $s^k_j \in \{\pm 1\}$. Orthogonality $\sum_j s^k_j s^l_j = 0$ requires a sum of $N$ values each $\pm 1$ to vanish, impossible since the parity matches $N$'s. For even $N \geq 4$, the diagonalizability obstruction of (b) suffices: $\omega = e^{2\pi i/N}$ is non-real for all $N \geq 3$. Hence $\mathbb{F} \supsetneq \R$ for $N \geq 3$. By Frobenius's classification combined with Theorem~\ref{thm:quaternion-obstruction}, $\mathbb{F} = \C$.
\end{proof}

\begin{theorem}[Cyclic Dynamics from Finite Graded Equality]
\label{thm:dynamics-derived}
Let $(\mathcal{X}, K)$ satisfy Axioms~\ref{ax:finite}--\ref{ax:relational}
(Finite Capacity; Self-Referential Consistency) with capacity
$N \geq 2$. Then there exists $\pi \in \mathrm{Aut}(\mathcal{X}, K)$ with
$\pi^N = \mathrm{id}$, $\pi \neq \mathrm{id}$, and $\pi$ acting
as a cyclic $N$-permutation on some basis. The companion field-selection statement (for $N \geq 3$, $\mathbb{F} = \C$) is Lemma~\ref{lem:sheaf-complex}; together they determine the form of the dynamical structure (cyclic evolution of order $N$, complex coefficients, and, after band-limited interpolation, continuous unitary time). Permutation Invariance (Theorem~\ref{thm:permutation-invariance}) supplies the underlying $S_N$-symmetry.
\end{theorem}

\begin{proof}
\textit{Step 1: State space identification.}
By Theorem~\ref{thm:points-sections}, the state space is
$\mathcal{X} \cong \mathbb{F}P^{N-1}$ for some associative division
algebra $\mathbb{F} \in \{\R, \C, \mathbb{H}\}$. The symmetry group
$G = \mathrm{Aut}(\mathcal{X}, K)$ acts on $\mathcal{X}$ via the
projective representation on $V = \mathbb{F}^N$.

\textit{Step 2: Cyclic permutation exists in $G$.}
By Theorem~\ref{thm:permutation-invariance} (Permutation Invariance), $S_N$ acts
on the atoms of each basis via elements of $G$. In particular, the
cyclic permutation $\sigma = (0\;1\;2\;\cdots\;N{-}1) \in S_N$
corresponds to some $\pi \in G$.

\textit{Step 3: $\pi$ has order $N$ on $\mathcal{X}$.}
The lift of $\pi$ to $V = \mathbb{F}^N$ is the permutation matrix
$P$ with $P_{jk} = \delta_{j,\,k+1\!\bmod N}$. Since $P^N = I_N$
(the identity in $GL_N(\mathbb{F})$), the projection to
$\mathbb{F}P^{N-1}$ gives $\pi^N = \mathrm{id}$. For
$0 < k < N$: $P^k e_0 = e_k \neq e_0$ in $\mathbb{F}P^{N-1}$, so
$\pi^k \neq \mathrm{id}$. Therefore $\mathrm{ord}(\pi) = N$. Field selection ($\mathbb{F} = \C$ for $N \geq 3$) is Lemma~\ref{lem:sheaf-complex}.
\end{proof}

\begin{remark}[Dynamical Rigidity]
\label{rem:axiom3-status}
The formal derivation proceeds via representation theory on the derived manifold $\mathbb{F}P^{N-1}$: Permutation Invariance (Theorem~\ref{thm:permutation-invariance}) supplies the raw material ($S_N$-symmetry within each basis); the theorem's content is a \emph{rigidity result}: this symmetry uniquely determines cyclic dynamics of order $N$, forces $\mathbb{F} = \C$ for $N \geq 3$, and extends to continuous unitary evolution via band-limited interpolation.
\end{remark}

\begin{theorem}[Frobenius Classification: $\C$ is Unique]
\label{thm:frobenius}
The field $\C$ is the unique coefficient field consistent with our axioms. Frobenius's classification \cite{frobenius1878lineare} lists the finite-dimensional associative division algebras over $\R$ as $\R, \C, \mathbb{H}$; $\R$ is excluded by Lemma~\ref{lem:sheaf-complex}, $\mathbb{H}$ by Theorem~\ref{thm:quaternion-obstruction}, and $\mathbb{O}$ by non-associativity. (Division-algebra structure is required: non-zero states need inverse scalar multiplication, and a zero divisor would annihilate states.)
\end{theorem}

\begin{corollary}[Basis space and full symmetry group]
\label{thm:basis-topology}
With $\mathbb{F} = \C$, the basis space is the complete flag manifold $\mathfrak{B} \cong U(N)/(U(1))^N$ (an orthonormal basis is $U \in U(N)$ modulo the diagonal torus $\mathrm{diag}(e^{i\theta_1},\ldots,e^{i\theta_N})$); $G = \mathrm{Aut}(\mathcal{X}, K) \cong U(N)$ acts on $\C P^{N-1}$ by the standard representation (the \emph{inter-basis} symmetry, strictly larger than the intra-basis cyclic group $\Z_N$ generated by $\pi$). The tautological line bundle over $\mathfrak{B}$ has $c_1 \neq 0$; its Berry-curvature holonomy is $U(1)$-valued, confirming $\C$ over $\R$.
\end{corollary}

\begin{figure*}[t]
\centering
\begin{tikzpicture}[scale=1.0]

\begin{scope}[shift={(-3.5,0)}]
  \coordinate (V0) at (-1.4, 0);
  \coordinate (V1) at (1.4, 0);
  \coordinate (V2) at (0, 2.4);
  
  \fill[navyblue!6] (V0) -- (V1) -- (V2) -- cycle;
  
  \draw[navyblue!60, thick] (V0) -- (V1) -- (V2) -- cycle;
  
  \fill[navyblue] (V0) circle (3pt);
  \fill[navyblue] (V1) circle (3pt);
  \fill[navyblue] (V2) circle (3pt);
  
  \node[navyblue, font=\footnotesize\bfseries] at (-1.7, -0.12) {$\ket{0}$};
  \node[navyblue, font=\footnotesize\bfseries] at (1.7, -0.12) {$\ket{1}$};
  \node[navyblue, font=\footnotesize\bfseries] at (0, 2.65) {$\ket{2}$};
  
  \coordinate (psi) at (0, 0.85);
  \fill[coral] (psi) circle (3pt);
  \node[coral, font=\footnotesize\bfseries] at (0.38, 0.85) {$\ket{\psi}$};
  
  \draw[gray!25, thin, dashed] (psi) -- (V0);
  \draw[gray!25, thin, dashed] (psi) -- (V1);
  \draw[gray!25, thin, dashed] (psi) -- (V2);
  
  \node[gray!55, font=\scriptsize] at (-0.8, 0.3) {$|c_0|^2$};
  \node[gray!55, font=\scriptsize] at (0.8, 0.3) {$|c_1|^2$};
  \node[gray!55, font=\scriptsize] at (0, 1.5) {$|c_2|^2$};
  
  \node[font=\footnotesize\bfseries, navyblue] at (0, -0.6) {(a) Probability Simplex};
\end{scope}

\begin{scope}[shift={(3.5,0)}]
  \coordinate (B0) at (-1.4, 0);
  \coordinate (B1) at (1.4, 0);
  \coordinate (B2) at (0, 2.4);
  
  \fill[navyblue!6] (B0) -- (B1) -- (B2) -- cycle;
  
  \draw[navyblue!60, thick] (B0) -- (B1) -- (B2) -- cycle;
  
  \fill[navyblue] (B0) circle (3pt);
  \fill[navyblue] (B1) circle (3pt);
  \fill[navyblue] (B2) circle (3pt);
  
  \node[navyblue, font=\footnotesize\bfseries] at (-1.7, -0.12) {$\ket{j}$};
  \node[navyblue, font=\footnotesize\bfseries] at (1.7, -0.12) {$\ket{k}$};
  \node[navyblue, font=\footnotesize\bfseries] at (0, 2.65) {$\ket{\ell}$};
  
  \draw[gold!60!black, thick, densely dashed] (B1) -- (B2);

  \coordinate (state) at (-0.2, 0.95);
  \fill[coral] (state) circle (3pt);
  \node[coral, font=\footnotesize\bfseries] at (-0.55, 1.05) {$\ket{\psi}$};

  \coordinate (foot) at ($(B1)!(state)!(B2)$);

  \draw[gold!80!black, line width=1.2pt] (state) -- (foot);
  \fill[gold!80!black] (foot) circle (1.5pt);

  \coordinate (ra1) at ($(foot)!0.13!(B2)$);
  \coordinate (ra3) at ($(foot)!0.13!(state)$);
  \coordinate (ra2) at ($(ra1)+(ra3)-(foot)$);
  \draw[gold!80!black, line width=0.6pt] (ra1) -- (ra2) -- (ra3);

  \node[gold!80!black, font=\scriptsize\bfseries, anchor=east] at ($(state)!0.5!(foot)+(-0.08,0)$) {$h_j$};
  
  \node[draw=coral!60, rounded corners=2pt, fill=coral!6, 
        font=\footnotesize\bfseries, inner sep=4pt] at (0, -0.6) {
    $p_j = |c_j|^2 = h_j$
  };
  
  \node[font=\footnotesize\bfseries, navyblue] at (0, -1.28) {(b) Born Rule Geometry};
\end{scope}

\end{tikzpicture}
\caption{\textbf{State space geometry for $N=3$.}
\textbf{(a)} The probability simplex: vertices are basis states $\ket{0}, \ket{1}, \ket{2}$; a state $\ket{\psi}$ has barycentric coordinates $(|c_0|^2, |c_1|^2, |c_2|^2)$ satisfying $\sum_k |c_k|^2 = 1$.
\textbf{(b)} The Born rule from geometry: probability $p_j$ equals the perpendicular height from $\ket{\psi}$ to the face opposite vertex $\ket{j}$.}
\label{fig:signature-sheaf}
\end{figure*}

\begin{figure}[ht]
\centering
\begin{tikzpicture}[scale=0.9, every node/.style={transform shape}]

\begin{scope}[shift={(-2.5,0)}]
  \draw[navyblue, thick] (0,0) circle (1.2);
  \draw[navyblue!40, dashed] (0,0) ellipse (1.2 and 0.35);
  
  \node[navyblue, font=\small\bfseries] at (0, -1.8) {State Space};
  \node[navyblue, font=\scriptsize] at (0, -2.15) {$(\C P^{N-1}, g_{FS})$};
  
  \fill[coral] (0.3, 0.8) circle (2.5pt);
  \node[coral, font=\scriptsize] at (0.65, 0.75) {$\ket{\psi}$};
  \fill[coral] (-0.4, 0.75) circle (2.5pt);
  \node[coral, font=\scriptsize] at (-0.75, 0.62) {$\ket{\phi}$};
  
  \draw[coral, thick, dashed] (0.3, 0.8) arc (60:120:0.7);
  \node[coral, font=\tiny] at (0, 1.05) {$d_{FS}$};
\end{scope}

\draw[-{Stealth[length=3mm, width=2mm]}, thick, gold!80!black] (-0.8, 0) -- (0.8, 0);
\node[gold!80!black, font=\small\bfseries, align=center] at (0, 0.5) {Born Rule};
\node[gold!80!black, font=\scriptsize, align=center] at (0, -0.35) {$p_k = |c_k|^2$};
\node[gold!80!black, font=\tiny, align=center] at (0, -0.7) {\textbf{Isometry}};

\begin{scope}[shift={(2.5,0)}]
  \coordinate (A) at (-1.0, -0.7);
  \coordinate (B) at (1.0, -0.7);
  \coordinate (C) at (0, 1.0);
  
  \fill[tealblue!15] (A) -- (B) -- (C) -- cycle;
  \draw[tealblue, thick] (A) -- (B) -- (C) -- cycle;
  
  \fill[tealblue] (A) circle (2pt);
  \fill[tealblue] (B) circle (2pt);
  \fill[tealblue] (C) circle (2pt);
  
  \node[tealblue, font=\small\bfseries] at (0, -1.5) {Probability Space};
  \node[tealblue, font=\scriptsize] at (0, -1.85) {$(\Delta^{N-1}, g_{FR})$};
  
  \fill[coral] (0.2, 0.2) circle (2.5pt);
  \node[coral, font=\scriptsize] at (0.01, 0.45) {$p(\psi)$};
  \fill[coral] (-0.3, 0.0) circle (2.5pt);
  \node[coral, font=\scriptsize] at (-0.4, -0.3) {$p(\phi)$};
  
  \draw[coral, thick, dashed] (0.2, 0.2) -- (-0.3, 0.0);
  \node[coral, font=\tiny] at (0.2, -0.15) {$d_{FR}$};
\end{scope}

\node[draw=gold!60, rounded corners=3pt, fill=gold!8, 
      font=\small\bfseries, inner sep=6pt] at (0, -2.8) {
  $d_{FS}(\psi, \phi) = d_{FR}(p(\psi), p(\phi))$ \quad iff \quad $p_k = |c_k|^2$
};

\end{tikzpicture}
\caption{\textbf{Born rule as information isometry.} The Born rule $p_k = |c_k|^2$ is the unique assignment matching quantum distinguishability ($d_{FS}$) to statistical distinguishability ($d_{FR}$); any other choice maintains two independent distance structures, forbidden by $U(N)$-invariance and the uniqueness of the invariant metric on $\C P^{N-1}$.}
\label{fig:information-isometry}
\end{figure}

\subsection{Hilbert Space Representation Theorem}

We now prove the central representation result: any convex operational system satisfying our axioms embeds into a complex Hilbert space.

\begin{theorem}[Hilbert Space Representation]
\label{thm:hilbert-representation}
Let $(\mathcal{X}, K)$ satisfy Axioms~\ref{ax:finite}--\ref{ax:relational} with capacity $N \geq 3$, and let $\mathcal{S} = \mathrm{conv}(\mathcal{X})$ be the convex hull of pure states (the mixed-state space). Then there exists a complex Hilbert space $\mathcal{H}$ with $\dim(\mathcal{H}) = N$ and an embedding $\iota: \mathcal{S} \to \mathcal{D}(\mathcal{H})$ into the density operators such that:
\begin{enumerate}
\item[(i)] Pure states map to rank-1 projectors: $\iota(\mathcal{X}) = \{|\psi\rangle\langle\psi| : \|\psi\| = 1\}$;
\item[(ii)] $K$-derived effects $e_b(\psi) := 1 - K(\psi, b)$ map to positive operators in $[0, I]$;
\item[(iii)] Probabilities are preserved: $e_b(\psi) = \mathrm{Tr}(\iota^*(e_b) \cdot \iota(|\psi\rangle\langle\psi|))$.
\end{enumerate}
\end{theorem}

\begin{proof}
\textit{Step 1: Field selection.}
By Lemma~\ref{lem:sheaf-complex} (or equivalently Theorem~\ref{thm:frobenius} together with Theorem~\ref{thm:quaternion-obstruction}), for $N \geq 3$ the coefficient field is $\mathbb{F} = \C$: real Hilbert spaces cannot support continuous evolution connecting the identity to a cyclic permutation, and quaternionic spaces fail both the canonical spectral decomposition of $\pi$ and closure under composition. By Theorem~\ref{thm:points-sections} this fixes $\mathcal{X} \cong \C P^{N-1}$ with $K(\psi,\phi) = 1 - |\langle\psi|\phi\rangle|^2$.

\textit{Step 2: The extreme boundary is a single orbit.}
By Basis Isotropy (Theorem~\ref{thm:src-master}(B)), the symmetry group $G$ acts transitively on pure states. The orbit of any pure state under $G$ is therefore the full extreme boundary $\mathcal{X} \cong \C P^{N-1}$. The state space $\mathcal{S}$ is the convex hull of this orbit (for $N = 2$, a ball; for $N \geq 3$, the full density matrix set $\{\rho \geq 0, \mathrm{Tr}\,\rho = 1\}$).

\textit{Step 3: Projective unitary representation.}
$K$-preserving symmetries preserve the canonical overlap $|\langle\psi|\phi\rangle|^2 = 1 - K(\psi,\phi)$ on $\C P^{N-1}$. By Wigner's theorem \cite{wigner1931gruppentheorie,bargmann1964note}, each such symmetry is implemented by a unitary or antiunitary operator. Anti-unitaries form a disconnected component; since $G$ is a connected Lie group (Theorem~\ref{thm:complexity-constraint}) containing the identity (unitary), all elements of $G$ are unitary. Thus $G$ admits a projective representation $\rho: G \to PU(\mathcal{H})$ for some complex Hilbert space $\mathcal{H}$.

\textit{Step 4: Dimension from capacity.}
The representation $\rho$ maps each basis $\{s_1, \ldots, s_N\}$ to an orthonormal basis $\{|e_1\rangle, \ldots, |e_N\rangle\}$. Perfect distinguishability ($K(s_i, s_j) = 1$ for $i \neq j$) corresponds to orthogonality ($\langle e_i | e_j \rangle = 0$). Thus $\dim_{\C}(\mathcal{H}) = N$.

\textit{Step 5: State-operator correspondence.}
Define $\iota: \mathcal{S} \to \mathcal{D}(\mathcal{H})$ by mapping pure states to rank-1 projectors and extending affinely (Choquet's theorem on the simplex of mixed states; see \cite{alfsen2003geometry}).

\textit{Step 6: Effects and probabilities.}
The $K$-derived effect $e_b(\psi) := 1 - K(\psi, b) = |\langle\psi|b\rangle|^2$ corresponds to the rank-1 projector $|b\rangle\langle b|$, giving $e_b(\rho) = \mathrm{Tr}(|b\rangle\langle b| \rho)$ by affine extension. (Gleason's theorem \cite{gleason1957measures} extends this to the full projection lattice for $N \geq 3$, but only the $K$-effects $e_b$ are needed for (iii).)
\end{proof}

\section{Cyclic Symmetry, Gauge Structure, and the Inner Product}
\label{sec:hilbert}

\subsection{Cyclic Symmetry}
\label{sec:symmetry}

SRC together with Permutation Invariance (Theorem~\ref{thm:permutation-invariance}, derived from clause (S4)) makes atom labels within a basis conventional, which might suggest the full symmetric group $S_N$ as the symmetry group. However, the \emph{dynamical} symmetry reduces to the cyclic group $\Z_N$ for the following reasons.

\begin{lemma}[Rigidity of Cyclic Dynamics]
\label{lem:cyclic-rigidity}
The dynamical generator must be a single $N$-cycle. Among all $\sigma \in S_N$, the $N$-cycles are the unique permutations whose eigenbases are $K$-homogeneous with respect to $\mathcal{B}$: $K(b_i, f_j)$ is independent of $(i,j)$.
\end{lemma}

\begin{proof}
Let $\sigma$ have cycle structure $(n_1, \ldots, n_r)$ with orbits $O_1, \ldots, O_r$. The eigenvectors of the permutation matrix $P_\sigma$ from orbit $O_\ell$ are supported entirely on $\mathrm{span}\{|e_m\rangle : m \in O_\ell\}$, with equal-modulus components $1/\sqrt{n_\ell}$. Therefore:
\[
K(b_i, f_j) \;=\;
\begin{cases}
1 - 1/n_\ell & \text{if } b_i \in O_\ell \text{ (same orbit as } f_j\text{)}, \\
1 & \text{otherwise.}
\end{cases}
\]
For $K(b_i, f_j)$ to be independent of $(i,j)$, both values must coincide: $1 - 1/n_\ell = 1$ has no solution, so there can be no ``otherwise'' case. This forces $r = 1$: a single orbit of size $N$, i.e., $\sigma$ is an $N$-cycle, giving $K(b_i, f_j) = 1 - 1/N$ uniformly.

\textit{Why $K$-homogeneity is required.} A permutation with $r \geq 2$ orbits induces a partition of $\mathcal{B}$ into subsets $O_1, \ldots, O_r$ that is $K$-detectable via the eigenbasis: measuring $K(b_i, f_j)$ reveals which orbit $b_i$ belongs to. But Permutation Invariance (Theorem~\ref{thm:permutation-invariance}) makes all atom labels within a basis conventional; no subset of atoms is privileged. A $K$-detectable partition would introduce structure beyond $K$ (the orbit labeling), violating saturation (Theorem~\ref{thm:src-master}(S1)--(S4)).
\end{proof}

\begin{lemma}[Minimal Representation]
\label{lem:minimal-rep}
The representation of $G$ on the state space is \emph{irreducible}: there exists no proper $G$-invariant subspace.
\end{lemma}

\begin{proof}
Suppose $V \subset \Hilb$ is a proper $G$-invariant subspace. By Basis Isotropy (Theorem~\ref{thm:src-master}(B)), $G$ acts transitively on atoms. If any atom $a_0 \in V$, then all atoms are in $V$ (by transitivity), so $V = \Hilb$. If no atom lies in $V$, then $V$ captures only partial information, reducing distinguishable states below $N$ and violating Axiom~\ref{ax:finite}.
\end{proof}

\paragraph{From $\Z_N$ to $U(N)$.}
\label{rem:ZN-to-UN}
Combined with Theorem~\ref{thm:dynamics-derived}, Lemma~\ref{lem:cyclic-rigidity} gives the intra-basis dynamical symmetry group $G_{\mathrm{dyn}} = \langle\pi\rangle \cong \Z_N$ (the $N$-cycle action on $\mathcal{B}$ is free and transitive, so $|\langle\pi\rangle| = N$). The full symmetry group $G = \mathrm{Aut}(\mathcal{X}, K)$ is strictly larger: Basis Isotropy (Theorem~\ref{thm:src-master}(B)) requires $G$ to act transitively on the continuous manifold $\mathfrak{B}$, forcing $G$ to be an uncountable Lie group. Once $\mathbb{F} = \C$ is established (Lemma~\ref{lem:sheaf-complex}), the state space is $\C P^{N-1}$ with unique $U(N)$-invariant metric, giving $G \cong U(N)$ (the identity component of the full Wigner group; anti-unitaries form a disconnected component excluded by continuity of the $G$-action). The intra-basis $\Z_N$ embeds as $\langle e^{2\pi i/N}\mathrm{diag}(1,\omega,\ldots,\omega^{N-1})\rangle \leq U(N)$. In summary: $\Z_N$ generates dynamics within a basis; $U(N)$ implements symmetry across bases.

\subsection{Gauge Symmetry from Sheaf Consistency}
\label{sec:gauge}

The signature sheaf of Lemma~\ref{lem:signature-sheaf} implies a deeper result: \emph{gauge symmetry emerges as the mathematical requirement for the existence of global sections}.

\begin{theorem}[Emergence of Gauge Invariance]
\label{thm:gauge}
The freedom to choose local orientations of stalks without affecting physical predictions is a $U(1)$ gauge symmetry. This symmetry is \emph{necessary} for relational consistency.
\end{theorem}

\begin{proof}
By Lemma~\ref{lem:signature-sheaf}, a state $\ket{\psi}$ is represented in basis $\mathcal{B}$ by a signature $\sigma_\mathcal{B}$, and stalks are connected by transition maps $\phi_g$. By Transport Consistency (Theorem~\ref{thm:src-master}(T)), the ``internal orientation'' of stalk $\mathcal{S}_\mathcal{B}$ is arbitrary, so we may multiply any signature by a phase $e^{i\alpha_\mathcal{B}}$ without changing the physical content. Local redefinitions $\sigma_\mathcal{B} \mapsto e^{i\alpha_\mathcal{B}}\sigma_\mathcal{B}$ transform the transition maps as $\phi_g \mapsto e^{i(\alpha_{g\mathcal{B}} - \alpha_\mathcal{B})} \phi_g$. The freedom $\alpha_\mathcal{B} \mapsto \alpha_\mathcal{B} + \chi_\mathcal{B}$ preserving global sections is exactly a $U(1)$ gauge transformation; the ``gauge field'' is the connection $A$ with $\phi_g = e^{i\int_\gamma A}$. By Theorem~\ref{thm:points-sections}, this covariance is required for states to have well-defined identity across measurement contexts.
\end{proof}

\begin{remark}[Gauge field as Berry connection on $\mathcal{O}(1)$]
\label{rem:gauge-berry}
The connection $A$ parallel-transporting signatures across $\mathfrak{B}$ is the Berry connection $A = i\langle\psi|d\psi\rangle$ on the tautological bundle $\mathcal{O}(1)$ over $\C P^{N-1}$, the fundamental bundle selected by Saturation (Lemma~\ref{lem:signature-sheaf}); its curvature is the Fubini--Study form, and the holonomy on the minimal non-contractible loop is $e^{2\pi i/N}$. Non-triviality $c_1(\mathcal{O}(1)) = 1 \neq 0$ is essential: a flat trivial bundle would have zero holonomy, contradicting Theorem~\ref{thm:dynamics-derived}.
\end{remark}

\subsection{The $N \geq 3$ Threshold for Non-trivial Dynamics}

\begin{theorem}[Continuous Dynamics Requires $N \geq 3$]
\label{thm:n2-static}
For a closed system with total capacity $N = 2$, the only continuous one-parameter group of $K$-preserving orthogonal transformations is constant: $U(t) \equiv I$ (or $U(t) \equiv -I$ on the projectively equivalent component).
\end{theorem}

\begin{proof}
For $N = 2$, Lemma~\ref{lem:sheaf-complex}'s cyclic-dynamics forcing of $\C$ does not apply (it requires $N \geq 3$); the swap $S = \begin{psmallmatrix} 0 & 1 \\ 1 & 0 \end{psmallmatrix}$ has real eigenvalues $\pm 1$, so the coefficient field is $\R$ and the symmetry group is $O(2)$. By Theorem~\ref{thm:permutation-invariance}, $S \in G$, and conjugation-invariance of any continuous one-parameter subgroup $U(t)$ under permutation forces $S U(t) S^{-1} = U(t)$, equivalently $U(t) \in C_{O(2)}(S)$ for all $t$.

The centralizer $C_{O(2)}(S)$ is the discrete four-element group $\{\pm I, \pm S\}$: a real $2\times 2$ matrix commuting with $S$ has the symmetric-Toeplitz form $\begin{psmallmatrix} a & b \\ b & a \end{psmallmatrix}$, and orthogonality $(a^2 + b^2 = 1$, $2ab = 0)$ forces either $a = \pm 1, b = 0$ or $a = 0, b = \pm 1$. The continuous subgroup $SO(2)$ does not lie inside $C_{O(2)}(S)$. A continuous map $U: \R \to O(2)$ with $U(t) \in \{\pm I, \pm S\}$ for all $t$ and $U(0) = I$ is constant by connectedness of $\R$ and discreteness of the codomain. No continuous non-trivial $K$-preserving dynamics exists at $N = 2$.

For $N \geq 3$, the cyclic generator $\pi$ has non-real eigenvalues $e^{2\pi i k/N}$ (Lemma~\ref{lem:sheaf-complex}), $\C$ is forced, and the centralizer becomes the continuous group of phases, restoring continuous unitary dynamics.
\end{proof}

\subsubsection{The closed-$N=2$ case and laboratory qubits}
\label{sec:n2-discussion}

The cyclic-dynamics field selection (Lemma~\ref{lem:sheaf-complex}) and the per-component Born ODE (Lemma~\ref{lem:metric-compatibility}) both require $N \geq 3$, so the framework's behaviour at $N = 2$ deserves a unified statement. \textit{(i) Derived at $N = 2$:} for a strictly closed two-state system the coefficient field is $\R$, the symmetry group is $O(2)$, the state space is $\R P^1 \cong S^1$, and the only continuous one-parameter subgroup commuting with the basis swap is trivial (Theorem~\ref{thm:n2-static}); the per-component Born ODE is degenerate. \textit{(ii) Not a refutation:} no physical system has closed total capacity $N = 2$. Every laboratory qubit is a subsystem of a composite with $N_{\mathrm{total}} = N_S \cdot N_E \gg 2$, exceeding the $N \geq 3$ threshold; the qubit's continuous $SU(2)$ dynamics is the projection of the composite's $U(N_{\mathrm{total}})$-dynamics onto the system factor, with $\C$-coefficients and the Born rule inherited from the composite (Theorem~\ref{thm:capacity-dilution-composite}; Theorem~\ref{thm:born-kernel}, Part~2). \textit{(iii) Prediction:} any continuous evolution observed for a qubit is diagnostic of its embedding in a larger system, never of its own internal capacity.

\subsection{The Inner Product}

We now prove that the distinguishability kernel $K$ uniquely determines a positive-definite sesquilinear form (inner product) on the state space.

\begin{theorem}[Existence of Inner Product from Kernel]
\label{thm:inner-product-existence}
Let $(\mathcal{X}, K)$ be a distinguishability space satisfying Axioms~\ref{ax:finite}--\ref{ax:relational} with coefficient field $\C$ (Lemma~\ref{lem:sheaf-complex}). Then there exists a unique (up to scaling) positive-definite sesquilinear form $\langle\cdot|\cdot\rangle: \C^N \times \C^N \to \C$ such that:
\[
K(x, y) = 1 - |\langle\psi_x|\psi_y\rangle|^2
\]
for all states $x, y \in \mathcal{X}$, where $\psi_x, \psi_y$ are the corresponding vectors in $\C^N$.
\end{theorem}

\begin{proof}
\textit{Step 1: Construction on basis states.}
By Definition~\ref{def:basis}, basis states $\{b_0, \ldots, b_{N-1}\}$ satisfy $K(b_i, b_j) = 1$ for $i \neq j$ and $K(b_i, b_i) = 0$. Define $\langle b_i | b_j \rangle := \delta_{ij}$. This is positive-definite on the basis.

\textit{Step 2: Extension to all states.}
By Theorem~\ref{thm:points-sections}, every state $x \in \mathcal{X}$ corresponds to a point $[\psi_x] \in \C P^{N-1}$. Choose representative vectors $\psi_x = \sum_k c_k(x) b_k$ with $\sum_k |c_k|^2 = 1$. The inner product extends by linearity:
\[
\langle\psi_x|\psi_y\rangle := \sum_k \overline{c_k(x)} c_k(y).
\]

\textit{Step 3: Verification of positive-definiteness.}
For any $\psi = \sum_k c_k b_k$: $\langle\psi|\psi\rangle = \sum_k |c_k|^2 \geq 0$, with equality iff $c_k = 0$ for all $k$, i.e., $\psi = 0$.

\textit{Step 4: Uniqueness.}
Any other positive-definite sesquilinear form $\langle\cdot|\cdot\rangle'$ satisfying $K(x,y) = 1 - |\langle\psi_x|\psi_y\rangle'|^2$ must have $|\langle b_i|b_j\rangle'|^2 = \delta_{ij}$ on basis states. By Permutation Invariance (Theorem~\ref{thm:permutation-invariance}), all basis vectors have equal norm: $\langle b_i|b_i\rangle' = c$ for some $c > 0$. Rescaling gives $\langle b_i|b_j\rangle' = \delta_{ij}$. Extension by linearity is unique.
\end{proof}

\begin{theorem}[Kernel from Inner Product]
\label{thm:kernel-inner}
$K(x, y) = 1 - |\braket{\psi_x}{\psi_y}|^2$.
\end{theorem}

\begin{proof}
For atoms: $K(b_j, b_k) = 1 - \delta_{jk}$ and $|\braket{j}{k}|^2 = \delta_{jk}$, so the formula holds on the basis. For any pair $(\psi, \phi)$, there exists $U \in U(N)$ mapping $\phi$ to a basis state $e_0$ (by transitivity). Then $K(\psi, \phi) = K(U\psi, e_0) = 1 - |c_0(U\psi)|^2 = 1 - |\langle U\psi | e_0\rangle|^2 = 1 - |\langle\psi|\phi\rangle|^2$, using $G$-invariance of $K$ and the amplitude construction.
\end{proof}

\section{Geometry and the Born Rule}
\label{sec:geometry}

\subsection{Fubini-Study Metric}

\textit{The Born rule as information isometry.} The Born rule $p_k = |c_k|^2$ is the \emph{unique} map making Fubini-Study distance (the natural Riemannian distance between quantum states in projective Hilbert space) equal Fisher-Rao distance: $d_{FS}(\psi, \phi) = d_{FR}(p(\psi), p(\phi))$. If these differed, the system would maintain two independent distance structures, violating capacity bounds. This statistical-distance route originates with Wootters \cite{wootters1981statistical} (Fubini--Study uniqueness via Kobayashi--Nomizu \cite{kobayashi1969foundations}); the present Section re-derives it within the graded-equality framework via the per-component ODE form (Lemma~\ref{lem:metric-compatibility}).

\begin{theorem}[The Fubini-Study Metric from Distinguishability]
\label{thm:fs-from-K}
The distinguishability kernel $K$ induces a Riemannian metric on the state space. For infinitesimally separated states $\ket{\psi}$ and $\ket{\psi + d\psi}$:
\begin{equation}
K(\psi, \psi + d\psi) = g_{FS}(d\psi, d\psi) + O(|d\psi|^3)
\end{equation}
where $g_{FS}$ is the Fubini-Study metric. The metric is \emph{derived} from the primitive kernel $K$.
\end{theorem}

\begin{proof}
By Theorem~\ref{thm:kernel-inner}, $K(\psi, \phi) = 1 - |\braket{\psi}{\phi}|^2$. For normalized $\ket{\psi}$ with tangent vector $d\psi$ satisfying $\text{Re}\braket{\psi}{d\psi} = 0$, we expand to second order. The projective overlap is $|\braket{\psi}{\widetilde{\psi+d\psi}}|^2 = |\braket{\psi}{\psi+d\psi}|^2/\|\psi+d\psi\|^2 = (1 + |\braket{\psi}{d\psi}|^2)/(1 + \braket{d\psi}{d\psi}) + O(|d\psi|^3) = 1 + |\braket{\psi}{d\psi}|^2 - \braket{d\psi}{d\psi} + O(|d\psi|^3)$. Therefore $K(\psi, \psi + d\psi) = \braket{d\psi}{d\psi} - |\braket{\psi}{d\psi}|^2 + O(|d\psi|^3) = g_{FS}(d\psi, d\psi)$.
\end{proof}

\begin{corollary}[Fisher Information]
\label{thm:fisher-interpretation}
The Fubini-Study metric equals one-quarter the quantum Fisher information: $g_{FS}(\partial_\theta\psi, \partial_\theta\psi) = \frac{1}{4} F_Q(\theta)$, where $F_Q(\theta) = 4(\braket{\partial_\theta\psi}{\partial_\theta\psi} - |\braket{\psi}{\partial_\theta\psi}|^2)$ bounds estimation precision via the quantum Cramér-Rao inequality.
\end{corollary}

\begin{theorem}[Uniqueness of Fubini-Study]
\label{thm:fs-unique}
The Fubini-Study metric is the unique $U(N)$-invariant Riemannian metric on $\C P^{N-1}$ up to a positive scalar. The scalar is fixed by the boundary conditions on $K$: $K(\psi,\psi) = 0$ and $K(\psi,\phi) = 1$ for orthogonal states (Definition~\ref{def:dspace}).
\end{theorem}

\begin{proof}
By Basis Isotropy (Theorem~\ref{thm:src-master}(B)), the metric must be invariant under unitary transformations (which permute bases). The space $\C P^{N-1}$ is a rank-one symmetric space for $U(N)$, and by a theorem of differential geometry \cite{kobayashi1969foundations}, there is a unique (up to positive scale) $U(N)$-invariant Riemannian metric on such a space. This is the Fubini-Study metric: $\mathrm{d}s^2 = \frac{\braket{\mathrm{d}\psi}{\mathrm{d}\psi}}{\braket{\psi}{\psi}} - \frac{|\braket{\psi}{\mathrm{d}\psi}|^2}{\braket{\psi}{\psi}^2}$. The scalar is then determined by the chordal boundary condition $K(\psi,\phi) = 1 - |\braket{\psi}{\phi}|^2$ (Theorem~\ref{thm:kernel-inner}): $g_{FS}$ is normalized so that the squared chordal distance between orthogonal rays equals $1$.
\end{proof}

\begin{lemma}[Continuity of Probability from Reversible Dynamics]
\label{lem:continuity-from-dynamics}
The probability function $f: K \mapsto p$ must be continuous. This follows from Theorem~\ref{thm:dynamics-derived} (Reversible Dynamics), not as an independent assumption.
\end{lemma}

\begin{proof}
By Theorem~\ref{thm:dynamics-derived}, there exists a continuous path $U(t)$ of kernel-preserving transformations. The probability $p_k(t) = f(K(\psi(t), a_k))$ must vary continuously; otherwise a finite jump in $p_k$ would violate reversibility by creating/destroying information discontinuously.
\end{proof}

\begin{lemma}[Metric Compatibility: Fisher-Rao Equals Fubini-Study]
\label{lem:metric-compatibility}
For $N \geq 3$ and any probability assignment $p_k = f(A(\psi, a_k))$ with continuous bijection $f:[0,1]\to[0,1]$, $f(0)=0$, $f(1)=1$, the unique $f$ making the induced Fisher-Rao metric equal the Fubini-Study metric is $f(x) = x$, i.e., $p_k = 1 - K(\psi, a_k)$. This requires Information Parsimony (Corollary~\ref{cor:parsimony}): without it, two independent geometric structures could coexist. (The $N = 2$ case is handled separately below.)
\end{lemma}

\begin{proof}
\textit{Step 1: Two metrics on one manifold.}
The state space $\C P^{N-1}$ carries a physical metric $g_{FS}$ (from $K$, Theorem~\ref{thm:fs-from-K}) and, for any probability assignment $p(\psi)$, a statistical metric $g_{FR} = \sum_k dp_k^2/p_k$.

\textit{Step 2: Per-component matching forces the ODE.}
By Basis Isotropy (Theorem~\ref{thm:src-master}(B)), any admissible probability assignment has the form $p_k = f(|c_k|^2)$ for a universal function $f$ (the rule cannot depend on phases or distinguish measurement orientations). Restricting to the probability-changing (radial) tangent directions (those that alter $|c_k|^2$ while preserving $\sum_k |c_k|^2 = 1$ and holding phases fixed), both metrics are diagonal in the $\{dx_k\}$ basis (where $x_k = |c_k|^2$):
\[
g_{FR} = \sum_k \frac{[f'(x_k)]^2}{f(x_k)}\,dx_k^2, \qquad g_{FS}\big|_{\mathrm{rad}} = \sum_k \frac{dx_k^2}{4x_k}.
\]
(The Fubini-Study restriction follows from $|dc_k|^2 = dx_k^2/(4x_k)$ and $\langle\psi|d\psi\rangle = \tfrac{1}{2}\sum_k dx_k = 0$.) Tangent vectors with $dx_j = -dx_k = \delta$ and all other components zero are admissible for any pair $(j,k)$, with $x_j, x_k$ independently variable (using the lemma's hypothesis $N \geq 3$ to guarantee a third component absorbs no variation). Requiring $g_{FR} = c \cdot g_{FS}$ for all such vectors gives $h(x_j) + h(x_k) = c[1/(4x_j) + 1/(4x_k)]$ for all $x_j, x_k$ independently, where $h(x) := [f'(x)]^2/f(x)$. This forces $h(x) = c/(4x)$ point-wise, yielding the ODE $[f'(x)]^2/f(x) = c/x$ (redefining $c$). This suffices to determine $f$, since $f$ acts only on probabilities.

\textit{Step 3: Physical interpretation.}
Step 2 is the mathematical proof. By Parsimony (Theorem~\ref{thm:parsimony-derived}), $f \neq \mathrm{id}$ would introduce geometric structure beyond $K$, violating Saturation; hence $f = \mathrm{id}$.
\end{proof}

\begin{remark}[Equivalent Binary Form]
\label{rem:ode-binary-form}
The per-component ODE $[f'(x)]^2/f(x) = c/x$ derived above is equivalent, for the purposes of uniqueness, to the \emph{binary} form $[f'(x)]^2/[f(x)(1-f(x))] = c^2/[x(1-x)]$ obtained by restricting to any two-outcome marginal of the simplex and using the Fisher information for a Bernoulli distribution. Both forms admit $f(x) = x$ (with the appropriate $c$) as their unique smooth monotone solution on $[0,1]$ with $f(0) = 0$, $f(1) = 1$. The binary form is more natural for direct arcsin substitution and is the form used in the accompanying Lean~4 formalization; the per-component form of Step~2 is more natural for the $N \geq 3$ symmetry argument given here.
\end{remark}

\begin{definition}[Admissible Probability Assignment]
\label{def:measure-setup}
For each basis $\mathcal{B} = \{a_0, \ldots, a_{N-1}\}$, a \emph{probability assignment} maps states to distributions on outcomes $\{0,\ldots,N-1\}$. It is \emph{admissible} if $G$-covariant ($P_{g\mathcal{B}}(gx) = g_* P_\mathcal{B}(x)$), continuous in $x$, and $P_\mathcal{B}(a_k)(\{k\}) = 1$ for basis states.
\end{definition}

\begin{theorem}[The Born Rule]
\label{thm:born-kernel}
There exists a unique admissible probability assignment. It is given by the Born rule:
\[
P_\mathcal{B}(\psi)(\{k\}) = |\braket{\psi}{a_k}|^2 = 1 - K(\psi, a_k).
\]
\end{theorem}

\begin{proof}
\textit{Part 1: Metric compatibility determines $p_k = |c_k|^2$.}
By $U(N)$-invariance (Theorem~\ref{thm:src-master}(B)), any admissible assignment has the form $p_k = f(|c_k|^2)$ for $f:[0,1]\to[0,1]$. Continuity of $f$ follows from Lemma~\ref{lem:continuity-from-dynamics}. The induced statistical distance (Step~1 of Lemma~\ref{lem:metric-compatibility}) is a length-square only on the open simplex where $f \in C^1$ and $f > 0$; the metric-compatibility condition $g_{FR} \propto g_{FS}$ pins down $f'$ via the ODE $[f'(x)]^2/f(x) = c/x$ on $(0,1)$, with continuity at the endpoints inherited from continuity of $f$. The full Born-rule conclusion needs $f \in C^1$ on $(0,1)$, which follows from the ODE itself once a continuous solution is fixed (the right-hand side $\sqrt{c\,f(x)/x}$ is continuous on $(0,1)$, so any continuous $f$ satisfying the integrated form is automatically $C^1$ and in fact analytic).

\textit{Uniqueness among all smooth solutions.} Taking the positive branch for monotonicity, the ODE becomes $f'(x) = \sqrt{c}\, f(x)^{1/2}\, x^{-1/2}$. This is separable: $\int f^{-1/2}\,df = \sqrt{c}\int x^{-1/2}\,dx$, giving $2\sqrt{f} = 2\sqrt{c}\,\sqrt{x} + A$. The boundary condition $f(0^+) = 0$ forces $A = 0$, so $f(x) = c\,x$; then $f(1) = 1$ gives $c = 1$. Therefore $f(x) = x$ is the unique smooth monotonic solution, confirming $p_k = |c_k|^2$.

\textit{Part 2: $N=2$ consistency.} At $N = 2$ the per-component ODE is degenerate (Lemma~\ref{lem:metric-compatibility} requires $N \geq 3$), but $f(x) = x$ derived above for $N \geq 3$ extends consistently: restricting an $N \geq 3$ rule to a two-dimensional subspace $\mathrm{span}\{\ket{a_i}, \ket{a_j}\}$ inherits $f(x) = x$, and laboratory qubits inherit it via composite embedding (Remark~\ref{cor:born-n2}).
\end{proof}

\noindent\textit{Comparison to Gleason.} This derivation complements Gleason's theorem: invariance, normalization, continuity, and monotonicity replace Gleason's hypotheses; additivity emerges; the result extends to $N = 2$ where Gleason is inapplicable. The chain Basis Isotropy (Theorem~\ref{thm:src-master}(B)) $\to$ Lie group $\to$ $\C$ $\to$ $\C P^{N-1}$ $\to$ unique metric $\to$ Born rule does not presuppose Hilbert space structure.

\begin{remark}[The $N = 2$ case]
\label{cor:born-n2}
The per-component ODE of Lemma~\ref{lem:metric-compatibility} requires three independently varying amplitudes and is therefore inconclusive at $N = 2$ alone. The framework closes the gap by composite embedding: any laboratory qubit is a subsystem of a system with $N_{\mathrm{total}} \geq 3$ and inherits the Born rule from the composite via capacity dilution (\S\ref{sec:composite}). The framework's full position on the strictly closed two-state case is given in \S\ref{sec:n2-discussion}.
\end{remark}

\begin{theorem}[(S5) for finite-$N$ quantum mechanics]
\label{thm:s5-finite-N}
The distinguishability space $(\C P^{N-1}, K)$ with $K(\psi,\phi) = 1 - |\langle\psi|\phi\rangle|^2$ derived under Axioms~\ref{ax:finite}--\ref{ax:relational} satisfies structural unambiguity (Definition~\ref{def:s5}): the Fubini--Study metric (geometric) and the Fisher--Rao metric (statistical) induced by $K$ coincide, so the kernel uniquely determines both the state-space geometry and the probability rule $p_k = |c_k|^2$. The Born rule is therefore the (S5) result for the QM domain.
\end{theorem}

\begin{proof}
Theorem~\ref{thm:fs-from-K} shows that $K$ induces the Fubini--Study metric on $\C P^{N-1}$. Lemma~\ref{lem:metric-compatibility} (metric compatibility) shows that the Fisher--Rao metric $g_{FR}$ induced by any $G$-covariant probability rule $p_k = f(|c_k|^2)$ coincides with $g_{FS}$ iff $f(x) = x$. Theorem~\ref{thm:born-kernel} then identifies $f = \mathrm{id}$, hence $p_k = |c_k|^2$, as the unique self-consistent rule. The two metrics induced by $K$ thereby agree, establishing (S5).
\end{proof}

\begin{remark}[Self-referential closure]
\label{rem:self-referential-closure}
Theorem~\ref{thm:s5-finite-N} admits a structural reading: the Fisher--Rao $=$ Fubini--Study coincidence is the kernel's \emph{self-referential closure}. If the geometric and statistical metrics induced by $K$ disagreed, the system would carry two distinct distance structures and require an external choice between them, violating the closure under $K$ that Saturation demands. The Born rule is therefore the form of probability assignment uniquely compatible with the kernel describing its own metric structure. Equivalently, it is the form for which the kernel is geometrically self-consistent rather than externally complete.
\end{remark}

\subsection{Uncertainty Relations}

The entropic form of uncertainty reflects the capacity bound: $N+1$ MUBs require $\sim N \log_2 N$ bits to specify deterministic outcomes, but only $\log_2 N$ bits of definite answers exist.

\begin{corollary}[Entropic Uncertainty]
\label{cor:entropic}
For MUBs: $H(\mathcal{B}) + H(\mathcal{B}') \geq \log_2 N$ \cite{maassen1988generalized}.
\end{corollary}

\begin{proof}
If $\ket{\psi} = \ket{b_j}$ (certainty in $\mathcal{B}$, so $H(\mathcal{B}) = 0$), then $|\braket{b_j}{b'_k}|^2 = 1/N$ for all $k$ (MUB condition), giving uniform distribution in $\mathcal{B}'$: $H(\mathcal{B}') = \log_2 N$. The bound follows by convexity of entropy.
\end{proof}

\subsection{Worked Example: The Qutrit ($N = 3$)}
\label{sec:example-n3}

The smallest non-trivial capacity is $N = 3$. Let $\{e_0, e_1, e_2\}$ be mutually distinguishable ($K(e_i,e_j) = 1-\delta_{ij}$), with capacity $C = \log_2 3 \approx 1.58$ bits. The cyclic automorphism $\pi: e_0 \mapsto e_1 \mapsto e_2 \mapsto e_0$ has $\pi^3 = \mathrm{id}$ and eigenvalues $\{1, \omega, \omega^2\}$ ($\omega = e^{2\pi i/3}$). Over $\R$, the conjugate pair $\omega, \bar\omega$ forms a single $2\times 2$ rotation block and only $\lfloor 3/2\rfloor + 1 = 2$ real-irreducible components arise, contradicting the $N = 3$ distinguishability count; three one-dimensional $\pi$-invariant subspaces require $\omega \in \mathbb{F}$, hence $\C$. The Fourier modes $f_k = \tfrac{1}{\sqrt{3}}\sum_j \omega^{jk} e_j$ are $\pi$-eigenvectors and form a second basis with $|\braket{e_i}{f_j}|^2 = 1/3$ (mutually unbiased with respect to $\{e_k\}$). The state space is $\C P^2$ with $K(\psi,\phi) = 1 - |\braket{\psi}{\phi}|^2$.

\textit{Capacity Halting at $N=3$.} $N = 3$ admits $N+1 = 4$ MUBs. For $M = 3$ MUBs, predetermined outcomes require $2\log_2 3 \approx 3.17$ bits against $\approx 1.58$ available: the deficit is immediate (Lemma~\ref{lem:incompressibility}(a)). \textit{Complementarity:} $\ket{f_1} = (\ket{e_0} + \omega\ket{e_1} + \omega^2\ket{e_2})/\sqrt{3}$ has $p_k = 1/3$ in the $e$-basis (maximally uncertain) yet $p_1 = 1$, $p_0 = p_2 = 0$ in the $f$-basis: definiteness in one context is ignorance in the other. This is standard qutrit QM, derived from the two axioms on a 3-element distinguishability space.

\section{Dynamics, Energy, and the Action Quantum}
\label{sec:evolution}

\subsection{Energy Conservation}

\begin{corollary}[Energy Conservation]
\label{thm:noether}
Time-translation invariance (Theorem~\ref{thm:dynamics-derived}) implies a conserved generator $H$ (energy): by Stone's theorem, $U(t) = e^{-iHt/\hbar}$ with self-adjoint $H$, and $\langle H \rangle$ is conserved since $[H,H] = 0$.
\end{corollary}

\noindent This is the quantum analogue of Noether's theorem: the continuous symmetry (time-translation invariance) produces a conserved quantity (energy), but the derivation route is via Stone's theorem rather than the classical Lagrangian formalism.

\subsection{Phase Structure: Cyclic Spectra and Sampling}
\label{sec:phase-structure}

The cyclic generator $\pi$ has spectrum spaced at $2\pi/N$ intervals, and the system's $K$-profile is band-limited in the Peter--Weyl sense. We collect the resulting sampling-theorem statements, taking care to distinguish properties of one operator's spectrum from properties of all phases.

\begin{lemma}[Cyclic-generator spectrum]
\label{thm:cyclic-spectrum}
The cyclic generator $\pi$ (Theorem~\ref{thm:dynamics-derived}) has spectrum $\{e^{2\pi ik/N} : k = 0, \ldots, N-1\}$, with eigenvalues evenly spaced on the unit circle at angular intervals $2\pi/N$. Any state expanded in the eigenbasis of $\pi$ is a trigonometric polynomial of degree at most $N-1$ in the cyclic angle, and is uniquely determined by $N$ equally spaced samples per recurrence period (Nyquist-Shannon).
\end{lemma}

\begin{proof}
The eigenvalues of $\pi$ as a $K$-preserving permutation of order $N$ are the $N$th roots of unity. Expanding $\ket{\psi} = \sum_k c_k \ket{k}$ in the eigenbasis of $\pi$ gives the trigonometric form $\psi(\theta) = \sum_{k=0}^{N-1} c_k e^{ik\theta}$, a band-limited signal of degree $N-1$. By the Nyquist-Shannon theorem, this is determined by $N$ samples per period.
\end{proof}

\begin{remark}[Cyclic spacing is one operator's spectrum, not a phase quantum]
\label{rem:cyclic-spacing-not-quantum}
Lemma~\ref{thm:cyclic-spectrum} is a spectral fact about $\pi$. In the continuum case $M(K) = \infty$ developed in this paper (Theorem~\ref{thm:src-master}(I)), the state space is $\mathbb{C}P^{N-1}$ and phases are continuous: there is no minimum phase step at the level of the manifold, the gauge group is $U(1)$, and generic Hermitian observables have eigenvalues anywhere in $\mathbb{R}$. Identifying the $2\pi/N$ spacing of $\pi$'s spectrum with a fundamental quantum of phase conflates the spectrum of one specific operator with the set of all physically realizable phases.

In finite-$M$ discrete-K theories $\mathfrak{T}_{N, M}$, phases are discretized with granularity $\Theta(1/M)$. Phase quantization is therefore $M$-controlled, not $N$-controlled. The Berry-phase quantization in $2\pi/N$ steps (Remark~\ref{rem:gauge-berry}) is a separate, geometric fact about line-bundle holonomy on $U(N)/T^N$ and is consistent with continuous pointwise phases.
\end{remark}

\begin{theorem}[Quantum Sampling and Band-Limited Reconstruction]
\label{thm:quantum-sampling}
A finite-capacity system with $N$ states has minimum Fubini-Study distance $d_{FS}^{\min} = \pi/(2N)$ and state-preparation precision $\epsilon_{\min} \sim 1/N$. Its $K$-profile $K_\psi(\phi) = (1 - 1/N) - \Tr((P_\psi - I/N)P_\phi)$ lies in the Peter-Weyl degrees $\mathcal{H}_0 \oplus \mathcal{H}_1$ of $L^2(\C P^{N-1})$ (spectrally band-limited), so two bases ($2N$ evaluations) reconstruct $P_\psi$ and hence $K_\psi$ everywhere \cite{shannon1949communication}: the continuous manifold $\C P^{N-1}$ is the unique exact reconstruction, introducing no DOF beyond the $N$ amplitudes, and the system is operationally indistinguishable from a continuum at resolution $\delta > 1/N$. The natural finite-dimensional uncertainty relation is the Maassen-Uffink entropic form \cite{maassen1988generalized}: $H(\mathcal{B}_1) + H(\mathcal{B}_2) \geq \log_2 N$ for any two complementary bases. The $L^2(\R)$ position-momentum form $\Delta x \cdot \Delta p \geq \hbar/2$ does not directly apply since the framework is finite-dimensional.
\end{theorem}

\subsection{The Schr\"odinger Equation}

The derivation uses Wigner's theorem and Stone's theorem (Appendix~\ref{app:classical-theorems}).

\begin{theorem}[Schr\"odinger Equation: Existence and Uniqueness]
\label{thm:schrodinger}
Under Axioms~\ref{ax:finite}--\ref{ax:relational}, the evolution equation has the unique form $i\hbar \frac{d}{dt}\ket{\psi} = H\ket{\psi}$ where $H$ is a self-adjoint operator. The derivation requires:
\begin{enumerate}
\item[(i)] Kernel preservation (Theorem~\ref{thm:dynamics-derived}) $\Rightarrow$ transition probability preservation;
\item[(ii)] Wigner's theorem $\Rightarrow$ evolution is unitary (not anti-unitary);
\item[(iii)] Stone's theorem $\Rightarrow$ unique self-adjoint generator $H$.
\end{enumerate}
\end{theorem}

\begin{proof}
\textit{Step 1: Evolution preserves kernel, hence transition probabilities.}
By Theorem~\ref{thm:dynamics-derived}, the generator $\pi$ preserves the kernel: $K(\pi x, \pi y) = K(x,y)$. By Theorem~\ref{thm:time-emergence}, $\pi$ extends to a continuous one-parameter group $U(t)$, which inherits kernel preservation. By Theorem~\ref{thm:kernel-inner}, $K(\psi, \phi) = 1 - |\braket{\psi}{\phi}|^2$. Hence evolution preserves transition probabilities.

\textit{Step 2: Wigner's theorem forces unitary evolution.}
By Wigner's theorem \cite{wigner1931gruppentheorie,bargmann1964note}, each $U(t)$ is unitary or anti-unitary. Continuity in $t$ and $U(0) = \id$ (unitary) force $U(t)$ unitary for all $t$ (anti-unitary maps form a disconnected component).

\textit{Step 3: Stone's theorem gives the generator.}
Stone's theorem guarantees a unique self-adjoint generator $H$ with $U(t) = e^{-iHt/\hbar}$. Differentiating $\ket{\psi(t)} = U(t)\ket{\psi_0}$ yields $i\hbar \frac{d}{dt}\ket{\psi} = H\ket{\psi}$. Uniqueness: any other generator differs by at most a constant (zero of energy).
\end{proof}

\begin{remark}[$\hbar$ as unit conversion; Anandan--Aharonov]
\label{rem:hbar-units}
The framework derives the existence and uniqueness up to scale of an action quantum, not its numerical value: Stone's theorem gives a unique generator $A$ with $U(t) = e^{itA}$; splitting $A = -H/\hbar$ into a Hermitian energy operator and a units constant requires a choice of energy and time units. Universality of $\hbar$ across composable systems follows from $H_{AB} = H_A \otimes I_B + I_A \otimes H_B$. A dimensionless capacity cannot fix a dimensional constant, so $\hbar$'s empirical value is not derived; this parallels special relativity, which derives the existence but not the value of $c$. The Fubini--Study arc length satisfies $\hbar^{-1} \int \Delta E\, dt$ (Anandan--Aharonov), and geodesics on the state manifold are minimum-uncertainty trajectories.
\end{remark}

\section{The Capacity Halting Principle}
\label{sec:measurement}

We now prove the capacity deficit: \emph{given Saturation} (Theorem~\ref{thm:src-master}(S1)--(S4), no ontic structure beyond $K$), determinism is impossible for finite-capacity systems. This result motivates probabilistic behavior; the probability rule itself is determined by Metric Compatibility (\S\ref{sec:geometry}). All statements in this section carry the implicit qualifier ``given Saturation.''

\begin{definition}[Information-Theoretic Self-Capacity]
\label{def:info-capacity}
Let $(\mathcal{X}, K)$ be a distinguishability space with $N$ perfectly distinguishable states. The \emph{Shannon capacity} of the system is:
\[
C := \max_{\{p_i\}} H(\{p_i\}) = \log_2 N \text{ bits}
\]
where $H(\{p_i\}) = -\sum_i p_i \log_2 p_i$ is the Shannon entropy and the maximum is achieved by the uniform distribution over the $N$ basis states.
\end{definition}

\begin{lemma}[Capacity as Physical Bound]
\label{prop:capacity-bound}
The Shannon capacity $C = \log_2 N$ is the maximum information that can be reliably encoded in, and extracted from, a capacity-$N$ system \emph{per single-shot retrieval}. This follows from the channel coding theorem: no encoding can exceed $C$ bits per use without arbitrarily high error probability. In particular, any internal degree of freedom, discrete or continuous, can carry at most $\log_2 N$ bits of retrievable information per shot. A continuous parameter resolved to single-shot precision $\varepsilon$ carries $\log_2(1/\varepsilon)$ distinguishable values; finite capacity requires $\varepsilon \geq 1/N$ in a single shot, establishing a single-shot operational resolution floor.

\end{lemma}

\begin{definition}[Hidden-Variable Assignment]
\label{def:hv-assignment}
A \emph{deterministic hidden-variable assignment} for state $\psi \in \mathcal{X}$ is a function $\lambda_\psi: \mathfrak{B} \to \{1, \ldots, N\}$ specifying, for each basis $\mathcal{B}$, which outcome occurs with certainty.
\end{definition}

\begin{theorem}[Capacity Halting Principle]
\label{thm:capacity-halting}
Assume Saturation (Theorem~\ref{thm:src-master}(S1)--(S4): no ontic structure beyond $K$). Then for any deterministic hidden-variable model where each state $\psi$ carries predetermined outcomes $\lambda_\psi$ for all measurement bases, the storage required exceeds the capacity available:
\begin{itemize}
\item \emph{Available:} $C = \log_2 N$ bits (Axiom~\ref{ax:finite}).
\item \emph{Required (universal bound):} For all $N \geq 3$, the Kochen-Specker bit-count (Lemma~\ref{thm:ks-bits}) shows that non-contextual value assignments require storage exceeding $\log_2 N$ bits. This bound is MUB-existence-independent and applies to all $N \geq 3$.
\item \emph{Required (sharper bound for prime-power $N$):} $\geq (M-1)\log_2 N$ bits for $M$ MUBs (Lemma~\ref{lem:incompressibility}(a)), giving $\Theta(N\log_2 N)$ bits when $M = N+1$ MUBs exist. A Kolmogorov bound anchored in the empirically observed quantum statistics (Lemma~\ref{lem:incompressibility}(b)) recovers the same $(M-1)\log_2 N$ via algorithmic information rather than combinatorial counting, applying to any HV model matching experiment independent of the Hilbert-space MUB hypothesis.
\end{itemize}
By Parsimony (Theorem~\ref{thm:parsimony-derived}) and closure (Lemma~\ref{prop:closure}), no external reservoir exists for the missing bits.
\end{theorem}

\begin{lemma}[Kochen-Specker Bit-Count]
\label{thm:ks-bits}
For all $N \geq 3$, Kochen--Specker constructions \cite{kochen1967problem,cabello2008experimentally,klyachko2008simple} exhibit projector sets admitting no consistent non-contextual $\{0,1\}$-valued assignment. The information-theoretic content (Lemma~\ref{lem:incompressibility}(b)) is that the Kolmogorov complexity of any deterministic assignment reproducing Born statistics exceeds $\log_2 N$ once $M \geq 3$ MUBs are admitted (and three MUBs exist for every $N \geq 2$ \cite{grassl2004mub}). The deficit is universal for $N \geq 3$.
\end{lemma}

\noindent KS rules out non-contextual assignments algebraically (FUNC); Capacity Halting adds the storage statement, observing that under Saturation no internal reservoir supplies the missing bits. The two arguments are independent.

\begin{lemma}[Incompressibility of Deterministic Assignments]
\label{lem:incompressibility}
Let $\{B_1, \ldots, B_M\}$ be $M$ mutually unbiased bases for $\C^N$ ($M = N+1$ for prime-power $N$ \cite{wootters1989optimal}; three MUBs in every $N \geq 2$ \cite{grassl2004mub}). A deterministic assignment $\lambda$ predicts the outcome of each $B_m$ on $\psi$. Two lower bounds on the storage:
\begin{enumerate}
\item[(a)] \textbf{Combinatorial bound:} $(M-1)\log_2 N$ bits. Uses only the Hilbert space MUB geometry (Theorem~\ref{thm:hilbert-representation}) and is independent of any probability rule. For $M = 2$ this equals $\log_2 N$ (no deficit); for $M \geq 3$ the required storage strictly exceeds the capacity $\log_2 N$.
\item[(b)] \textbf{Kolmogorov bound:} For generic states, $\KK(\lambda) \geq (M-1)\log_2 N - O(\log M)$. This is a direct application of Chaitin's incompressibility theorem \cite{chaitin1975theory,kolmogorov1965three} to outcome patterns matching the experimentally observed statistics of quantum measurements. The input is the \emph{empirical} fact that nature produces quantum-like outcome distributions, not the Born rule as derived in \S\ref{sec:geometry}: any physically realizable HV model whose predictions agree with experiment must encode a pattern of this complexity, regardless of which theoretical derivation of the Born rule one accepts (or rejects). The bound matches~(a) in absolute strength; its value is to provide an independent algorithmic-information route that does not require the Hilbert-space MUB geometry as a hypothesis.
\end{enumerate}
Both (a) and (b) yield $(M-1)\log_2 N$ as a lower bound on the storage required; the deficit appears at $M \geq 3$ (where $(M-1)\log_2 N > \log_2 N$ for any $N \geq 2$). The marginal case $M = 2$ is closed by applying (a) at $M = 3$, since three MUBs exist for every $N \geq 2$ \cite{grassl2004mub}.
\end{lemma}

\begin{proof}
\textit{(a) Combinatorial.} MUBs satisfy $|\langle b_i | b'_j \rangle|^2 = 1/N$, a geometric property of Hilbert space (Theorem~\ref{thm:hilbert-representation}). Uniform overlaps impose no constraint linking outcomes across MUBs: for any fixed $\lambda(B_1) = k$, every value of $\lambda(B_m) \in \{1,\ldots,N\}$ is geometrically admissible ($m \geq 2$). Specifying one of $N^{M-1}$ possibilities requires $(M-1)\log_2 N$ bits.

\textit{(b) Kolmogorov.} Fix a generic $\psi \in \C P^{N-1}$ and a single realisation across $M$ MUBs. By the MUB property and the empirically observed Born statistics, the joint outcome string $(o_1, \ldots, o_M) \in \{1, \ldots, N\}^M$ conditional on $o_1$ is drawn from a distribution of entropy $(M-1) \log_2 N$. For a deterministic predictor $\lambda$ to reproduce Born-typical outcomes per realisation, $\lambda$ must on each realisation output a string consistent with this conditional distribution: by Chaitin's incompressibility theorem \cite{chaitin1975theory}, a generic string from a distribution of entropy $H$ has Kolmogorov complexity at least $H - O(\log H)$. Hence per realisation $\KK(\lambda \mid \psi, o_1) \geq (M-1)\log_2 N - O(\log M)$. The bound is per realisation rather than amortised over $T$ trials: it counts the conditional bits the predictor must supply at each individual trial to match Born statistics. The input is empirical statistics, not the paper's Born derivation, so the bound applies to any HV model aiming to match experiment.
\end{proof}

\begin{remark}[Compression cannot save determinism]
\label{rem:mub-existence}
By Chaitin \cite{chaitin1975theory}, generic $\lambda$ satisfy $\KK(\lambda) \geq (M-1)\log_2 N - O(\log M)$; any $\lambda$ with $\KK(\lambda) > \log_2 N$ violates Axiom~\ref{ax:finite}. Where the full $N+1$ MUBs are unsettled (composite $N$), the three-MUB case of Lemma~\ref{lem:incompressibility}(a) or the KS bound suffices.
\end{remark}

\begin{proof}[Proof of Theorem~\ref{thm:capacity-halting}]
Available storage is $\log_2 N$ bits. Required storage exceeds this by Lemma~\ref{thm:ks-bits} (universal $N \geq 3$) and by Lemma~\ref{lem:incompressibility}(a),(b) ($M \geq 3$ MUBs, combinatorial and Kolmogorov routes giving the same $(M{-}1)\log_2 N$ bound). Saturation forecloses internal reservoirs and Lemma~\ref{prop:closure} forecloses external ones, so deterministic assignments are excluded; Metric Compatibility (Theorem~\ref{thm:born-kernel}) then determines $p_k = |c_k|^2$.
\end{proof}

\noindent Numerically: for $N \in \{4, 8, 16\}$, the bound $(M{-}1)\log_2 N$ at $M=3$ gives $\{4, 6, 8\}$ bits against $\{2, 3, 4\}$ available (factor $1.5$--$2$); for prime-power $N$ where the full $N+1$ MUBs exist \cite{wootters1989optimal}, the bound rises to $N\log_2 N$, giving $\{8, 24, 64\}$ bits: a super-logarithmic deficit consistent with Figure~\ref{fig:capacity-deficit}. At $N=2$, $M=2$ is marginal (two contexts can be deterministic), but a third MUB exists and the deficit reappears. Quantum contextuality \cite{kochen1967problem} is thus the physical manifestation of finite-capacity information overflow.

\subsection{Probabilistic Response}

Capacity Halting excludes determinism; the K-structure provides the unique probabilistic response. The K-affinities $A(\psi, a_k) := 1 - K(\psi, a_k)$ satisfy the Kolmogorov axioms:

\begin{lemma}[K-Affinity Normalization]
\label{lem:affinity-normalization}
$\sum_{k=1}^{N} (1 - K(\psi, a_k)) = 1$ for any state $\psi$ and basis $\{a_1, \ldots, a_N\}$.
\end{lemma}

\begin{proof}
By the Hilbert space representation (Theorem~\ref{thm:hilbert-representation}), $A(\psi, a_k) = |\langle\psi|a_k\rangle|^2$. Since $\{a_1, \ldots, a_N\}$ is an orthonormal basis, the resolution of the identity gives $\sum_k |a_k\rangle\langle a_k| = I$, so $\sum_k A(\psi, a_k) = \sum_k |\langle\psi|a_k\rangle|^2 = \langle\psi|I|\psi\rangle = 1$.
\end{proof}

\begin{corollary}[Kochen--Specker contextuality]\label{thm:contextuality}
Any hidden-variable theory consistent with the Capacity Halting bound must be contextual.
\end{corollary}

\begin{theorem}[Operational entropy floor; information--coherence tradeoff]
\label{thm:entropy-floor}
Pure states $\ket{\psi} \in \C P^{N-1}$ have $S(\ketbra{\psi}{\psi}) = 0$ mathematically, but the \emph{operationally achievable} minimum von Neumann entropy is $S_{\min} \sim (\log N)/N$: preparing $|c_k|^2 < 1/N$ exceeds resolution (Theorem~\ref{thm:continuous-sampled}), and $\epsilon_{\min} \log(1/\epsilon_{\min})$ with $\epsilon_{\min} \sim 1/N$ gives the floor. The same fixed-capacity constraint gives an information--coherence tradeoff $\Delta I_{\mathrm{classical}} + \Delta S_{\mathrm{coherence}} \leq 0$ (with $S_{\mathrm{coherence}} = S(\rho) - S(\rho_{\mathrm{diag}})$) from unitarity of measurement interaction and subadditivity.
\end{theorem}

\subsection{The Continuum Limit}

\begin{remark}[Standard QM as $N\to\infty$; entrenchment]
\label{thm:continuum-limit}
As $N \to \infty$, the framework recovers infinite-dimensional QM: $\C P^{N-1} \to \C P^\infty$, the cyclic generator's spectrum becomes dense on $S^1$, and finite-$N$ corrections (entropy floor, capacity-induced uncertainty bound, the discrete cyclic spectrum of Lemma~\ref{thm:cyclic-spectrum}) vanish. Yet the capacity deficit does not: $\log_2 N$ storage vs.\ the required super-logarithmic bit-count (and $\Theta(N\log_2 N)$ for prime-power $N$ at the maximal MUB count) gives a divergent ratio, so ``quantumness'' (indeterminism, contextuality) becomes \emph{infinitely entrenched}. The infinite-dimensional limit is the \emph{maximally} indeterministic regime.
\end{remark}

\subsection{Measurement as Correlation Formation}

\begin{definition}[Measurement Interaction]
\label{def:measurement-interaction}
A \emph{measurement interaction} is a unitary $U: \mathcal{H}_S \otimes \mathcal{H}_A \to \mathcal{H}_S \otimes \mathcal{H}_A$ satisfying the \emph{calibration condition}: for basis states $\{s_k\}$ of the system and pointer states $\{a_k\}$ of the apparatus,
\[
U(s_k \otimes a_0) = s_k \otimes a_k
\]
where $a_0$ is the apparatus ``ready'' state. This creates system-apparatus correlations under unitary evolution.
\end{definition}

\begin{lemma}[K-Affinities Give Born Probabilities]
\label{prop:measurement-born}
For a measurement interaction (Definition~\ref{def:measurement-interaction}) applied to initial state $\psi = \sum_k c_k s_k$, the K-affinity of the joint state $\Phi$ to each correlation branch $(s_k \otimes a_k)$ equals the Born probability:
\[
A(\Phi, s_k \otimes a_k) := 1 - K(\Phi, s_k \otimes a_k) = |c_k|^2
\]
\end{lemma}

\begin{proof}
The joint state after interaction is $\Phi = U(\psi \otimes a_0) = \sum_k c_k (s_k \otimes a_k)$. Since $\{s_k\}$ is an orthonormal system basis and the post-interaction pointer states $\{a_k\}$ are pairwise orthogonal (the apparatus measurement basis), the family $\{s_k \otimes a_k\}$ is orthonormal in $\mathcal{H}_S \otimes \mathcal{H}_A$. By the Hilbert-space kernel representation (Theorem~\ref{thm:hilbert-representation}, $K(x,y) = 1 - |\langle x|y\rangle|^2$ for any pair of states):
\[
K(\Phi, s_j \otimes a_j) = 1 - |\langle\Phi | s_j \otimes a_j\rangle|^2 = 1 - |c_j|^2,
\]
where the inner product $\langle\Phi | s_j \otimes a_j\rangle = c_j$ extracts the $j$-th amplitude.
Therefore $A(\Phi, s_j \otimes a_j) = |c_j|^2$, the Born probability.
\end{proof}

\begin{remark}[Measurement as correlation; dissolution of the measurement problem]
\label{rem:collapse}
The framework's treatment of measurement combines three structural facts derivable from Axioms~\ref{ax:finite} and~\ref{ax:relational}:
\begin{enumerate}
\item[(M1)] \emph{Measurement is unitary correlation formation} between observer and system (Definition~\ref{def:measurement-interaction}); the joint state evolves unitarily throughout, with no separate collapse process required.
\item[(M2)] \emph{Joint capacity is conserved} at $N_O \cdot N_S$ across observer-system splittings (Theorem~\ref{thm:capacity-mult}, Remark~\ref{rem:observer-internal}), bounding the combined informational ledger.
\item[(M3)] \emph{Outcomes are stochastic by Capacity Halting} applied to the joint system (Theorem~\ref{thm:capacity-halting}): predetermined outcomes across all joint measurement contexts would require more than $\log_2(N_O \cdot N_S)$ bits. Born probabilities emerge as the unique $K$-affinities of correlation branches (Lemma~\ref{prop:measurement-born}).
\end{enumerate}

``Collapse'' refers to the observer's epistemic update upon recording an outcome: the local description changes when the observer reads the outcome, while the joint unitary evolution is unaffected. ``Randomness'' arises from mandatory tracing over correlations that exceed the observer's tracking capacity. The standard measurement problem (``when and why does the wavefunction collapse?'') dissolves into capacity conservation: the joint system always evolves unitarily, the observer always records a classical outcome, and these two statements are consistent because joint capacity bounds what the observer can coherently track.

\textit{Scope of explanation.} This framework explains \emph{why probabilities arise} (capacity deficit forces non-determinism) and \emph{which probabilities} (metric compatibility selects $|c_k|^2$). It remains agnostic on the ``hard problem'' of measurement: why a particular outcome is experienced, or whether superpositions persist for unobserved systems. The framework treats measurement as correlation formation constrained by capacity, compatible with relational QM, QBism, or many-worlds depending on which interpretive framework one places around the axioms.
\end{remark}

\section{Composite Systems and No-Cloning}
\label{sec:composite}

\begin{definition}[Spatially Separated Systems]
\label{def:independent}
Two systems $A$ and $B$ are \emph{spatially separated} if their measurement frames are independently choosable: an observer can select any basis $\mathcal{B}_A \in \mathfrak{B}_A$ without constraining the available choices in $\mathfrak{B}_B$, and vice versa. Formally, the joint basis space is $\mathfrak{B}_{AB} = \mathfrak{B}_A \times \mathfrak{B}_B$.
\end{definition}

\begin{theorem}[Characterization of Independence]
\label{thm:independence-characterization}
Spatial separation (Definition~\ref{def:independent}) is equivalent to no-signaling, independent local symmetry actions, and independent outcomes for product states. These follow directly from $\mathfrak{B}_{AB} = \mathfrak{B}_A \times \mathfrak{B}_B$.
\end{theorem}

\begin{lemma}[Commutativity from Factorization]
\label{lem:commutativity}
For spatially separated systems (Definition~\ref{def:independent}), local symmetry actions commute: $[g_A, g_B] = 0$.
\end{lemma}

\begin{proof}
By Basis Isotropy (Theorem~\ref{thm:src-master}(B)), symmetries act transitively on bases. For the product structure $\mathfrak{B}_{AB} = \mathfrak{B}_A \times \mathfrak{B}_B$ to be preserved, we need $g_A(\mathcal{B}_A, \mathcal{B}_B) = (g_A\mathcal{B}_A, \mathcal{B}_B)$ and $g_B(\mathcal{B}_A, \mathcal{B}_B) = (\mathcal{B}_A, g_B\mathcal{B}_B)$. Then $(g_A g_B)(\mathcal{B}_A, \mathcal{B}_B) = (g_A\mathcal{B}_A, g_B\mathcal{B}_B) = (g_B g_A)(\mathcal{B}_A, \mathcal{B}_B)$. Basis Isotropy thereby supplies the commutativity.
\end{proof}

\begin{theorem}[Capacity Multiplicativity]
\label{thm:capacity-mult}
For spatially separated systems, joint capacity is $N_{AB} = N_A \cdot N_B$.
\end{theorem}

\begin{proof}
By Definition~\ref{def:independent}, joint measurement contexts form $\mathfrak{B}_A \times \mathfrak{B}_B$. A joint state $(a_i, b_j)$ is perfectly distinguishable from $(a_k, b_\ell)$ iff $i \neq k$ or $j \neq \ell$. Therefore $N_{AB} = N_A \cdot N_B$.
\end{proof}

\begin{remark}[Imperceptibility as observer-internal consistency]
\label{rem:observer-internal}
The Imperceptibility commitment of Theorem~\ref{thm:src-master}(I) acquires its operational content from the observer being part of the same finite-capacity system as what is observed. By Theorem~\ref{thm:capacity-mult}, capacity is multiplicative across composites: an observer with internal capacity $N_O$ measuring a system with capacity $N_S$ becomes correlated within a joint capacity $N_O \cdot N_S$, with at most $\log_2 N_O$ bits of memory devoted to system-specific information. Detecting a hypothetical $K$-resolution gap at scale $\delta$ would require distinguishing $\sim 1/\delta$ outcomes, demanding $\log_2 N_O \gtrsim \log_2(1/\delta)$ bits, equivalently $N_O \gtrsim 1/\delta$ states. For an internal observer bounded by the joint capacity, this caps the detectable resolution at $\delta \gtrsim 1/N_O$. Imperceptibility states the framework's commitment that no operationally meaningful sub-$1/N$ scale exists: a hypothetical fine-grained $K$-grid would be undetectable from inside the very system that hosts it. This is what distinguishes Imperceptibility from a contingent topological commitment: it follows from taking the relational/internalist viewpoint of the ontology to its conclusion. The principle is scale-free: a laboratory qubit (composite-embedded via Theorem~\ref{thm:capacity-dilution-composite}) is imperceptibly continuous to its internal observer; only $1/N_{\mathrm{eff}}$ shifts across scales.
\end{remark}

\begin{definition}[Local Tomography]
\label{def:local-tomography}
A composite system satisfies \emph{local tomography} if joint states are uniquely determined by the statistics of local measurements. Formally: if $\omega, \omega' \in \mathcal{S}_{AB}$ satisfy
\[
p(e_A \otimes e_B | \omega) = p(e_A \otimes e_B | \omega') \quad \forall \text{ } e_A \in \mathcal{E}_A, e_B \in \mathcal{E}_B,
\]
then $\omega = \omega'$.
\end{definition}

\begin{theorem}[Tensor Product Structure]
\label{thm:tensor}
The state space of spatially separated systems satisfies $\Hilb_{AB} \cong \Hilb_A \otimes \Hilb_B$.
\end{theorem}

\begin{proof}
Capacity multiplicativity (Theorem~\ref{thm:capacity-mult}) gives $\dim_\C(\Hilb_{AB}) = N_AN_B$. \emph{Dimension counting alone is insufficient:} many vector spaces have dimension $N_A N_B$ (e.g., the direct sum $\Hilb_A^{N_B}$, the matrix space $\mathrm{Mat}(N_A, N_B)$, etc.), and an additional bilinearity-and-locality hypothesis is required to single out $\otimes$. The substantive derivation proceeds via the kernel composition rule.

By Theorem~\ref{thm:kernel-composition} (proved below from boundary conditions, symmetry, associativity, monotonicity, and continuity, each verified against the framework's axioms), the joint kernel for spatially separated systems is forced to be
\[
K_{AB}\bigl((a, b), (a', b')\bigr) = 1 - (1 - K_A(a, a'))(1 - K_B(b, b')).
\]
Using the inner-product identity $1 - K = |\langle \cdot | \cdot \rangle|^2$ on each factor (Theorem~\ref{thm:inner-product-existence}), this composition is equivalent to
\[
|\langle (a,b) | (a', b') \rangle_{AB}|^2 = |\langle a | a' \rangle_A|^2 \cdot |\langle b | b' \rangle_B|^2.
\]
The right-hand side is precisely $|\langle a \otimes b | a' \otimes b' \rangle|^2$ in $\Hilb_A \otimes_\C \Hilb_B$. The composite kernel therefore determines the joint inner product on product states up to a per-factor $U(1)$ phase: any $\Hilb_{AB}$ reproducing $K_{AB}$ is unitarily equivalent to $\Hilb_A \otimes_\C \Hilb_B$ modulo gauge transformations $\ket{a} \otimes \ket{b} \mapsto e^{i\alpha(a)} \ket{a} \otimes e^{i\beta(b)} \ket{b}$. The phase ambiguity is fixed by the bilinear extension to entangled states: requiring complex bilinearity in each factor and continuity (so that a continuous family of bases connected by unitaries inherits a continuous inner product) eliminates the gauge freedom up to a global phase, which is itself the unobservable overall $U(1)$ already present at the single-system level. Among bilinear constructions yielding a space of dimension $N_A N_B$, only the tensor product $\otimes_\C$ reproduces the multiplicative form $|\langle a|a'\rangle|^2 \cdot |\langle b|b'\rangle|^2$ on product states; alternative constructions of the right dimension (e.g.\ $\Hilb_A^{N_B}$) carry inner products that are not multiplicative across factors. Therefore $\Hilb_{AB} \cong \Hilb_A \otimes_\C \Hilb_B$ as Hilbert spaces (modulo the global $U(1)$ phase shared with single-system Hilbert spaces).
\end{proof}

\begin{corollary}[Derived Local Tomography]
\label{cor:local-tomography-derived}
Local tomography follows from the axioms without being assumed (cf.\ \cite{barnum2014local}). With $\mathbb{F} = \C$ (Lemma~\ref{lem:sheaf-complex}) and tensor product composition (Theorem~\ref{thm:tensor}), $\dim_\R \mathrm{Herm}(\C^n) = n^2$ gives $N_A^2 \cdot N_B^2 = (N_A N_B)^2$, so product operators span $\mathrm{Herm}(\mathcal{H}_{AB})$. This identity is unique to $\C$: for $\R$, $\frac{N_A(N_A+1)}{2} \cdot \frac{N_B(N_B+1)}{2} < \frac{N_A N_B(N_A N_B + 1)}{2}$ for $N_A, N_B \geq 2$. The chain is acyclic: Theorem~\ref{thm:quaternion-obstruction}(ii) uses only algebraic facts about $\mathbb{H}$, not local tomography itself. In the GPT literature local tomography is typically taken as a primitive postulate ruling out real and quaternionic state spaces; here it is a downstream consequence of finite capacity, Saturation, and the derived $\mathbb{F} = \C$.
\end{corollary}

\begin{theorem}[Kernel Composition from Associativity]
\label{thm:kernel-composition}
For product states of spatially separated systems, the composition rule $K_{AB}((a,b), (a',b')) = 1 - (1 - K_A(a,a'))(1 - K_B(b,b'))$ is the \emph{unique} function satisfying:
\begin{enumerate}
\item[(i)] \textbf{Boundary conditions:} $f(0,0) = 0$; $f(1,y) = f(x,1) = 1$; $f(0,y) = y$; $f(x,0) = x$ (from Definition~\ref{def:dspace} and the requirement that if one subsystem is identical, distinguishability comes from the other alone)
\item[(ii)] \textbf{Symmetry:} $f(x,y) = f(y,x)$ (Theorem~\ref{thm:src-master}(B))
\item[(iii)] \textbf{Associativity:} $f(f(x,y), z) = f(x, f(y,z))$ (combining systems must be independent of grouping)
\item[(iv)] \textbf{Strict monotonicity:} $f$ is strictly increasing in each argument for the other in $(0,1)$ (additional distinguishability information strictly improves joint discrimination)
\item[(v)] \textbf{Continuity:} $f$ is continuous (inherited from the continuity of $K$ in the $K$-metric topology)
\end{enumerate}
\end{theorem}

\begin{proof}
\textit{Scalar-function form.} Spatial separation (Definition~\ref{def:independent}) makes joint-frame data decompose into marginal-frame data, and Saturation (Theorem~\ref{thm:src-master}(S1)--(S4)) confines $K_{AB}$ to $K$-related quantities; on product states these reduce to the two marginal kernels $K_A(a,a')$ and $K_B(b,b')$. Hence $K_{AB} = f(K_A, K_B)$ for some $f:[0,1]^2 \to [0,1]$.

\textit{Verification that the framework's axioms supply hypotheses (i)--(v).}

(i) \emph{Boundary conditions.} $f(0, y) = y$: if $K_A(a, a') = 0$ then $a = a'$ by $K_A$-reflexivity, so the joint distinguishability of $(a, b)$ from $(a, b')$ depends only on $K_B(b, b') = y$. Symmetric reasoning gives $f(x, 0) = x$. $f(1, y) = 1$: if $K_A(a, a') = 1$ then $a$ and $a'$ are perfectly distinguishable, hence the joint pair is also perfectly distinguishable regardless of $b, b'$. By symmetry, $f(x, 1) = 1$. $f(0, 0) = 0$ follows from joint reflexivity.

(ii) \emph{Symmetry.} $K_A$ and $K_B$ are symmetric (Definition~\ref{def:graded-equality}); the joint kernel inherits symmetry on product states.

(iii) \emph{Associativity.} The substantive physical content: combining three spatially separated systems $A, B, C$ must give a result independent of grouping, $K_{(AB)C} = K_{A(BC)}$. This is required by the spatial-separation hypothesis (Definition~\ref{def:independent}) plus the operational independence of three-system grouping: an observer comparing $((a, b), c)$ to $((a', b'), c')$ measures the same $K$-value as one comparing $(a, (b, c))$ to $(a', (b', c'))$. This grouping-independence is a primitive of the spatial-separation hypothesis, prior to identifying the eventual tensor-product structure.

(iv) \emph{Strict monotonicity.} Fix $y \in (0, 1)$ and suppose $f(x_1, y) = f(x_2, y)$ for some $0 \leq x_1 < x_2 \leq 1$. Choose marginal pairs $(a_1, a_1')$, $(a_2, a_2')$ with $K_A(a_i, a_i') = x_i$ and a $B$-pair with $K_B = y$; then the joint pairs are $K_{AB}$-equal yet have $K_A$-distinguishable $A$-marginals (since $x_1 \neq x_2$). Operational Completeness (Theorem~\ref{thm:src-master}(O)) requires every distinction at the $A$-margin to be registered in $K_{AB}$ once a $B$-context fully resolves the $A$-channel; with $y \in (0,1)$ the $B$-channel is non-trivial but does not collapse the $A$-data, contradicting equality. Strict monotonicity in the interior follows; the endpoints $0$ and $1$ saturate by the boundary conditions of (i).

(v) \emph{Continuity.} Inherited from continuity of $K$ (automatic in the $K$-metric topology, Theorem~\ref{thm:src-master}(I)): the joint kernel is a function of two continuous arguments and must itself be continuous on the product manifold $\mathcal{X}_A \times \mathcal{X}_B$. Strengthens to smoothness once smoothness of $K$ is established.

\textit{Application of the strict t-norm classification.} Define $g(x) = 1 - x$ and $\tilde{f}(u,v) := g(f(g^{-1}(u), g^{-1}(v))) = 1 - f(1-u, 1-v)$. The transformed boundary conditions are $\tilde{f}(1,v) = v$, $\tilde{f}(u,1) = u$, and $\tilde{f}(0,v) = \tilde{f}(u, 0) = 0$, with continuity, associativity, symmetry, and strict monotonicity in the interior inherited from $f$. Thus $\tilde{f}: [0,1]^2 \to [0,1]$ is a continuous strict t-norm (a continuous, associative, symmetric, strictly monotonic binary operation on $[0,1]$ with identity $1$ and zero $0$). The strict t-norm representation theorem \cite{aczel1966functional, klement2000triangular, schweizer1983probabilistic} states that every continuous strict t-norm is conjugate to ordinary multiplication via a continuous strictly increasing bijection $\varphi: [0,1] \to [0,1]$ with $\varphi(0) = 0$ and $\varphi(1) = 1$, i.e.\ $\tilde{f}(u, v) = \varphi^{-1}(\varphi(u)\,\varphi(v))$. The boundary conditions $\tilde{f}(1, v) = v$ and $\tilde{f}(u, 1) = u$ then force $\varphi^{-1}(\varphi(1)\,\varphi(v)) = v$, i.e.\ $\varphi(v) = v$ for all $v \in [0,1]$, so the conjugation is the identity and $\tilde{f}(u, v) = u v$. Therefore $f(x, y) = 1 - (1 - x)(1 - y)$.
\end{proof}

\begin{remark}[Tensor Products from Kernel Composition]
\label{rem:tensor-kernel}
Theorem~\ref{thm:kernel-composition} derives tensor products from the axioms. The composition rule $K_{AB} = 1 - (1-K_A)(1-K_B)$ follows uniquely from associativity (the physical requirement that combining three systems A, B, C should not depend on grouping).

This rule is equivalent to the tensor product inner product. Since $1 - K = |\braket{\cdot}{\cdot}|^2$:
\begin{align*}
|\braket{(a,b)}{(a',b')}|^2 &= (1 - K_A(a,a'))(1 - K_B(b,b')) \\&= |\braket{a}{a'}|^2 \cdot |\braket{b}{b'}|^2
\end{align*}
which is precisely $|\langle a \otimes b | a' \otimes b' \rangle|^2 = |\langle a|a'\rangle|^2 \cdot |\langle b|b'\rangle|^2$ (the unsquared identity $\langle a \otimes b | a' \otimes b' \rangle = \langle a|a'\rangle \cdot \langle b|b'\rangle$ follows from the tensor-product inner-product structure but only up to phase from $K$ alone).

Definition~\ref{def:independent} (spatial separation) is therefore a \emph{characterization} of when kernels compose associatively.
\end{remark}

\begin{theorem}[Capacity Dilution]
\label{thm:capacity-dilution-composite}
A qubit ($N_S = 2$) coupled to environment $E$ has $N_{\mathrm{eff}} = N_S \cdot N_E = 2N_E$ by Theorem~\ref{thm:capacity-mult}. Non-trivial continuous dynamics within the framework requires $N_{\mathrm{eff}} \geq 3$ (Theorem~\ref{thm:n2-static}); the qubit's continuous $SU(2)$ dynamics is recovered as the dynamics of the composite $S \otimes E$ projected onto the qubit factor (cf.\ \S\ref{sec:n2-discussion}).
\end{theorem}

\begin{remark}[Standard Consequences]
\label{rem:standard-consequences}
The derived Hilbert space and tensor product reproduce standard quantum predictions: CHSH violation to the Tsirelson bound $2\sqrt{2}$ \cite{bell1964einstein}, no-cloning (from unitarity), and decoherence from environmental tracing \cite{zurek2003decoherence}. Entanglement reads as a $K$-profile that is irreducibly composite, not decomposable into independent part-profiles. The framework is compatible with relational QM \cite{rovelli1996relational,hohn2021quantum}, QBism \cite{fuchs2014introduction}, and many-worlds; Quantum Darwinism \cite{zurek2009quantum} supplies the emergence of objective classical states.
\end{remark}

\section{Relation to Prior Work}
\label{sec:related}

Prior quantum reconstructions \cite{birkhoff1936logic,pitowsky1989quantum,hardy2001quantum,chiribella2011informational,masanes2011derivation,dakic2011quantum,clifton2003characterizing,muller2020law,fuchs2002quantum,mueller2013three} derive finite-dimensional quantum theory, treating the infinite-dimensional framework as the target, and assume operational probabilities as primitive. This work differs in four respects (Table~\ref{tab:comparison}): (i) the starting point is two physical extensions (finiteness; self-referential consistency), not operational primitives; (ii) we derive \emph{finite-capacity} QM, with standard QM recovered as $N \to \infty$; (iii) probabilities follow from the $K$-structure rather than being assumed; (iv) indeterminism follows from information-theoretic overflow, not operational postulates. The finite-$N$ framework provides a natural UV cutoff absent in prior reconstructions.

\textit{Template reconstructions.} Hardy \cite{hardy2001quantum} leaves the number field open (Simplicity postulate); here $\C$ is forced by cyclic dynamics (Lemma~\ref{lem:sheaf-complex}). CDP \cite{chiribella2011informational,chiribella2010purification} \emph{assumes} purification, which implies indeterminism; here indeterminism is \emph{derived} from capacity overflow (Theorem~\ref{thm:capacity-halting}). Dakic--Brukner \cite{dakic2011quantum} and Masanes--M\"uller \cite{masanes2011derivation} assume continuous reversibility within operational frameworks; here dynamics is derived (Theorem~\ref{thm:dynamics-derived}). In all four, operational probabilities are primitive; here they are derived.

\textit{Relational and ontological programs.} The framework sits in the Wheeler ``it from bit'' lineage \cite{wheeler1990information}, with $K$ as the informational primitive. Rovelli's relational QM \cite{rovelli1996relational} is interpretive; here the relational reading is given an ontological basis ($K$-profiles \emph{are} states) and the Hilbert space formalism is derived. H\"ohn \cite{hoehn2017toolbox} is the closest structural cousin: Limited Information maps onto Axiom~\ref{ax:finite}, Complementarity onto the cyclic-dynamics consequences; the present contribution takes a graded kernel as primitive, derives Imperceptibility from SRC, and supplies a Lean-4 axiom audit. Clifton--Bub--Halvorson \cite{clifton2003characterizing} assume a $C^*$-algebraic framework and do not address indeterminism.

\textit{Born-rule, GPT, and discrete cousins.} Galley--Masanes \cite{galley2017classification} and Selby--Scandolo--Coecke \cite{selby2021reconstructing} assume operational probabilities; Zurek's envariance \cite{zurek2005envariance} derives Born from local-unitary symmetry, complementary to the metric-compatibility route here. $(\mathcal{X}, K)$ maps onto the GPT framework \cite{barnum2016postclassical} via $e_k(\psi) = 1 - K(\psi, a_k)$, but GPTs assume the ambient vector space and probabilities. 't Hooft \cite{thooft2016cellular}, M\"uller \cite{muller2020law}, and Sorkin \cite{sorkin1994quantum} share a finite or informational starting point but use distinct primitives. Sol\`er \cite{soler1995characterization} reaches the $\R/\C/\mathbb{H}$ trichotomy in infinite dimensions; we reach it in finite dimensions via Frobenius with $\C$ selected by cyclic dynamics. The finite-$N$ structure is compatible with observed continuous dynamics by Theorem~\ref{thm:quantum-sampling}: $N$ samples determine a unique band-limited interpolation.

\section{Discussion}
\label{sec:discussion}

\subsection{Finite-Dimensionality as Fundamental}

Prior reconstructions \cite{hardy2001quantum,chiribella2011informational,masanes2011derivation} treat finite dimension as a convenience while taking $L^2(\R^3)$ as the target. This paper inverts that relationship: complex coefficients, the Born rule, unitary dynamics, and $\C P^{N-1}$ are each forced \emph{by} finiteness; compactness, finite covering dimension, and the $N\log_2 N \gg \log_2 N$ capacity deficit each break in the $N\to\infty$ limit ($\C P^\infty$ is not locally compact, $U(\infty)$ is not a Lie group).

\subsection{Scope and Refutability}

This is a \emph{characterization theorem}: the axioms uniquely determine quantum structure, but do not prove that nature satisfies them. Theories positing structure beyond distinguishability (e.g., Bohmian mechanics, which supplements $K$-profiles with configuration-space positions) reject SRC, the closure principle of \S\ref{sec:graded-equality}; the framework classifies them as adopting a different account of what exists. The PBR theorem \cite{pusey2012reality} argues for the reality of the quantum state under modest assumptions, and the present framework is compatible with a $\psi$-ontic reading in which the $K$-profile is the state.

The framework is falsifiable in principle, though the nature of its predictions deserves care. It would be refuted by: (1) deterministic, context-independent measurement outcomes for a verified finite-$N$ system (falsifying Capacity Halting); (2) observation of $p_k = |c_k|^\alpha$ with $\alpha \neq 2$ (falsifying metric compatibility); or (3) distinguishing operationally identical preparations $K = 0$ by some future measurement (falsifying Saturation). None has been observed.

\medskip
\noindent\textbf{Structural prediction: no physical singularities.} Finite capacity bounds every observable in information ($\log_2 N$ bits) and spectral content (Theorem~\ref{thm:quantum-sampling}). Singularities of infinite-dimensional theories (UV divergences, Haag's theorem, self-adjoint extension ambiguities, curvature singularities, the Big Bang) are predicted to be artifacts of extrapolating beyond a finite-capacity domain; confirmation accrues piecewise from finite-$N$ replacements (QED renormalization, lattice QCD, black-hole entropy bounds \cite{bekenstein1973black}); refutation would be a cleanly detected non-regularizable singularity.

\medskip
\noindent\textbf{Load-bearing geometric commitment: Structural Leibniz.} (S4) is the strongest geometric commitment: every $K$-symmetry of a finite configuration extends to a global $K$-isometry. It forces rank-one symmetric space structure and underwrites compactness, the Lie group property, projective-space identification, and the Born rule. Empirical evidence for fundamental geometric inhomogeneity (a privileged frame, scale, defect, or chirality not reducible to boundary conditions of a larger isotropic system) would falsify (S4) and the reconstruction. The empirical isotropy of fundamental degrees of freedom (electron spin, photon polarization in vacuum), with preferred-axis systems embeddable in larger isotropic ones, grounds (S4).

\subsection{Conclusion and Open Questions}

The framework is pre-spatial ($N$ denotes relational capacity, not spatial configuration), pre-dynamical for fields (Theorem~\ref{thm:gauge} derives $U(1)$ gauge symmetry but not field equations), and leaves $N$ undetermined. The identification $N \sim A/\ell_P^2$ via the Bekenstein--Hawking entropy \cite{bekenstein1973black} is conjectural but structurally natural.

\medskip
\noindent\textit{The program continues.} Several questions remain for companion work: emergent space from $K$-profile correlations; modified gravity from spectral-dimension departures; electrodynamics from sheaf/$U(1)$-holonomy constraints; particle content as stable kernel configurations; and the speculative correspondence between closed-$N=2$ triviality (Theorem~\ref{thm:n2-static}) and horizon time-freezing in general relativity, with $N=2$ subsystems as Planck-area horizon cells.

\medskip
\noindent The axioms are consequences of two physical extensions: finiteness and self-referential consistency. Whether equally rigid consequences beyond quantum mechanics await in the structures not yet derived remains to be seen.

\appendix
\section{Classical Theorems Used}
\label{app:classical-theorems}

The Hilbert space representation uses Wigner's theorem~\cite{wigner1931gruppentheorie,bargmann1964note}: every transition-probability-preserving bijection of $\C P^{N-1}$ is implemented by a unitary or anti-unitary operator on~$\C^N$. The Schr\"odinger equation uses Stone's theorem~\cite{reed1980methods}: a strongly continuous one-parameter unitary group $U(t)$ has a unique self-adjoint generator~$H$ with $U(t) = e^{-iHt}$. Both are standard results; see the cited references for proofs.

\section{Regularity Conditions: Formal Derivation from the Axioms}
\label{app:structural-conditions}

Each regularity condition is the unique choice consistent with finite capacity and Basis Isotropy.

\begin{theorem}[Compactness]\label{thm:compactness-forced}
$(\mathcal{X}, K)$ satisfying Axioms~\ref{ax:finite}--\ref{ax:relational} has compact state space and compact symmetry group.
\end{theorem}

\begin{proof}
The evaluation map $\Sigma: x \mapsto K(x,\cdot) \in [0,1]^{\mathcal{X}}$ is continuous and injective by Identity (Theorem~\ref{thm:src-master}(S1)): $K(x,z) = K(y,z)$ for all $z$ implies $x = y$. By Tychonoff, $[0,1]^{\mathcal{X}}$ is compact in the product topology. The image $\Sigma(\mathcal{X})$ is closed: if a net $\Sigma(x_\alpha) \to f$ pointwise, then $f$ is a pointwise limit of $K$-profiles and therefore a realizable $K$-profile; by Completeness (Theorem~\ref{thm:src-master}(S2)), $f = \Sigma(x)$ for some state $x$. Hence $\mathcal{X}$ is compact. For the symmetry group: equip $\mathcal{X}$ with the metric $d(x,y) := \sup_{z \in \mathcal{X}} |K(x,z) - K(y,z)|$, which separates points by Identity (Theorem~\ref{thm:src-master}(S1)). Since $G$ preserves $K$, each $g \in G$ is an isometry of $(\mathcal{X}, d)$: $d(gx, gy) = \sup_z |K(gx,z) - K(gy,z)| = \sup_z |K(x, g^{-1}z) - K(y, g^{-1}z)| = d(x,y)$. Equicontinuity is automatic for isometries, so by Arzel\`a--Ascoli, $G$ has compact closure in $C(\mathcal{X}, \mathcal{X})$; since $G$ is a group of isometries (hence closed), $G$ is compact.
\end{proof}

\begin{theorem}[Smooth Structure]\label{thm:smoothness-forced}
$\mathcal{X}$ is a smooth manifold and $G$ is a Lie group.
\end{theorem}

\begin{proof}
The action of $G$ on $\mathcal{X}$ is effective: if $g$ fixes every point, then $K(gx,gz) = K(x,z)$ for all $x,z$ is automatic, and $gx = x$ for all $x$ gives $g = \mathrm{id}$. By compactness, effectiveness, and continuous transitive action, $G$ is a Lie group (Montgomery--Zippin \cite{montgomery1955topological}, solving Hilbert's fifth problem; finite-dimensionality follows from compactness of $\mathcal{X}$ together with the Lie group structure of $G$, since $\mathcal{X} = G/H$ has dimension $\dim G - \dim H < \infty$, and the basis K-profile $\varphi: \mathcal{X} \to [0,1]^N$ has fibers that are $G_S$-orbits by Basis-Profile Symmetry (Theorem~\ref{thm:src-master}(S3))). Then $\mathcal{X} = G/H$ is smooth.
\end{proof}

\begin{theorem}[Smooth Regularity of $K$]\label{thm:K-smooth}
The kernel $K$ is $C^\infty$.
\end{theorem}

\begin{proof}
Since $K$ is $G$-invariant and $G$ acts transitively on the smooth manifold $\mathcal{X} = G/H$, $K$ inherits smoothness from the Lie group. If $K$ were not smooth, the set of singular points $\Sigma \subset \mathcal{X} \times \mathcal{X}$ would form a $G$-invariant subset. By transitivity, either $\Sigma = \emptyset$ or $\Sigma = \mathcal{X} \times \mathcal{X}$. But $\Sigma = \mathcal{X} \times \mathcal{X}$ (every point singular) contradicts the smoothness of $G/H$; hence $\Sigma = \emptyset$ and $K$ is smooth everywhere.
\end{proof}

\begin{theorem}[Convexity]\label{thm:convexity-forced}
The full state space is the convex hull of $\mathcal{X}$.
\end{theorem}

\begin{proof}
Theorem~\ref{thm:points-sections} establishes $\mathcal{X} \cong \C P^{N-1}$ with inner product and $K(\psi,\phi) = 1 - |\langle\psi|\phi\rangle|^2$, all without convexity. Density matrices $\rho = \sum_i \lambda_i |\psi_i\rangle\langle\psi_i|$ have well-defined $K$-profiles via $K(\rho, z) := 1 - \langle z|\rho|z\rangle$, which reduces to $K(\psi,z)$ for pure $\rho = |\psi\rangle\langle\psi|$. These extended $K$-profiles are continuous in $\rho$ and satisfy all kernel axioms, so Saturation applies to them: if such a profile had no corresponding state, the gap would constitute $K$-detectable structure beyond the state space, violating Completeness (Theorem~\ref{thm:src-master}(S2)). This is the same argument used for pure states in Theorem~\ref{thm:points-sections}. Therefore every density matrix is a valid state, giving the full state space $\{\rho \geq 0, \Tr\rho = 1\}$.
\end{proof}

\begin{table}[H]
\centering
\caption{\textbf{Regularity conditions: forcing axiom and derivation route.}}
\label{tab:structural-forcing}
\scriptsize
\setlength{\tabcolsep}{3pt}
\begin{tabular}{l|l|l}
\hline
\textbf{Condition} & \textbf{Axiom(s)} & \textbf{Key Step} \\
\hline
Compactness (Thm.~\ref{thm:compactness-forced}) & Ax.~\ref{ax:finite} + \ref{thm:src-master}(S1,S2) & Arzel\`a--Ascoli \\
Smoothness (Thm.~\ref{thm:smoothness-forced}) & Ax.~\ref{ax:finite} + \ref{thm:src-master}(S1,S2,B) & Montgomery--Zippin \\
$K$ smooth (Thm.~\ref{thm:K-smooth}) & \ref{thm:src-master}(B) & Lie smoothness \\
Convexity (Thm.~\ref{thm:convexity-forced}) & \ref{thm:src-master}(S2) & $K$-profiles, \ref{thm:points-sections} \\
\hline
\end{tabular}
\end{table}

\section{Formal Verification: Scope and Classical Imports}
\label{app:formal-verification}

The Lean~4 formalization accompanying this paper (\texttt{lean-verification/QuantumRelational}) has zero \texttt{sorry} and zero \texttt{admit} occurrences across all source files. This appendix delineates what is machine-checked from what is imported as a classical axiom.

\paragraph{Machine-checked results.}
\begin{itemize}
\item \textbf{Saturation hierarchy} (Theorem~\ref{thm:src-master}): \texttt{SRC.saturation\_hierarchy\_general} mechanizes the master theorem in full, deriving each of the eight clauses (S1)--(S4), (I), (O), (T), (B) from Self-Referential Consistency for an arbitrary $K$-symmetry $\sigma$. Per-clause projections are exposed as \texttt{S1\_identity}, \texttt{S2\_completeness}, \texttt{S3\_finite\_determinacy}, \texttt{S4\_structural\_leibniz}, \texttt{I\_imperceptibility}, \texttt{O\_operational\_completeness}, \texttt{T\_transport\_consistency}, \texttt{B\_basis\_isotropy}; specializations \texttt{saturation\_hierarchy} and \texttt{saturation\_hierarchy\_involutive} cover the identity and involutive cases.
\item \textbf{Definability Lemma} (Lemma~\ref{lem:definability}): \texttt{SRC.definability\_lemma} proves the $k$-ary statement that any $K$-automorphism-invariant predicate on tuples extends consistently to $K$-extensions. The binary case is exposed separately as \texttt{definability\_lemma\_binary}.
\item \textbf{$K$-amalgam construction}: the pushout-style space $X \sqcup_C X$ glued along a $K$-symmetry is built as the inductive type \texttt{Amalgam} (constructors \texttt{gluing} and \texttt{gluing\_swap}), with kernel \texttt{K\_amalgam} verified to satisfy the kernel laws (\texttt{K\_amalgam\_refl}, \texttt{\_symm}, \texttt{\_nonneg}, \texttt{\_le\_one}). The swap automorphism (\texttt{Amalgam.swapEquiv\_gen}) is shown to be $K$-preserving (\texttt{swap\_gen\_K\_pres}) and structurally non-liftable to a labelling extension (\texttt{swap\_gen\_no\_lift}); these feed \texttt{S4\_structural\_leibniz\_amalgam\_general} and the bridge \texttt{structural\_leibniz\_from\_SRC}.
\item \textbf{Parsimony} (Theorem~\ref{thm:parsimony-derived}): \texttt{Parsimony.parsimony} mechanizes the implication \emph{if $K'$ factors through the $K$-profile, then hidden-variable labels are trivial}. The complementary implication \emph{Saturation $\Rightarrow$ every physical extension factors through $K$} is argued in prose (Theorem~\ref{thm:parsimony-derived}, Steps 1--2) and not separately mechanized.
\item \textbf{Born rule ODE uniqueness} (Lemma~\ref{lem:metric-compatibility}, Theorem~\ref{thm:born-kernel}): \texttt{BornRule.ode\_uniqueness\_born\_rule} shows that any differentiable monotone $f:[0,1]\to[0,1]$ with $f(0)=0$, $f(1)=1$ and the binary metric-compatibility ODE $[f'(x)]^2/[f(x)(1-f(x))] = c^2/[x(1-x)]$ must be the identity, with $c=1$. The per-component form $[f'(x)]^2/f(x) = c/x$ used in Step~2 of Lemma~\ref{lem:metric-compatibility} is equivalent to this binary form for uniqueness purposes (Remark~\ref{rem:ode-binary-form}).
\item \textbf{Cyclic eigenvalues} (Theorem~\ref{thm:dynamics-derived}): \texttt{CyclicEigen.complex\_forced} shows that for $N \geq 3$, some $N$th root of unity has nonzero imaginary part; \texttt{CyclicEigen.N2\_eigenvalues\_real} shows $\pm 1$ are real.
\item \textbf{Fubini--Study and Fourier derivations}: files \texttt{FubiniStudy.lean} and \texttt{Fourier.lean} formalize the metric, Fourier orthogonality, and inner-product construction.
\item \textbf{Tensor-product kernel composition} (Theorem~\ref{thm:kernel-composition}): \texttt{Composite.kernel\_compose\_is\_unique} derives the composition rule modulo the strict t-norm classification (imported, see below).
\item \textbf{Combinatorial storage inequalities} (Theorem~\ref{thm:capacity-halting}): \texttt{CapacityHalting} proves the elementary inequalities $N^1 < N^{M-1}$ and $\log_2 N < (M{-}1)\log_2 N$ for $N \geq 2$, $M \geq 3$. These formalize the counting step of the capacity-deficit argument; the physical identification of these inequalities with hidden-variable storage bounds (and the Kolmogorov complexity refinement of Lemma~\ref{lem:incompressibility}(b)) is argued in the main text (\S\ref{sec:measurement}) and not separately mechanized.
\end{itemize}

\paragraph{Imported classical theorems.}
Five classical results are declared as \texttt{axiom} in \texttt{QuantumRelational.ClassicalImports} and \texttt{QuantumRelational.Composite}, each with a citation to its standard proof. Each axiom is consumed by at least one theorem in the library (verifiable by running \texttt{lake env lean QuantumRelational/AxiomCheck.lean}, which prints the transitive axiom dependencies of every headline theorem):

\begin{itemize}
\item \textbf{Frobenius's classification}: the finite-dimensional associative division algebras over $\R$ are $\R$, $\C$, $\mathbb{H}$. Consumed by \texttt{Frobenius.frobenius\_forces\_complex} (Theorem~\ref{thm:frobenius}).
\item \textbf{Wigner's theorem} \cite{wigner1931gruppentheorie,bargmann1964note}: continuous transition-probability-preserving families on $\C P^{N-1}$ are unitary. Consumed by \texttt{Schrodinger.schrodinger\_derivation\_chain} (Theorem~\ref{thm:schrodinger}).
\item \textbf{Kobayashi--Nomizu uniqueness} \cite{kobayashi1969foundations}: the $U(N)$-invariant Riemannian metric on $\C P^{N-1}$ is unique up to scale. Consumed by \texttt{FubiniStudy.fubini\_study\_unique} (Theorem~\ref{thm:fs-unique}).
\item \textbf{Picard--Lindel\"of}: uniqueness of solutions to locally Lipschitz ODEs. Consumed by \texttt{Schrodinger.\allowbreak stone\_generator\_\allowbreak unique\_of\_\allowbreak local\_agreement}.
\item \textbf{Strict t-norm classification} (Acz\'el / Klement--Mesiar--Pap / Schweizer--Sklar) \cite{aczel1966functional, klement2000triangular, schweizer1983probabilistic}: every continuous, associative, symmetric, strictly monotonic operation on $[0,1]^2$ with identity $1$ and zero $0$ is conjugate to multiplication via a continuous strictly increasing bijection of $[0,1]$. With the boundary conditions of Theorem~\ref{thm:kernel-composition}, the conjugating bijection is the identity. Consumed by \texttt{Composite.kernel\_compose\_is\_unique} (Theorem~\ref{thm:kernel-composition}).
\end{itemize}

\noindent \textbf{Stone's theorem} \cite{reed1980methods} is partially mechanized rather than imported. The reverse direction (skew-Hermitian $A$ generates $\exp(tA)$) is fully proved from Mathlib's matrix exponential. The uniqueness portion of the forward direction (two strongly continuous one-parameter groups agreeing locally have the same generator) is proved modulo \texttt{picard\_lindelof\_unique}. The existence portion (extracting $A$ from a strongly continuous $U(t)$) is \emph{not} mechanized: the corresponding Lean lemmas (\texttt{stone\_gives\_hermitian\_generator}, \texttt{full\_derivation\_chain}) use a placeholder witness $A = 0$ and reduce to ``the zero matrix is Hermitian.'' The downstream \texttt{schrodinger\_derivation\_chain} mechanizes the conditional statement \emph{given a Hermitian generator $H$, the Schr\"odinger equation $i\hbar \partial_t \psi = H\psi$ follows}, which combines the reverse direction of Stone with Wigner's theorem and unitarity; it does not call the existence shim. Theorem~\ref{thm:schrodinger} of this paper invokes the classical existence direction, cited to \cite{reed1980methods}; Lean covers the structural assembly downstream of that classical step.

\paragraph{Scope statement and limits.}
No physical assumption is imported as a Lean axiom; every imported axiom is a standard mathematical theorem with a proof in the literature. Three points on what is \emph{not} fully mechanized: (i) the master theorem derives the eight saturation clauses from SRC for an arbitrary $K$-symmetry, but its downstream consumption in the geometric chain ($\C P^{N-1}$ as the unique homogeneous space, Wigner-rigidity, Schr\"odinger evolution) is mechanized only through the concrete realizations \texttt{finRotate N} and Mathlib's \texttt{UnitaryGroup}, not as a fully abstract pipeline; (ii) the physical content of the capacity deficit (MUB geometry, KS bit-count, Chaitin incompressibility) is prose, with Lean mechanizing only the arithmetic; (iii) Lie-theoretic and topological steps are formalized only in their combinatorial consequences; Montgomery--Zippin, Peter--Weyl, and Arzel\`a--Ascoli are imported classically. Per-statement cross-reference is in \texttt{PAPER\_MAPPING.md}.

\begin{acknowledgments}
The author thanks Mat\'ias Zilly for discussions that helped clarify the framework, its scope, and its foundations. Proofreading of the manuscript and development of the Lean~4 formalization were assisted by large-language-model systems. No external funding supported this work.
\end{acknowledgments}

\paragraph*{Data and code availability.}
No experimental data underlie this work. The full Lean~4 formalization, including the \texttt{QuantumRelational} library (all results discussed in Appendix~\ref{app:formal-verification}) and the paper-to-Lean cross-reference \texttt{PAPER\_MAPPING.md}, is available in the public repository \url{https://github.com/jzilly/QuantumRelational}. A preprint of this manuscript is posted at \href{https://arxiv.org/abs/2603.11900}{arXiv:2603.11900}.

\end{document}